\documentclass{amsart}

\usepackage{etoolbox}

\allowdisplaybreaks

\usepackage{caption}
\usepackage[leftcaption]{sidecap}

\usepackage{enumitem}
\setlist[enumerate,1]{label=\bf (\roman*)}

\usepackage[
  backend=biber,
  style=alphabetic,
  giveninits=true, 
  backref=true
]{biblatex}
\usepackage{lmodern}
\usepackage{anyfontsize}
    
\usepackage{multirow}
\usepackage{hhline}

\usepackage{tabularray}
\UseTblrLibrary{booktabs}

\usepackage{adjustbox}
\usepackage{tcolorbox}
\newtcolorbox{standout}{
  colback=gray!15,
  boxrule=0pt,
  left=.3cm,
  right=.3cm,
  top=.18cm,
  bottom=.18cm,
  boxsep=0pt
}

\usepackage{mathtools}
\usepackage{amssymb}
\usepackage{amsthm}
\usepackage{tensor}
\usepackage{xfrac}
\usepackage{stmaryrd} 
\usepackage[safe]{tipa} 
\usepackage{dsfont} 
\usepackage{relsize} 

\DeclareFontFamily{T3}{lmr}{}
\DeclareFontShape{T3}{lmr}{m}{n}{<-> ssub * cmr/m/n}{}

\theoremstyle{plain}
\newtheorem{theorem}{Theorem}[section]
\newtheorem{lemma}[theorem]{Lemma}
\newtheorem{proposition}[theorem]{Proposition}
\newtheorem{corollary}[theorem]{Corollary}
\theoremstyle{definition}
\newtheorem{definition}[theorem]{Definition}
\newtheorem{example}[theorem]{Example}

\theoremstyle{remark}
\newtheorem{remark}[theorem]{Remark}
\usepackage[
  colorlinks=true,
  linkcolor=darkgreen,
  citecolor=darkgreen,
  urlcolor=darkgreen
]{hyperref}
\usepackage{xurl} 

\usepackage{cleveref}
\crefname{equation}{}{}
\crefname{section}{\S}{\S\S}
\crefname{subsection}{\S}{\S\S}
\crefname{subsubsection}{\S}{\S\S}
\crefname{definition}{Def.}{Defs.}
\crefname{theorem}{Thm.}{Thms.}
\crefname{corollary}{Cor.}{Cors.}
\crefname{lemma}{Lem.}{Lems.}
\crefname{proposition}{Prop.}{Props.}
\crefname{remark}{Rem.}{Rems.}
\crefname{notation}{Ntn.}{Ntns.}
\crefname{fact}{Fact}{Fact}
\crefname{example}{Ex.}{Exs.}
\crefname{figure}{Fig.}{Figs.}
\crefname{table}{Tab.}{Tabs.}
\crefname{footnote}{ftn.}{ftns.}
\Crefname{footnote}{Ftn.}{Ftns.}

\usepackage{tikz}
\usetikzlibrary{
  cd,
  calc,
  arrows.meta,
  backgrounds,
  decorations,
  decorations.pathmorphing, 
}

\tikzcdset{
  arrow style=tikz, 
  diagrams={
    trim left=
    {([xshift=5pt]current bounding box.west)},
    trim right=
    {([xshift=-5pt]current bounding box.east)},
    baseline=
    {([yshift=-3pt]current bounding box.center)},
    >={
      Computer Modern Rightarrow[
        length=4pt, width=4pt
      ]
    },
    hook/.style={
      {Hooks[right, length=2pt, width=8pt]}->
    },
    hook'/.style={
      {Hooks[left, length=2pt, width=8pt]}->
    },
    shorten=0pt,
  }
}

\definecolor{darkblue}{rgb}{0.05,0.25,0.65}
\definecolor{darkgreen}{RGB}{20,140,10}
\definecolor{lightgray}{rgb}{0.9,0.9,0.9}
\definecolor{darkorange}{RGB}{200,100,5}
\definecolor{darkyellow}{rgb}{.91,.91,0}
\definecolor{lightolive}{RGB}{225, 220, 185}

\let\originalsslash\sslash
\renewcommand{\sslash}{\mathord{\originalsslash}}

\DeclareRobustCommand{\rchi}{{\mathpalette\irchi\relax}}
\newcommand{\irchi}[2]{\raisebox{\depth}{$#1\chi$}} 

\newcommand{\BigDelta}{\mathop{\mathchoice
  {\raisebox{-0.5ex}{\scalebox{1.7}{$\Delta$}}} 
  {\raisebox{-0.3ex}{\scalebox{1.3}{$\Delta$}}} 
  {\raisebox{-0.5ex}{\scalebox{1.6}{$\Delta$}}} 
  {\raisebox{-0.1ex}{\scalebox{1.1}{$\Delta$}}} 
}}

\renewcommand{\setminus}{-}

\tikzset{
  snake left/.style={
    rounded corners,
    to path={
      let \p1 = (\tikztostart.east),
          \p2 = (\tikztotarget.west),
          \p3 = ($(\p1)!0.5!(\p2)$),
          \n1 = {8pt} 
      in
      (\p1)
      -- (\x1 + \n1, \y1)
      -- (\x1 + \n1, \y3)
      -- (\x2 - \n1, \y3) \tikztonodes
      -- (\x2 - \n1, \y2)
      -- (\p2)
    }
  }
}

\tikzset{
  uphordown/.style={
    rounded corners,
    to path={
      let \p1 = (\tikztostart.north),
          \p2 = (\tikztotarget.north),
          \n1 = {max(\y1,\y2) + 8pt}
      in
      (\p1)
      -- (\x1, \n1)
      -- (\x2, \n1) \tikztonodes 
      -- (\p2)
    }
  }
}

\tikzset{
  downhorup/.style={
    rounded corners,
    to path={
      let \p1 = (\tikztostart.south),
          \p2 = (\tikztotarget.south),
          \n1 = {min(\y1,\y2) - 8pt}
      in
      (\p1)
      -- (\x1, \n1)
      -- (\x2, \n1) \tikztonodes 
      -- (\p2)
    }
  }
}

\tikzset{
  rightvertleft/.style={
    rounded corners,
    to path={
      let \p1 = (\tikztostart.east),
          \p2 = (\tikztotarget.east),
          \n1 = {max(\x1,\x2) + 8pt}
      in
      (\p1)
      -- (\n1, \y1)
      -- (\n1, \y2) \tikztonodes 
      -- (\p2)
    }
  }
}

\tikzset{
  leftvertright/.style={
    rounded corners,
    to path={
      let \p1 = (\tikztostart.west),
          \p2 = (\tikztotarget.west),
          \n1 = {min(\x1,\x2) - 8pt}
      in
      (\p1)
      -- (\n1, \y1)
      -- (\n1, \y2) \tikztonodes 
      -- (\p2)
    }
  }
}

\newcommand{\shape}{%
  \hspace{.7pt}%
  \raisebox{0.8pt}{\rm\normalfont\textesh}%
  \hspace{1pt}%
}

\newcommand{\defneq}{\equiv}

\newcommand{\hotype}[1]{\mathcal{#1}}

\newcommand{\plus}{\hspace{.8pt}{\adjustbox{scale={.5}{.77}}{$\sqcup$} \{\infty\}}}

\makeatletter
\newcommand{\cpt}{\mathpalette\cpt@inner\relax}
\newcommand{\cpt@inner}[2]{%
  \scalebox{0.5}[0.9]{$#1\cup$}
  #1\{\infty\}
}
\makeatother

\newcommand{\cptIndex}[1]{\hspace{.8pt}{\adjustbox{scale={.5}{.77}}{$\cup$} \{\infty_{\hspace{-.5pt}{}_{#1}}\hspace{-1.9pt}\}}}

\newcommand{\grayunderbrace}[2]{\mathcolor{gray}{\underbrace{\mathcolor{black}{#1}}}_{\mathcolor{gray}{#2}}}

\newcommand{\grayoverbrace}[2]{\mathcolor{gray}{\overbrace{\mathcolor{black}{#1}}}^{\mathcolor{gray}{#2}}}

\newcommand{\HilbertSpace}{%
  \mathcal{H}%
}

\newcommand{\TensorUnit}{\mathds{1}}

\begin{document}

\title[Topological Quantization of Higher Gauge Fields]
{Complete Topological Quantization 
\\ of Higher Gauge Fields}

\thanks{\emph{Funding} by Tamkeen UAE under the 
NYU Abu Dhabi Research Institute grant {\tt CG008}.}

\author{Hisham Sati}
\address{Mathematics Program and Center for Quantum and Topological Systems, New York University Abu Dhabi}
\curraddr{}
\email{hsati@nyu.edu}
\thanks{}

\author{Urs Schreiber           }
\address{Mathematics Program and Center for Quantum and Topological Systems, New York University Abu Dhabi}
\curraddr{}
\email{us13@nyu.edu}
\thanks{}

\subjclass[2020]{
Primary:
81T70, 
81T45, 
81T13, 
81T30, 
55N20, 
53C08, 
55P62, 
Secondary:
83E50, 
81V27, 
81V70, 
18G35, 
81T35. 
}

\keywords{
Higher gauge theory, 
flux quantization, 
topological quantum field theory,
nonabelian differential cohomology, 
rational homotopy theory, 
Cohomotopy, 
11D supergravity, 
M-theory, 
fractional quantum Hall effect, 
anyons}

\date{\today}

\dedicatory{
  \href{https://ncatlab.org/nlab/show/Center+for+Quantum+and+Topological+Systems}{\includegraphics[width=3.1cm]{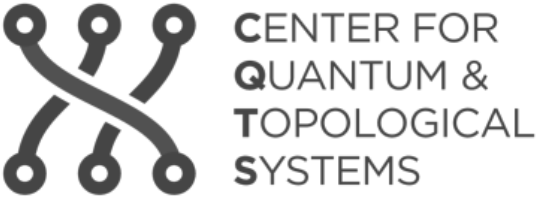}}
}

\begin{abstract}
After global completion of higher gauge fields (as appearing in higher-dimensional supergravity) by proper flux quantization in extraordinary nonabelian cohomology, the (non-perturbative, renormalized) topological quantum observables and quantum states of solitonic field histories are completely determined through a topological form of light-front quantization. We survey the logic of this construction and expand on aspects of the quantization argument.

In the instructive example of 5D Maxwell-Chern-Simons theory (the gauge sector of 5D SuGra)  dimensionally reduced to 3D, a suitable choice of flux quantization in Cohomotopy (``Hypothesis h'') recovers this way the fine detail of the traditionally renormalized (Wilson loop) quantum observables of abelian Chern-Simons theory and makes novel predictions about anyons in fractional quantum (anomalous) Hall systems. 

An analogous choice (``Hypothesis H'') of global completion of 11D higher Maxwell-Chern-Simons theory (the higher gauge sector of 11D SuGra) realizes various aspects of the topological sector of the conjectural ``M-theory'' and its M5-branes.
\end{abstract}

\maketitle

\newpage

\setcounter{tocdepth}{2}
\tableofcontents

\vspace{-.6cm}

These are expanded lecture notes, prepared as handout material for a mini-course of the same title \cite{Edinburgh}, held at the workshop \cite{ICMS:2025:GHS}:
\begin{quote}
``\emph{Geometry, Higher Structures and Physics}'', 
ICMS, Edinburgh (2025).
\end{quote}
The notes are to provide a concise overview of our program \parencites{SS24-Phase,SS24-Obs} of non-Lagrangian global completion of  (supergravity-type, cf. \parencites{FreedmanVanProeyen2012}[\S 6]{Fre2013}{nLab:supergravity}) higher gauge field theories (cf. \cite{SatiSchreiberHigherGauge}) and of the non-perturbative quantization of their solitonic topological sector. 
This is motivated by:
\begin{enumerate}
\item the open problem of global completion \parencites{GSS24-SuGra,FSS20-H} of 11D SuGra (``M-theory'', cf. \cite{Duff1999World,nLab:Mtheory}) and \parencites{GSS25-M5,FSS21-Hopf,FSS21-StrStruc} of its M5-probes (``Theory $\mathcal{X}$'');  

\item relatedly \cite{SS25-Seifert,SS26-Rickles}, the open problem of global understanding \cite{SS25-FQH,SS25-FQAH,SS25-CrystallineChern} of anyons in fractional quantum (anomalous) Hall materials (``topological quantum''). 
\end{enumerate}

Our aim here is to provide a transparent picture of the main steps of the construction, while relegating mathematical background to the references. (Much of that background has been laid out in the monographs \cite{FSS23-Char,SS25-Bun,SS26-Orb}, while further exposition of the approach is in \parencite{SS25-Flux,SS25-Srni,SS25-WilsonLoops}.)

Required of the audience is, in the first part of these lectures, nothing more than familiarity with basics of differential forms and algebraic topology, as found for instance in \cite{Bott1982}, and with applications to physics in \cite{Frankel2011}. At the end of \cref{GlobalCompletionOfHigherGaugeFields} and in parts of \cref{QuantizationInTopologicalSector} we invoke basics of geometric \cite{Lurie2009} and stable homotopy theory \cite{BarnesRoitzheim2020}, such as surveyed in \parencites{Sc18-ToposLectures}[\S 1]{FSS23-Char}. But after the dust has settled, the examples and applications in \cref{ExamplesAndApplications} again require only classical homotopy theory, such as found in \cite{Whitehead1978,Bredon1993}.

Familiarity with physics jargon will aid the reader to orient themselves, but is not strictly necessary. In particular, nowhere do we rely on string/M-theory folklore --- all our concepts have definitions, and our claims have proofs.

\smallskip

For lists of background references beyond the scope of what can reasonably be cited here we will point to entries of the \emph{nLab} research wiki at \href{https://ncatlab.org/nlab/show/HomePage}{\tt ncatlab.org}.

\section{Global Completion of Higher Gauge Fields}
\label{GlobalCompletionOfHigherGaugeFields}

Our ground field is the real numbers. All manifolds and their differential forms are understood to be smooth. Repeated indices imply that we sum over them.

\subsection{Higher Flux Densities}
\label{NonlinearFluxDensities}
\footnote{
  This section parallels \cite[\S 2]{SS25-Flux}. The closest in traditional literature is the ``geometric'' formulation of higher-dimensional supergravity reviewed in \cite{CDF1991}. cf. \cref{GaussLawLInfinityAlgebras}.
}

We discuss the Bianchi identities and equations of motion of higher flux densities, generalizing Maxwell's equations for the electromagnetic field in vacuum. Below in \cref{OnTheCompletedPhaseSpace}, the \emph{on-shell} higher flux densities (those satisfying their equations of motion) are promoted to globally well-defined higher gauge fields with higher \emph{gauge potentials}.
While non-abelian Yang-Mills fields are \emph{not} part of our discussion, we do crucially pay attention to \emph{non-linear} relations in the Bianchi identities/Gau{\ss} laws of the flux density. These lead to the corresponding charges being in extraordinary \emph{non-abelian cohomology} (in \cref{TotalFluxQuantization}).

\subsubsection{Nonlinear Flux Densities}
\footnote{
  We follow \parencites[\S 2.4]{SS25-Flux}[\S 2.1]{SS24-Phase}.
}

Ever since Faraday discovered \emph{magnetic field lines} in the 1850s, the field strength of gauge fields is expressed as the \emph{density of flux} of the field through infinitesimal hypersurfaces in spacetime, which is measured, in modern mathematical language, by differential forms (cf. \cref{MagneticFluxLines})

\begin{figure}[htb]
\caption{
  \label{MagneticFluxLines}
  The local \emph{field strength} $B$ of the magnetic field may be understood as a density of \emph{magnetic flux lines} through any given surface, as such measured by a differential 2-form $F_2$ on space(time)  (the \emph{Faraday tensor}). 
  In higher-dimensional generalization, flux densities of higher gauge fields are encoded by differential forms of higher degree. Graphics adapted from \cite{NaveHPH-MagneticFlux}.
}
\centering
\adjustbox{scale=0.9,
  rndfbox=4pt
}{
\begin{tikzpicture}

  \draw (0,0) node {
    \includegraphics[width=8.8cm]{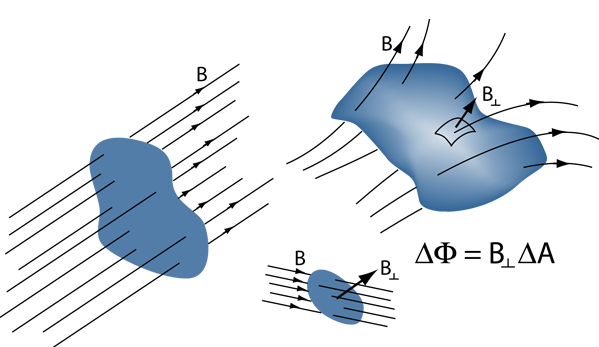}
  };

\node[rotate=-52]
  at (-.15,-.15) {
    \scalebox{.7}{
      \color{darkblue}
      \bf
      \;\;\;magnetic flux lines
    }
  };

\draw[
  draw=white,
  fill=white
]
  (1.7,-.9) rectangle (3.8,-1.5);

\node[
 rotate=-43
]
  at (2,.4) {
   \scalebox{.8}{
     $\Delta \, x^{{}_1}$
   }
  };

\node[
 rotate=+40
]
  at (2.06,.88) {
   \scalebox{.8}{
     $\Delta \, x^{{}_2}$
   }
  };

\node[
  rotate=-40
]
  at (.7,-2.6) {
    \scalebox{.9}{$
      \vec \Delta x^{{}_1} 
      \!\!
      \wedge
      \vec \Delta x^{{}_2}
    $}
  };

\node
  at (3.3,-1.6) {
    \scalebox{1}{
      $
        \def\arraystretch{1.2}
        \begin{array}{l}
        \overset{
          \mathclap{
            \scalebox{.7}{
              \def\arraystretch{.9}
              \color{darkblue}
              \bf
              \begin{tabular}{c}
                magnetic flux 
                \\
                through surface element
              \end{tabular}
            }
          }
        }{
        F_2\big(\vec \Delta x^{_1}, \vec \Delta x^{_2}\big)
        }
        \\
        \,=\,
        B_{\perp} 
          \cdot
        \Delta x^{_1} 
          \cdot 
        \Delta x^{_2    }
        \end{array}
      $
    }
  }; 
\end{tikzpicture}
}
\end{figure}
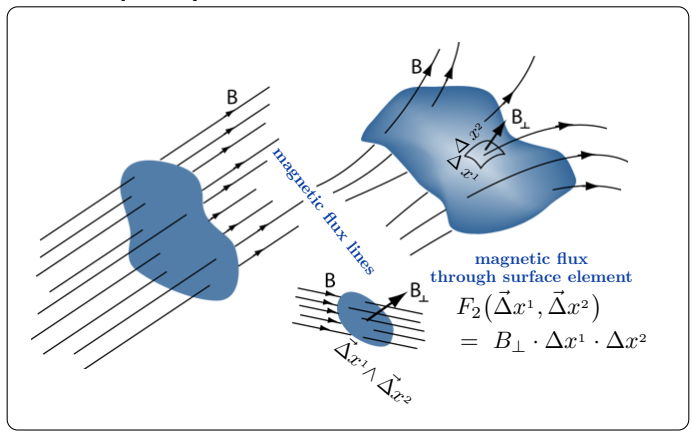

The flux density of the electromagnetic field on a spacetime $X^{1,3}$ is encoded by a differential 2-form $F_2 \in \Omega^2_{\mathrm{dR}}(X^{1,3})$ (the \emph{Faraday tensor}) subject to \emph{Maxwell's equations of motion} (here: in vacuum, cf. \cite[\S 3.5 \& \S7.2b]{Frankel2011}):
\begin{equation}
  \label{MaxwellEquationOfMotion}
  \begin{aligned}
    \mathrm{d} F_2 & = 0 \,,
    \\
    \mathrm{d} \star F_2 & = 0 \,.
  \end{aligned}
\end{equation}

It is a surprisingly good idea (cf. \cite{HehlObukhov2003, nLab:PreMetricElectromagnetism}) to reformulate this equivalently as a \emph{pair} of differential forms $F_2, F'_2 \in \Omega^2_{\mathrm{dR}}(X^{1,3})$ subject to a pair of differential equations complemented by a Hodge-duality constraint:
\begin{equation}
  \label{DualitySymmetricMaxwellEquationOfMotion}
  \begin{aligned}
    \mathrm{d} F_2 & = 0
    \,,
    \\
    \mathrm{d} F'_2 & = 0
    \mathrlap{\,,}
  \end{aligned}
  \;\;\;
  F'_2 = \star F_2
  \,.
\end{equation}

\subsubsection{The equations of Motion}
\label{OnTheEquationsOfMotion}

In this ``pre-geometric'' or ``duality-symmetric'' formulation we say that: 

Flux densities of a \emph{higher gauge field} of higher Maxwell-type on a higher-dimensional spacetime $X^{1,d}$ are a finite set $I$ of \emph{flux species} and an $I$-tuple of differential forms $F^{(i)}$ (the \emph{flux densities} themselves) of positive degree $\mathrm{deg}(i)>0$ on spacetime,
\begin{equation}
  \label{TheFluxDensities}
  \vec F = \big(
    F^{(i)}
    \in
    \Omega^{\mathrm{deg}(i)}(X^{1,d})
  \big)_{i \in I}
  \mathrlap{\,,}
\end{equation}
subject to \emph{equations of motion} of this form:
\begin{equation}
  \label{TheHigherMaxwellEquations}
  \mathllap{
  \forall_{i \in I}
  \;\;\;
  }
  \mathrm{d}F^{(i)}
  =
  P^{(i)}\big(\vec F\, \big)
  \,,\;\;\;
  \star F^{(i)} = 
  \mu^{(i)}(\vec F\, )
  \mathrlap{\,,}
\end{equation}
where the $P^{(i)}$ are wedge-polynomials and the $\mu^{(i)}$ are invertible linear functions in the set $I$ of variables. Hence,
in evident shorthand, the equations of motion read:
\begin{equation}
  \label{ShorthandOfTheHigherMaxwellEquations}
  \mathrm{d}\vec F = \vec P(\vec F\,)
  \,,\;\;\;\;
  \star \vec F = \vec \mu(\vec F\,)
  \mathrlap{\,.}
\end{equation}

\begin{example}[Higher gauge fields]
\label[example]{HigherGaugeFields} {\ }
\begin{description}
  \item[Self-dual/chiral field] on $\Sigma^{1,4k+1}$
  (e.g., gauge sector on M5 for $k=1$)
  \begin{equation}
    \label{EoMOfSelfDualField}
    \mathrm{d}H_{2k+1} = 0
    \,,
    \;\;\;\;
    \star H_{2k+1} = H_{2k+1}
    \mathrlap{\,.}
  \end{equation}

  \item[Type IIA RR-field] on $X^{1,9}$
  (gauge sector of type IIA 10D SuGra)
  \begin{equation}
    \label{EoMOfIIARRField}
    \mathrm{d} F_{2\bullet} = 0
    \,,
    \;\;\;\;\;
    \star F_{2\bullet} = F_{10-2\bullet}
    \mathrlap{\,.}
  \end{equation}

  \item[5D Maxwell-Chern-Simons]
  on $X^{1,4}$
  (gauge sector of minimal 5D SuGra)
  \begin{equation}
    \label{EoMOf5DMCS}
    \begin{aligned}
      \mathrm{d}F_2 & = 0\,,
      \\
      \mathrm{d}F_3 & = 
        \tfrac{1}{2}F_2 \wedge F_2
      \,,
    \end{aligned}
    \;\;\;
    F_3 = \star F_2
    \mathrlap{\,.}
  \end{equation}

  \item[11D higher Maxwell-Chern-Simons]
  on $X^{1,10}$
  (gauge sector of 11D SuGra)
  \begin{equation}
    \label{EoMOfCFieldIn11DSugra}
    \begin{aligned}
      \mathrm{d}G_4 & = 0 \,,
      \\
      \mathrm{d}G_7 & = 
        \tfrac{1}{2}G_4 \wedge G_4\,,
    \end{aligned}
    \;\;\;
    G_7 = \star G_4
    \mathrlap{\,.}
  \end{equation}
\end{description}

Beware that \textbf{nonabelian Yang-Mills fields} are \emph{not} an example of Maxwell-type higher gauge fields \cref{ShorthandOfTheHigherMaxwellEquations} -- we discuss this subtlety in \cref{TheNonExampleOfYangMillsFields}.
\end{example}

More generally, given a \emph{background} $\vec G$ of such higher flux densities then higher flux densities $\vec F$ \emph{coupled to} (``\emph{twisted by}'') the background  are subject to differential equations of the more general form:
\begin{equation}
  \mathrm{d}
  F^{(i)} 
    =
  P^{(i)}\big(
    \vec F, \mathcolor{purple}{\vec G}
  \,\big)
  \mathrlap{\,.}
\end{equation}

\begin{example}[Twisted higher gauge fields]
\nopagebreak \hfill
\begin{description}
  \item[Twisted Type IIA RR-field]
  a \emph{Kalb-Ramond B-field background}
  on $X^{1,9}$
  \begin{equation}
    \label{EoMForBField}
    \begin{aligned}
      \mathrm{d}
      H_3 & = 0
      \\
      \mathrm{d} H_7 & = 0 
      \mathrlap{\,,}
    \end{aligned}
    \;\;\;\;
    H_7 = \star H_3
  \end{equation}
  twists the type IIA RR-field as:
  \begin{equation}
    \mathrm{d} F_{2\bullet}
    =
    F_{2 \bullet -2 }
    \wedge
    \mathcolor{purple}{H_3}
    \mathrlap{\,.}
  \end{equation}

  \item[Twisted self-dual field]
  A C-field background flux 
  $G_4$ 
  twists the self-dual field on an M5-probe as
  $\begin{tikzcd}[sep=small] \Sigma^{1,5} \ar[r, "{ \phi }"] & X^{1,10}\end{tikzcd}$
  \begin{equation}
    \label{M5WorldvolumeBianchi}
    \mathrm{d} H_3 
      =
    \phi^\ast \mathcolor{purple}{F_4} \,.
  \end{equation}
  
\end{description}
\end{example}

\subsubsection{On-shell Flux Densities}
\label{OnShellFluxDensities}
\footnote{
  For more on on-shell flux densities in our context see \parencites[\S 2.5 \& \S 3.1]{SS25-Flux}[\S 2.1]{SS24-Phase}.
}

We proceed to analyze the \emph{space of solutions} of such equations of motion on higher flux densities (which is the precursor of the full \emph{phase space} stack of the corresponding completed higher gauge fields in \cref{OnTheCompletedPhaseSpace}).

For that purpose, assume from now on that spacetime is \emph{globally hyperbolic}, 
\begin{equation}
  \label{GloballyHyperbolicSpacetime}
  X^{1,d} 
    \simeq 
  \mathbb{R}^{1,0} 
    \times 
  X^d
  \mathrlap{\,,}
\end{equation}
with \emph{Cauchy surface} 
\begin{equation}
  \label{InclusionOfCauchySurface}
  \begin{tikzcd}
    X^d 
    \simeq
    \{t_0\} \times X^d
    \ar[
      r, 
      hook,
      "{ \iota_{t_0} }"
    ]
    &
    X^{1,d}
    \mathrlap{\,.}
  \end{tikzcd}
\end{equation}
Let us denote the set of \emph{germs of solutions} of our equations of motion, around $t_0$, by
\begin{equation}
  \label{LocalSolutionSpace}
  \mathrm{LocSol}_{t_0}
  :=
  \left\{
  \substack{
    \scalebox{.7}{germs of solutions}
    \\
    \scalebox{.7}{of equations of motion}
    \\
    \scalebox{.7}{around $\{t_0\} \times X^d$}
  }
  \right\}
  \mathrlap{.}
\end{equation}
Hence, an element in this set is represented by a solution 
\begin{equation}
  \label{ARepresentativeSolution}
  \vec F = 
  \Big(
  F^{(i)} \in 
    \Omega^{\mathrm{deg}}\big(
    (-\epsilon, + \epsilon) \times X^d
  \big)
  \Big)_{i \in I}\,
  ,\quad 
  \mathrm{d} \vec F = \vec P(\vec F\,)
  ,\;\;
  \star \vec F = \vec \mu(\vec F\,)
\end{equation}
for some positive time interval $\epsilon > 0$, and two such representatives are identified if they agree on the time interval where they are both defined. 

It turns out that the above duality symmetric formulation of the equations of motion lends itself to the formulation of their initial value problem:

\begin{proposition}[{\cite[Thm. 2.2]{SS24-Phase}}]
  \label[proposition]{IdentifyingCauchyData}
  Pullback of duality-symmetric flux densities \cref{ShorthandOfTheHigherMaxwellEquations} to the Cauchy surface \cref{InclusionOfCauchySurface} identifies the local solution space \cref{LocalSolutionSpace} with the space of $I$-indexed flux densities on $X^d$ that satisfy the \emph{higher Gau{\ss} law}:
  \begin{equation}
    \label{LocalSolutionsViaInitialValues}
    \begin{tikzcd}[row sep=-2pt, column sep=3pt]
      \mathrm{LocSol}_{t_0}
      \ar[
        rr,
        "{ \iota_{t_0}^\ast }",
        "{ \sim }"{swap}
      ]
      &&
      \Big\{
      \vec B 
      \defneq
      \big(
      B^{(i)}
      \in
      \Omega^{\mathrm{deg}(i)}_{\mathrm{dR}}
      (X^d)
      \big)_{i \in I}
      \;\big\vert\;
      \grayoverbrace{
      \mathrm{d} \vec B
      =
      \vec P(\vec B)
      }{
        \smash{
        \scalebox{.7}{\rm higher Gau{\ss} law}
        }
      }
      \Big\}
      \\
      \vec F
      &\longmapsto&
      \vec B :=  \iota_{t_0}^\ast \vec F
      \mathrlap{\,.}%
      \hspace{5.54cm}
    \end{tikzcd}
  \end{equation}
\end{proposition}
\begin{proof}
  This follows by direct computation spelled out in \cite[\S A]{SS24-Phase}.
\end{proof}

We next reformulate the local solution set \cref{LocalSolutionsViaInitialValues} in a neat way via $L_\infty$-algebras.
To that end, first we recall how finite-type $L_\infty$-algebras are dual to dgc-algebras whose underlying graded algebra is free.

\subsubsection{CE-Algebras of Lie algebras}
\label{CEAlgebrasOfLieAlgebras}

First, for $(\mathfrak{g}, [-,-])$ a finite-dimensional Lie algebra, its \emph{Chevalley-Eilenberg algebra} $\mathrm{CE}(\mathfrak{g})$ is the differential graded-commutative (dgc) algebra, which is the Grassmann algebra $\wedge^\bullet \mathfrak{g}^\ast$ on the dual vector space $\mathfrak{g}^\ast$ equipped with the differential which on generators is the dual of the Lie bracket. So if $(t_i)_{i = 1}^{\mathrm{dim}(\mathfrak{g})}$ is a linear basis for $\mathfrak{g}$ with structure constants $[t_i, t_j] = f^k_{ij} t_k$, and with $(t_1^k)_{i = 1}^{\mathrm{dim}(\mathfrak{g})}$ denoting the dual basis of $\wedge^1 (\mathfrak{g}^\ast)$ (in degree 1), then
\begin{equation}
  \label{CEAlgebraOfALieAlgebra}
  \mathrm{CE}(\mathfrak{g})
  \:=
  \Big(
  \wedge^\bullet \mathfrak{g}^\ast
   ,\;
   \mathrm{d} 
     : 
     t_1^k 
     \mapsto -  
     \tfrac{1}{2} f^k_{ij} 
     \, t_1^i \,t_1^j 
  \Big)
\end{equation}
is a dgc-algebra. In fact, the condition that $\mathrm{d}^2 = 0$ in $\mathrm{CE}(\mathfrak{g})$ is equivalent to the Jacobi identity on $\mathfrak{g}$.
Moreover, the passage $\mathfrak{g} \mapsto \mathrm{CE}(\mathfrak{g})$ is \emph{faithful}: Lie algebra homomorphisms $\begin{tikzcd}[sep=small] \mathfrak{g} \ar[r, "{\phi}"] & \mathfrak{h}\end{tikzcd}$ are in natural bijection with dgc-algebra homomorphisms $\begin{tikzcd}[sep=small] \mathrm{CE}(\mathfrak{g}) \ar[r, <-, "{ \phi^\ast }"{pos=.65}] & \mathrm{CE}(\mathfrak{h})\end{tikzcd}$.

\subsubsection{$L_\infty$-Algebras via CE-Algebras}
\footnote{
  For general introduction to $L_\infty$-algebras cf. \parencites{LadaStacheff1993}{Reinhold2019l}{KraftSchnitzer2024}.
  For our discussion of finite-type $L_\infty$-algebras via their CE-algebras (``FDAs'') see \parencites[Def. 13]{SatiSchreiberStasheff2009}[\S 3]{FSS19-RationalM}[\S 4]{FSS23-Char}{CastellaniDAuria2025}.
}

After this passage, regarding Lie algebras as duals of certain dgc-algebras, it is straightforward to generalize.
For $\mathfrak{g}$ an $\mathbb{N}$-graded vector space of finite type (finite dimensional in each degree) and with $\mathfrak{g}^\vee$ denoting its degreewise dual, consider a differential $\mathrm{d}$ on its graded Grassmann algebra $\wedge^\bullet \mathfrak{g}^{\vee}$. Dually, this encodes a tower of higher graded skew-symmetric brackets on $\mathfrak{g}$:
\begin{equation}
  \begin{tikzcd}[
    column sep=70pt
  ]
    \wedge^1 \mathfrak{g}^{\vee}
    \ar[
      rr,
      "{
        \mathrm{d}
          _{\vert 
            \scalebox{.6}{$\wedge^1 \mathfrak{g}^\vee$}}
        \,=\,
        [-,-]^\ast
        \,+\,
        [-,-,-]^\ast
        \,+\,
        \cdots
      }"
    ]
    &&
    \wedge^2 \mathfrak{g}^\vee
    \oplus
    \wedge^3 \mathfrak{g}^\vee
    \oplus 
    \cdots
    \subset
    \wedge^\bullet \mathfrak{g}^\vee
    \mathrlap{\,,}
  \end{tikzcd}
\end{equation}
on which the condition $\mathrm{d}^2 = 0$ is equivalently a system of higher Jacobi identities.
These structures 
\begin{equation}
  \label{LInfinityAlgebras}
  \big(\mathfrak{g}, [-,-], [-,-,-], \cdots\big)
  \;\;\; \mbox{s.t.}\;\;\;
  \left\{
  \def\arraystretch{1.1}
  \begin{array}{l}
    \mathrm{d} \in 
    \mathrm{Der}^1(\wedge^\bullet \mathfrak{g}^\vee)
    \;
    \mbox{with}
    \\
    \mathrm{d}_{\scalebox{.6}{$\vert \wedge^1 \mathfrak{g}^{\vee}$}}
    \,=\,
    [-,-]^\ast + [-,-,-]^\ast + \cdots
    \\
    \mbox{satisfies}\;
    \mathrm{d}^2 = 0
  \end{array}
    \right.
\end{equation}
are the (real, connected and finite-type) \emph{$L_\infty$-algebras}. 

If we denote by $\mathbb{R}_{\mathrm{d}}[I]$ the \emph{free dgc algebra} on a set $I$ of graded generators, then we may conveniently present the CE-algebras of such $L_\infty$-algebras as quotients by differential ideals, as follows:
\begin{equation}
  \label{CEAlgebraAsQuotientOfFreeDGCA}
  \mathrm{CE}(\mathfrak{g})
  \simeq
  \mathbb{R}_{\mathrm{d}}
  \Big[
    \big\{
      b^{(i)}_{\mathrm{deg}(i)}
    \big\}_{i \in I}
  \Big]
  \Big/
  \Big(
    \mathrm{d}b^{(i)} = P^{(i)}\big(\vec b\,\big)
  \Big)_{i \in I}
  \mathrlap{\,,}
\end{equation}
for some (graded symmetric) polynomials $P^{(i)}$ in the set $I$ of graded variables.
(You are surely seeing now where this is headed!)

For example, when the Lie algebra $\mathfrak{g}$ in \cref{CEAlgebraOfALieAlgebra} carries an $\mathrm{ad}$-invariant metric $g$, then $\mu := g(-,[-,-])$ is a cocycle, $\mathrm{d} \big(\mu_{i j k } t^i t^j t^k\big) = 0$, so that we obtain a 2-term $L_\infty$-algebra (a \emph{Lie 2-algebra}) $\mathfrak{string}_{\mathfrak{g}}$, characterized by
\begin{equation}
  \label{CEOfStringLie2Algebra}
  \mathrm{CE}\big(
    \mathfrak{string}_{\mathfrak{g}}
  \big)
  \simeq
  \mathbb{R}_{\mathrm{d}}\left[
    \begin{aligned}
      &(t^i_1)_{i = 1}^{\mathrm{dim}(\mathfrak{g})}
      \\
      &\;\;\;b_2 
    \end{aligned}
  \right]
  \Big/
  \left(
  \begin{aligned}
    \mathrm{d} \, t_1^k & = 
    - \tfrac{1}{2} f^k_{ij} \,t^i_1\, t^j_1
    \\[-2pt]
    \mathrm{d} \, b_2 &= \mu_{i j k} 
    \,
    t_1^i \, t_1^j \, t_1^k
  \end{aligned}
  \right)
  \mathrlap{.}
\end{equation}
The string Lie 2-algebra \cref{CEOfStringLie2Algebra} is an important example of $L_\infty$-algebras,
which serves as \emph{coefficients of gauge potentials} of higher gauge fields. However, we are now actually interested in a rather different class of examples of $L_\infty$-algebras, namely those that serve as \emph{coefficients of flux densities} of higher gauge fields.

\subsubsection{Whitehead bracket $L_\infty$-Algebras}
\label{OnWhitheadBracketLInfinityAlgebras}
\footnote{
  For detailed surveys of $\mathbb{R}$-rational homotopy theory, Sullivan models and their dual Whitehead $L_\infty$-algebras see \cite[\S 5 \& Prop. 5.11, Rem. 5.4]{FSS23-Char}. For more examples, see \cite[p. 16]{SS25-Flux}.
}

For our purpose, the key source of (connected, finite-type) \emph{nilpotent $L_\infty$-algebras} is the following:

\begin{proposition}[Whitehead $L_\infty$-algebras and minimal Sullivan models, {cf. \cite[Prop. 5.11]{FSS23-Char}}]
\label{WhiteheadLInfinityAlgebras}
Given $\mathcal{A}$ a connected topological space with abelian fundamental group and finite-dimensional rational cohomology in each degree, then there exists an essentially unique $L_\infty$-algebra $\mathfrak{l}\mathcal{A}$ such that:
\begin{enumerate}
  \item its underlying graded vector space is that of the $\mathbb{R}$-rationalized homotopy groups of the loop space of $\mathcal{A}$:
  \begin{equation}
    (\mathfrak{l}\mathcal{A})_\bullet 
    \simeq 
    \pi_\bullet\big( \Omega \mathcal{A} \big)
    \otimes_{\mathbb{Z}}
    \mathbb{R}
    \mathrlap{\,;}
  \end{equation}

  \item
  its brackets are such that the cochain cohomology of its CE-algebra is the $\mathbb{R}$-cohomology of $\mathcal{A}$:
  \begin{equation}
    H^\bullet\big(
      \mathrm{CE}(
        \mathfrak{l}\mathcal{A}
      )
    \big)
    \simeq
    H^\bullet\big(
      \mathcal{A}
      ,\,
      \mathbb{R}
    \big)
    \mathrlap{\,.}
  \end{equation}
\end{enumerate}
\end{proposition}
Some jargon:
\begin{itemize}
\item 
the brackets of $\mathfrak{l}\mathcal{A}$ are the \emph{higher Whitehead brackets} of $\mathcal{A}$ over $\mathbb{R}$, 
\item the CE-algebra $\mathrm{CE}(\mathfrak{l}\mathcal{A})$ is known as the \emph{minimal Sullivan model} of $\mathcal{A}$.
\end{itemize}

For example:
\begin{enumerate}
\item
An Eilenberg-MacLane space
$\mathcal{A} = K(n,\mathbb{Z}) = B^n \mathbb{Z}$ is characterized by the fact that its homotopy is concentrated on $\pi_{n}(B^n \mathbb{Z}) = \mathbb{Z}$, whence $\mathrm{CE}(\mathfrak{l}B^n \mathbb{Z})$ has a single generator $\omega_{n}$ spanning $\pi_n \otimes \mathbb{R} \simeq \mathbb{R}$, which therefore must necessarily be closed, and hence:
\begin{equation}
  \label{CEAlgebraOfEMSpace}
  \mathrm{CE}(\mathfrak{l}B^n \mathbb{Z})
  \simeq
  \mathbb{R}_{\mathrm{d}}
  \big[
    \omega_{n}
  \big]
  \big/
  \big(
    \mathrm{d}\, \omega_n = 0
  \big)
  \mathrlap{\,.}
\end{equation}
The same result is obtained for $K(n, \mathbb{Q}) \simeq B^n \mathbb{Q}$: $\mathfrak{l}(-)$ and $\mathrm{CE}(\mathfrak{l}(-))$ retain (only) the \emph{rational homotopy type}.

\item 
  For a product of EM-spaces, $\mathcal{A} = B^{n_1} \mathbb{Z} \times B^{n_2} \mathbb{Z}$, their homotopy groups are the products of the factors, and hence:
\begin{equation}
  \label{CEAlgebraOfProductOfEMSpaces}
  \mathrm{CE}\big(
    \mathfrak{l}(B^{n_1} \mathbb{Z}
    \times
    B^{n_2} \mathbb{Z}
    )
  \big)
  \simeq
  \mathbb{R}_{\mathrm{d}}
  \left[
    \begin{aligned}
      \omega_{n_1}
      \\
      \omega'_{n_2}
    \end{aligned}
  \right]
  \big/
  \left(
    \begin{aligned}
      \mathrm{d}\, \omega_{n_1} & = 0
      \\
      \mathrm{d}\, \omega'_{n_2} & = 0
    \end{aligned}
  \right)
  \mathrlap{.}
\end{equation}

\item
The classifying space of the stable unitary group
$\mathcal{A} = \widetilde{\mathrm{KU}} := \bigcup_{n \in \mathbb{N}} B \mathrm{U}(n)$ has homotopy groups concentrated on $\pi_{2\bullet}(B \mathrm{U}) \simeq \mathbb{Z}$ and hence
\begin{equation}
  \label{CEAlgebraOfBU}
  \mathrm{CE}\big(
    \mathfrak{l}B \mathrm{U}
  \big)
  \simeq
  \mathbb{R}_{\mathrm{d}}
  \big[
    \omega_{2\bullet}
  \big]
  \big/
  \big(
    \mathrm{d} \, \omega_{2 \bullet}  = 0
  \big)\,.
\end{equation}

\end{enumerate}

Moreover, by the \emph{Serre finiteness theorem} we have:
\begin{itemize}
  \item[{\bf (iv)}] For $\mathcal{A} = S^{2k+1}$ an odd-dimensional sphere, the only non-torsion homotopy group is $\pi_{2k+1}(S^{2k+1}) \simeq \mathbb{Z}$ and also the $\mathbb{R}$-cohomology is concentrated on $H^{2k+1}(S^{2k+1}; \mathbb{R}) \simeq \mathbb{R}$. Therefore $\mathrm{CE}(\mathfrak{l}S^{2k+1})$ has a single generator $\omega_{2k+1}$ spanning $\pi_{2k+1} \otimes \mathbb{R} \simeq \mathbb{R}$, which must be closed to also span $H^{2k+1}$. Being of odd degree, it generates no further cohomology under wedge product, and hence:
  \begin{equation}
    \mathrm{CE}\big(
      \mathfrak{l}S^{2k+1}
    \big)
    \simeq
    \mathbb{R}_{\mathrm{d}}
    \big[
    \omega_{2k + 1}
    \big]
    \big/
    \big(
      \mathrm{d} \, \omega_{2k + 1} = 0
    \big)
    \mathrlap{\,,}
  \end{equation}
  just as for $B^{2k+1} \mathbb{Z}$ \cref{CEAlgebraOfEMSpace}.
  
  \item[{\bf (v)}] An even dimensional sphere $\mathcal{A} = S^{2k}$ however has non-torsion homotopy concentrated not just on $\pi_{2k}(S^{2k}) \simeq \mathbb{Z}$ but also on $\pi_{4k-1}(S^{2k}) \simeq \mathbb{Z}$, while the real cohomology is still concentrated on $H^{2k}(S^{2k}; \mathbb{R}) \simeq \mathbb{R}$. Therefore $\mathrm{CE}(\mathfrak{l}S^{2k})$ has two generators, $\omega_{2k}$ and $\omega_{4k-1}$ spanning $\pi_{2k} \otimes \mathbb{R} \simeq \mathbb{R}$ and $\pi_{4k-1} \otimes \mathbb{R} \simeq \mathbb{R}$, of which $\omega_{2k}$ still needs to be closed in order to span $H^{2k}$. But now its square $\omega_{2k} \, \omega_{2k}$ does not vanish, while still being closed. In order to remove this contribution from the cochain cohomology of the CE-algebra, the other generator must be a coboundary of the square, and hence:
  \begin{equation}
    \label{CElOfEvenDimensionalSphere}
    \mathrm{CE}(\mathfrak{l}S^{2k})
    \simeq
    \mathbb{R}_{\mathrm{d}}
    \left[
    \begin{aligned}
      \omega_{2k}\;\;
      \\[-2pt]
      \omega_{4k-1}
    \end{aligned}
    \right]
    \Big/
    \left(
    \begin{aligned}
      \mathrm{d}\, \omega_{2k} & = 0
      \\[-2pt]
      \mathrm{d}\, \omega_{4k-1} & = 
      \tfrac{1}{2} \omega_{2k} \omega_{2k}
    \end{aligned}
    \right)
  \end{equation}
  (where the prefactor of $1/2$ is just convention: one could use any nonvanishing factor here without changing the isomorphism class of the CE-algebra).

  By \cref{LInfinityAlgebras}, the corresponding $L_\infty$-algebra $\mathfrak{l}S^{2k}$ has generators $v_{2k-1}$ and $v_{4k-2}$ on which the unique nonvanishing bracket is
  \begin{equation}
    [v_{2k-1}, v_{2k-1}]
    =
    - v_{4k-2}
    \mathrlap{\,.}
  \end{equation}
  For $k = 2$ this is also known as the \emph{gauge algebra of 11D SuGra} \parencites[(2.6)]{CremmerJuliaLuPope1998}[\S 4]{Sati2010}[Ex. 2.2]{SatiVoronov2024}{nLab:CFieldGaugeAlgebra} and as such it appears below in \cref{SomeCharacteristicLInfinityAlgebras}.
\end{itemize}

\subsubsection{Closed $L_\infty$-valued Forms}
\label{ClosedLInfinityValuedForms}
\footnote{
  For more on closed $L_\infty$-algebraic forms see \cite[\S 6]{FSS23-Char}, following \cite[Prop. 26]{SatiSchreiberStasheff2009}.
}

Given an $L_\infty$-algebra 
$\mathfrak{g}$ \cref{LInfinityAlgebras}, we say that the (\emph{flat} or) \emph{closed $\mathfrak{g}$-valued differential forms} on a manifold $X$ are the dg-algebra homomorphisms from its CE-algebra \cref{CEAlgebraAsQuotientOfFreeDGCA} to the de Rham algebra of differential forms on $X$:
\begin{equation}
  \label{ClosedLInifnityValuedForms}
  \Omega^1_{\mathrm{dR}}\big(
    X;\, \mathfrak{g}
  \big)_{\mathrm{cl}}
  :=
  \mathrm{Hom}_{\mathrm{dgAlg}}\big(
    \mathrm{CE}(\mathfrak{g})
    ,\,
    \Omega^\bullet_{\mathrm{dR}}(X)
  \big)
  \mathrlap{\,.}
\end{equation}

For example, given $\mathfrak{g}$ an ordinary Lie algebra $(\mathfrak{g}, [-,-])$ as in \cref{CEAlgebrasOfLieAlgebras},  then 
\begin{equation}
  \begin{aligned}
  \Omega^1_{\mathrm{dR}}\big(
    X;\, \mathfrak{g}
  \big)_{\mathrm{cl}}
  &=
  \left\{
  \begin{tikzcd}[sep=5pt]
    \Omega^\bullet_{\mathrm{dR}}(X)
    \ar[
      rrr,
      <-,
      "{
        A
      }"
    ]
    &[-25pt]&&
    \mathrm{CE}(\mathfrak{g})
    \\[-5pt]
    A^i
    \ar[
      rrr,
      shorten=6pt,
      <-|
    ]
    \ar[
      dd,
      |->,
      "{ \mathrm{d} }"
    ]
    &&&
    t^i_1
    \ar[
      ddd,
      |->,
      "{ \mathrm{d} }"
    ]
    \\
    \\
    \mathrm{d} A^i
    \ar[
      dr,
      equals
    ]
    \\
    &
    -\tfrac{1}{2}
    f^{i}_{jk} \, A^j \wedge A^k
    \ar[
      rr,
      <-|
    ]
    &&
    - \tfrac{1}{2} f^{i}_{jk} \, t^j \, t^k
  \end{tikzcd}
  \right\}
  \\
  & = 
  \Big\{
  A \in \Omega^1_{\mathrm{dR}}(X)
  \,\big\vert\,
  F_A \defneq 
  \mathrm{d} A + \tfrac{1}{2}[A\wedge A]
  = 0
  \Big\}
  \\
  & \simeq
  \Omega^1_{\mathrm{dR}}(X; \mathfrak{g}
  )_{\mathrm{flat}}
  \end{aligned}
\end{equation}
is the familiar set of flat $\mathfrak{g}$-valued 1-forms. 

On the other hand, consider the \emph{inner automorphism Lie 2-algebra} $\mathfrak{inn}(\mathfrak{g})$ of this Lie algebra $\mathfrak{g}$, characterized by:
\begin{equation}
  \mathrm{CE}\big(
  \mathfrak{inn}(\mathfrak{g})
  \big)
  \simeq
  \mathbb{R}_{\mathrm{d}}
  \left[
  \begin{aligned}
    (t_1^i)_{i = 1}^{\mathrm{dim}(\mathfrak{g})}
    \\
    (r_2^i)_{i = 1}^{\mathrm{dim}(\mathfrak{g})}
  \end{aligned}
  \right]
  \Big/
  \left(
  \begin{aligned}
    \mathrm{d}\, t^k_1 &= 
    -\tfrac{1}{2} f^k_{i j} \, t^i_1\, t^j_1
    -
    r^k_2
    \\
    \mathrm{d}\, r_2^k 
    & =
    -
    \;
    f^k_{i j} \, t_1^j \, r_2^k
  \end{aligned}
  \right)
  \mathrlap{.}
\end{equation}
The closed $\mathfrak{inn}(\mathfrak{g})$-valued differential forms are nominally pairs consisting of an \emph{unconstrained} $\mathfrak{g}$-valued 1-form $A$ and a 2-form $F$ which is constrained to be its curvature, and whose flatness/closure condition is the classical \emph{Bianchi identity} (which is no further constraint but a consequence):
\begin{equation}
  \begin{aligned}
  \Omega^1_{\mathrm{dR}}\big(
    X;\,
    \mathfrak{inn}(\mathfrak{g})
  \big)
  &
  \simeq
  \left\{
  \begin{aligned}
    A 
    \in \Omega^1_{\mathrm{dR}}(X; \mathfrak{g})
    \\
    F_A 
    \in \Omega^2_{\mathrm{dR}}(X; \mathfrak{g})
  \end{aligned}
  \;\middle\vert\;
  \begin{aligned}
    \mathrm{d} A & = - \tfrac{1}{2}[A\wedge A]
    -
    F_A
    \\
    \mathrm{d} F_A & = \;\; [A \wedge F_A]
  \end{aligned}
  \right\}
  \\
  & 
  \simeq
  \Omega^1_{\mathrm{dR}}(X; \mathfrak{g})
  \mathrlap{\,.}
  \end{aligned}
\end{equation}

It is this kind of constraint that we are after: Flatness of $L_\infty$-valued forms understood as \emph{Gau{\ss} laws or Bianchi identities on flux densities}. But because the word ``flatness'' invokes the situation of gauge potentials satisfying a non-generic constraint, whereas Gau{\ss} laws/Bianchi identities are generic identities specializing to closure for the case of an abelian gauge field, we will speak of \emph{closed} $L_\infty$-valued forms. It is the same mathematical concept, just a more suggestive word for it, in our context (cf. \cref{RoleOfLInfinityAlgebras}).

\begin{table}[htb]
\caption{
  \label{RoleOfLInfinityAlgebras}
  It is a common idea in higher gauge theory to consider \emph{local gauge potentials} with coefficients in $L_\infty$-algebras (cf. \parencites[\S\S 2,4]{Alfonsi2025HigherGeometry}{BorstenEtAl2025-HigherGaugeTheory}, going back to \cite{BaezSchreiber2007}). In contrast,
  {\protect\footnotemark}
  in our discussion here the $L_\infty$-algebras appear as coefficients of the \emph{flux densities} (while the gauge potentials are provided globally and need not have any local expression in terms of $L_\infty$-value forms, cf. \cref{TheHigherGaugePotentials}).
  This change of role is not as widely appreciated, though it is implicit in the old ``geometric'' approach to supergravity via ``FDAs'' (cf. \cref{GaussLawLInfinityAlgebras} and \cref{SpaceOfChoicesOfFluxQuantizationLaws}).
}
\centering
\vspace{-.3cm}
\adjustbox{rndfbox=4pt}{
\begin{tblr}{
  colspec={cc|c},
  row{1} = {bg=lightgray}
}
&
\SetCell[c=2]{c}
  $L_\infty$-algebras $\mathfrak{g}$/$\mathfrak{a}$ as coefficients 
  &
\\
For 
&
 gauge potentials & flux densities
\\
considered & elsewhere & here
\\
are 
&
$\begin{array}{c}
  \mathrm{su}(n), \mathrm{string}(n),...
\end{array}$
& nilpotent
\\
and 
  & 
\SetCell[c=2]{c}
these $\mathfrak{g}$/$\mathfrak{a}$-valued forms
\\
are & 
\begin{tabular}{c}generically\\not flat\end{tabular} 
& 
\begin{tabular}{c}
  always flat, aka closed:
  \\
  Bianchi identity/Gau{\ss} law
\end{tabular}
\end{tblr}
}
\end{table}

\footnotetext{
  For nonabelian $L_\infty$-algebras there is no tight relation between the two perspectives in\cref{RoleOfLInfinityAlgebras}. However, gauge potentials in the left column with coefficients in $\mathfrak{g}$ tend to provide representatives (``cycle data'') for differential cocycles defined (in \cref{OnTheCompletedPhaseSpace}) via the right column with coefficients in $L_\infty$-algebras of invariant polynomials on $\mathfrak{g}$ \parencites[\S 8.3]{SatiSchreiberStasheff2009}{FSSt12-DiffClasses}). The archetypical example here is the representation of classes in differential K-theory (RR-fields according to \emph{Hypothesis K}, \cref{HypothesisK}) by principal connections (Yang-Mills gauge potentials), see \cref{TheNonExampleOfYangMillsFields}. 
}

For example, we now have that the Bianchi identities on the duality-symmetric flux densities of 11D supergravity equivalently characterize closed $\mathfrak{l}S^4$-valued differential forms:
\begin{equation}
  \label{ClosedS4ValuedForms}
  \Omega^1_{\mathrm{dR}}\big(
    X;\,
    \mathfrak{l}S^4
  \big)_{\mathrm{cl}}
  =
  \left\{
  \begin{aligned}
    G_4 \in \Omega^4_{\mathrm{dR}}(X)
    \\
    G_7 \in \Omega^7_{\mathrm{dR}}(X)
  \end{aligned}
  \,\middle\vert\,
  \begin{aligned}
    \mathrm{d}\, G_4 &= 0
    \\
    \mathrm{d}\, G_7 & =
    \tfrac{1}{2} G_4 \wedge G_3
  \end{aligned}
  \right\}
  \mathrlap{.}
\end{equation}

\subsubsection{Gau{\ss} Law $L_\infty$-Algebras}
\label{GaussLawLInfinityAlgebra}
\footnote{%
  \addtocounter{footnote}{-1}%
  \refstepcounter{footnote}%
  \label{GaussLawLInfinityAlgebras}%
  This notion of \emph{Gau{\ss} law $L_\infty$-algebras} is due to \cite[Rem. 2.4]{SS24-Phase}, reviewed in \parencites[\S 3.1]{SS25-Flux}, going back to \parencites{FSS17-Sphere}[\S 2.5]{Sati2018}[\S 7]{FSS19-RationalM}. In some form this may also be recognized in the ``geometric formulation'' of supergravity theories via ``FDA''s (a misnomer for the \emph{semi-free} Chevalley-Eilenberg DGAs) due to \cite{NeemanRegge1978,DAuriaFre1982,vanNieuwenhuizen1983}, cf. \cite{CDF1991,nLab:DFRSupergravity}, reviewed in \cite[\S 6]{Fre2013} and, with dual the $L_\infty$-algebra structure made manifest, in \parencites{FSS15-WZW}{FSS18-TD}{HSS2019}{GSS25-TD}{CastellaniDAuria2025}.
}
\nopagebreak \newline
\indent The expression \cref{LocalSolutionsViaInitialValues} for the solution space of the equations of motion for higher flux densities resembles the expressions \cref{CEAlgebraAsQuotientOfFreeDGCA,ClosedLInifnityValuedForms} for closed $L_\infty$-algebra valued forms: 
the graded-symmetric polynomials $\vec P$ in each case satisfy the same kind of condition and play the same role:

\begin{proposition}
  \label[proposition]
   {SolutionSpaceAsClosedDiffForms}
  Given equations of motion for higher flux densities \cref{ShorthandOfTheHigherMaxwellEquations}, then there exists a unique $L_\infty$-algebra $\mathfrak{a}$ \cref{LInfinityAlgebras} such that the solution set of on-shell flux densities \textup{(\cref{OnShellFluxDensities})} is in natural bijection with that of closed $\mathfrak{a}$-valued forms on any Cauchy surface $X^d$:
  \begin{equation}
    \label{SolutionSpaceAsClosedForms}
    \mathrm{LocSol}
    \simeq
    \Omega^1_{\mathrm{dR}}\big(
      X^d;\, \mathfrak{a}
    \big)_{\mathrm{cl}}
    \,.
  \end{equation}
\end{proposition}

\begin{figure}[htb]
\label{TimeEvolutionAsConcordance}
\caption{
  \label{CauchyEvolutionOfGaussLawData}
  The higher Gau{\ss} law condition on higher electric/magnetic flux densities on any Cauchy surface, and with it the space of on-shell spacetime field histories of higher flux densities (\cref{IdentifyingCauchyData}), are equivalently the closure conditions on differential forms valued in a characteristic $L_\infty$-algebra $\mathfrak{a}$ (\cref{ClosedLInfinityValuedForms}). Under time evolution, this closed $\mathfrak{a}$-valued differential form data undergoes a \emph{concordance}, which preserves exactly the total flux image in $\mathfrak{a}$-valued nonabelian de Rham cohomology (\cref{TotalFluxInNonaDRCohomology}).
}
\centering
\adjustbox{
  rndfbox=4pt
}{
\includegraphics[width=.9\textwidth]{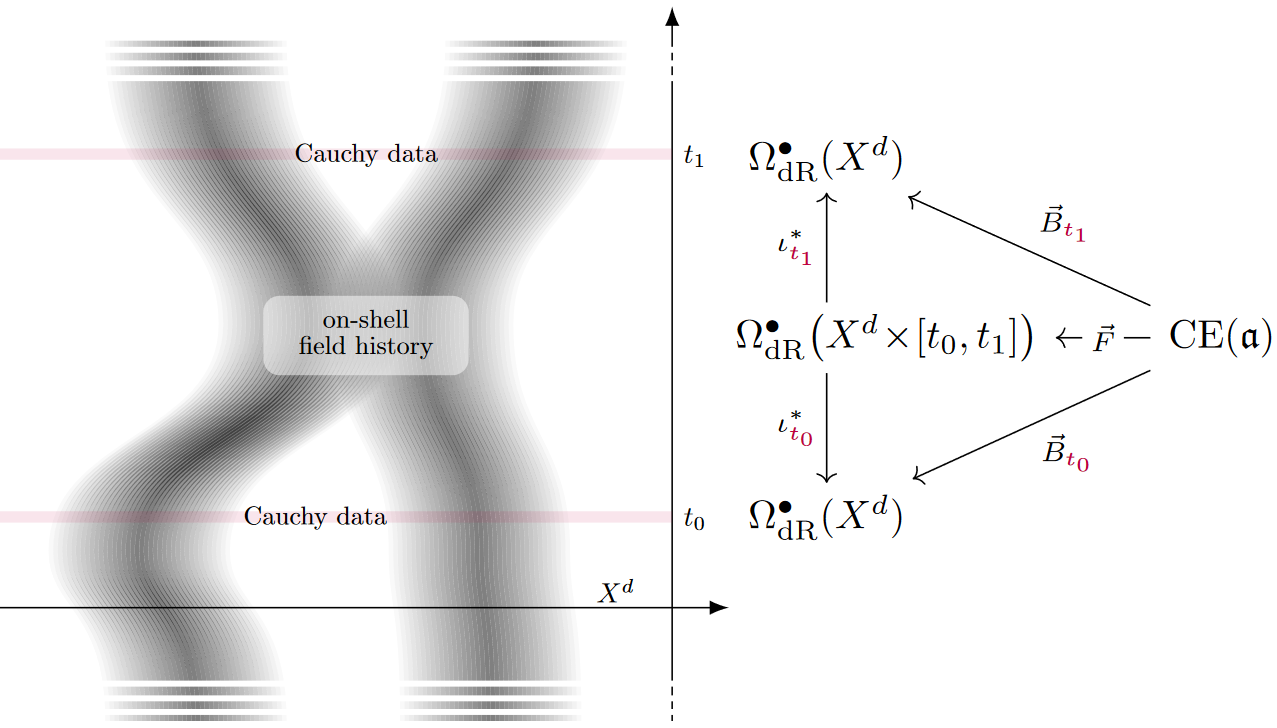}
}
\end{figure}

\begin{example}
\label[example]{SomeCharacteristicLInfinityAlgebras}
Direct comparison shows that the characteristic $L_\infty$-algebra $\mathfrak{a}$ is:

\begin{enumerate}[
  ref={\bf (\roman*)} in \cref{SomeCharacteristicLInfinityAlgebras}
]
  \item\label{CharacteristicOfMaxwellField}
  $\mathfrak{a} =\mathfrak{l} (B^{2} \mathbb{Z} \times B^2 \mathbb{Z})$
  \cref{CEAlgebraOfProductOfEMSpaces}
  for the Maxwell field in vacuum
  \cref{DualitySymmetricMaxwellEquationOfMotion},

  \item\label{CharacteristicOfSelfDualField}
  $\mathfrak{a} =\mathfrak{l} B^{2k+1} \mathbb{Z}$
  \cref{CEAlgebraOfEMSpace}
  for the self-dual field in $D=4k+2$ \cref{EoMOfSelfDualField},

  \item\label{CharacteristicOfIIARRField}
  $\mathfrak{a} = \mathfrak{l} B \mathrm{U}$
  \cref{CEAlgebraOfBU}
  for the RR-field in type IIA 10D SuGra
  \cref{EoMOfIIARRField},
  \footnote{
    Since we have been assuming connected classifying spaces here for simplicity, the degree=0 RR-flux is not included as stated, but the generalization is straightforward, cf. \cite[(91)]{HSS2019}.
  }

  \item\label{CharacteristicOf5DMCSField} 
  $\mathfrak{a} = \mathfrak{l}S^2$ \cref{CElOfEvenDimensionalSphere}
  for the 5D MCS theory in 5D SuGra 
  \cref{EoMOf5DMCS},
  
  \item\label{CharacteristicOfCField} 
  $\mathfrak{a} = \mathfrak{l}S^4$ \cref{CElOfEvenDimensionalSphere}
  for the 7D MCS theory  in 11D SuGra 
  \cref{EoMOfCFieldIn11DSugra}.
\end{enumerate}
\end{example}

This shows how the equations of motion of higher flux densities are actually ``homotopy theoretic'' in nature,  reflecting the \emph{rationalization} of potentially richer topological structure. This perspective is what allows us now to naturally progress to \emph{flux quantization} of these flux densities, by lifting them through the rationalization process, concretely by lifting them through the \emph{character map} on \emph{nonabelian cohomology}.

\subsection{Total Flux Quantization}
\label{TotalFluxQuantization}
\footnote{
  The notion of \emph{flux quantization} (often: ``\emph{charge quantization}'') goes back to \cite{Dirac1931} for the case of the ordinary magnetic flux quantized in ordinary 2-cohomology (``Dirac charge quantization''); modern accounts are \parencites{Alvarez1985}[\S 16.4e]{Frankel2011}. The idea that magnetic ``B-field'' flux can and should analogously be quantized in ordinary 3-cohomology is due to \cite{Gawedzki1988} and was popularized by \cite{FreedWitten1999}. The idea that magnetic ``C-field'' flux can and should analogously be quantized in a (shifted form of) ordinary 4-cohomology is due to \cite{Witten1997flux,Witten1997fivebrane} and motivated the definition of Whitehead-generalized (meaning: abelian generalized) differential cohomology by \cite{HopkinsSinger2005}. 
  The proposal that ``D-brane charge is in twisted K-theory'' (\emph{Hypothesis K}, cf. \cref{SomeFluxQuantizationHypotheses}) is, originally somewhat implicitly \parencites{MooreWitten2000}{nLab:DBraneChargeQuantizationInKTheory}, the further generalization to flux quantization of the RR-field. Understanding this in a general context of higher fluxes quantized in Whitehead-generalized (meaning: abelian) cohomology theories is due to \cite{Freed2000}. 
  
  But quantizing RR-flux in K-theory means, implicitly, to quantize not just the magnetic but also the electric fluxes (cf. again \cref{SomeFluxQuantizationHypotheses}), an evident fact whose significance may have remained underappreciated. 
  That in the presence of non-linear Bianchi identities the joint quantization of magnetic and electric fluxes must be in nonabelian cohomology was made fully explicit in \parencites[\S 3]{FSS22-Twistorial}{SS24-Phase}{SS25-Flux}, with the required \emph{nonabelian character map} established in \cite{FSS23-Char}, following the specific discussion for 11D SuGra fluxes in \cite{FSS20-H} (which provides a twisted nonabelian solution to the above ``shifted quantization'' problem of the C-field, by \cite[\S 3.4 Prop. 3.13]{FSS20-H}).
}

With the on-shell flux densities understood, the task is now to impose their \emph{flux quantization}. The key to progress on that matter is that, with \cref{SolutionSpaceAsClosedDiffForms}, we have identified the solution space of flux densities purely as an object of $\mathbb{R}$-\emph{rational homotopy theory}. This allows us to understand it as the image of a rationalization/character map of spaces of discrete (``quantized'') charges in nonabelian cohomology.

\subsubsection{Total flux in nonabelian dR cohomology}
\label{TotalFluxInNonaDRCohomology}
\footnote{
  The notion of total flux in nonabelian de Rham cohomology is made explicit in \cite[\S 3.1]{SS25-Flux}.
}

First, we need to grasp the notion of \emph{total flux} embodied by flux densities:

The flux densities seen on a Cauchy surface $ \begin{tikzcd}[sep=small]\{t_0\} \times X^d \ar[r, hook] & X^{1,d} \end{tikzcd}$ \cref{LocalSolutionsViaInitialValues} depend on $t_0$, but differ by the flux densities at another time $t_1$ by a \emph{concordance} of closed $\mathfrak{a}$-valued differential forms \cref{SolutionSpaceAsClosedForms}, namely by a joint extension to the interval $[t_0,t_1]$ provided by the spacetime flux density $\vec F$ (cf. \cref{CauchyEvolutionOfGaussLawData}):
\begin{equation}
  \begin{tikzcd}[
    column sep=0pt
  ]
    \{t_1\}
    \times
    X^d
    \ar[
      d,
      hook,
      "{
        \iota_{t_1}
      }"
    ]
    &[10pt]
    \vec B(t_1)
    &\in&
    \Omega^1_{\mathrm{dR}}\big(
      X^d; \mathfrak{a}
    \big)_{\mathrm{cl}}
    \\
    {[t_0, t_1]} \times X^d
    &
    \vec F
    \ar[
      d,
      |->,
      shorten=2pt,
      "{
        \iota_{t_1}^\ast
      }"
    ]
    \ar[
      u,
      |->,
      shorten=2pt,
      "{
        \iota_{t_1}^\ast
      }"{swap}
    ]
    &\in&
    \Omega^1_{\mathrm{dR}}\big(
      [t_0, t_1] \times X^d; \mathfrak{a}
    \big)_{\mathrm{cl}}
    \\
    \{t_0\}
    \times
    X^d
    \ar[
      u,
      hook',
      "{
        \iota_{t_0}
      }"{swap}
    ]
    &
    \vec B(t_0)
    &\in&
    \Omega^1_{\mathrm{dR}}\big(
      X^d; \mathfrak{a}
    \big)_{\mathrm{cl}}
    \mathrlap{\,.}
  \end{tikzcd}
\end{equation}
Hence, the \emph{total flux} $[\vec B\,]$ --- which should be that quantity which is preserved under time evolution of flux densities --- is the \emph{concordance class} of a flux density $\vec B$. The set of the concordance classes of flat $\mathfrak{a}$-valued forms is the \emph{nonabelian de Rham cohomology} $H^1_{\mathrm{dR}}(-;\mathfrak{a})$:
\begin{equation}
  \label{NonabelianDeRhamCohomology}
  [\vec B\,]
  \in
  H^1_{\mathrm{dR}}(X^d; \mathfrak{a})
  :=
  \Big(
  \Omega^1_{\mathrm{dR}}\big(
    X^d; \mathfrak{a}
  \big)_{\mathrm{cl}}
  \Big)\big/\mathrm{concordance}
  \mathrlap{\,.}
\end{equation}

For example:
\begin{itemize}
  \item
  For $\mathfrak{a} = \mathfrak{l}B^n \mathbb{Z}$ \cref{CEAlgebraOfEMSpace}, the nonabelian de Rham cohomology \cref{NonabelianDeRhamCohomology} reduces to ordinary de Rham cohomology:
  \begin{equation}
    \label{OrdinaryDeRhamAsNonabelianDeRham}
    H^1_{\mathrm{dR}}\big(
      X; \mathfrak{l}B^n \mathbb{Z}
    \big)
    \simeq
    H^n_{\mathrm{dR}}(X)
    \,.
  \end{equation}
  Therefore, the \emph{total magnetic flux}, in the above sense, encoded in the ordinary Maxwell flux density $F_2$ \cref{MaxwellEquationOfMotion} is its de Rham cohomology class, as it should be.
\end{itemize}

\subsubsection{Charge in Nonabelian Cohomology}
\label{TotalCharge}
\footnote{
  Early consideration of this elegant notion \cref{TheNonabelianCohomologySet} of generalized higher nonabelian cohomology (subsuming the component-based definitions going back to Giraud) is in \parencites[Def. 6.0.6]{Toen2002}[Def. 2.3]{Schreiber2009OWR}{NSS2015a}[Def. 6]{Lurie2014}, monographs include \parencites[\S2-3]{FSS23-Char}{SS26-Orb}.

  The relation to physical charges is further discussed in \cite[p. 24]{SS25-Flux}, going back to \cite{Schreiber2009OWR}. 
}

In the case that 
\begin{equation}
  \begin{aligned}
  X^d 
    &\simeq 
  \mathbb{R}^d \setminus \mathbb{R}^p
  \\
    &
    \simeq
  \mathbb{R}^p 
    \times 
  \mathbb{R}_{> 0}
    \times 
  S^{d-p-1}
  \end{aligned}
\end{equation}
is the complement of a ``$p$-brane'' submanifold (assumed Cartesian here, just for the moment and for simplicity)
then the total flux on $X^d$ \cref{NonabelianDeRhamCohomology} is to be thought of as ``sourced'' by and thus reflecting \emph{charge} distributed on $\mathbb{R}^p$. But generically this charge comes in indivisible units, hence is ``quantized'' in (for the moment) the na{\"i}ve sense of being \emph{discretized}, while the total flux as defined so far does \emph{a priori} vary continuously, cf. \cref{OrdinaryDeRhamAsNonabelianDeRham}.

In other words, there ought to be:
\begin{enumerate}
\item
a set of charges, to be denoted $H^1(X^d; \Omega \mathcal{A})$, depending contravariantly on the homotopy type of $X^d$, only, 
\item a map assigning to these charges  the total flux which they source, to be denoted
\begin{equation}
  \begin{tikzcd}
  H^1(X^d;\, \Omega \mathcal{A})
  \ar[
    rr,
    "{
      \mathrm{ch}^{\mathcal{A}}_{X^d}
    }"
  ]
  &&
  H^1_{\mathrm{dR}}\big(
    X^d;\,
    \mathfrak{a}
  \big)
  \end{tikzcd}
\end{equation}
and being natural in $X^d$.
\end{enumerate}

The first demand on charges in \cref{TotalCharge} may be satisfied by choosing any (connected, pointed) topological space $\mathcal{A}$ as a \emph{classifying space} of charges and taking the set of charges to be the homotopy classes $\pi_0(-)$ of the continuous maps into this classifying space
\begin{equation}
  \label{TheNonabelianCohomologySet}
  H^1(X^d; \Omega \mathcal{A})
  :=
  \pi_0\, \mathrm{Map}\big(
    X^d, \mathcal{A}
  \big)
  \mathrlap{\,.}
\end{equation}
This is the \emph{nonabelian cohomology} with coefficients in the \emph{loop $\infty$-group} of $\mathcal{A}$. \footnote{
  Evidently, the right hand side of \cref{TheNonabelianCohomologySet} immediately applies also to non-connected $\mathcal{A}$, in which generality we write $H^0(X; \mathcal{A}) := \pi_0\, \mathrm{Map}(X^d, \mathcal{A})$. But the restriction to connected $\mathcal{A}$ and the understanding of \cref{TheNonabelianCohomologySet} as cohomology in degree 1 with coefficients in $\Omega \mathcal{A}$ makes a little more transparent the nature of the nonabelian character map in \cref{OnTheCharacterMap} below, where the \emph{Whitehead $L_\infty$-algebra} $\mathfrak{l}\mathcal{A}$ (\cref{OnWhitheadBracketLInfinityAlgebras}) is really the rational incarnation of $\Omega \mathcal{A}$. Of course, one could alternatively instead consider Whitehead $L_\infty$-\emph{algebroids}. In any case, this means that the connectedness of $\mathcal{A}$ is at most of notational relevance.}

\begin{example}[Charges in Generalized and Nonabelian Cohomology] 
\nopagebreak \hfill
\begin{description}[beginpenalty=10000]
\item[Ordinary cohomology]
 With $\mathcal{A} \defneq B^n \mathbb{Z}$ we obtain \emph{ordinary cohomology}
 \begin{equation}
   H^1(X; \Omega B^n \mathbb{Z})
   \simeq
   H^n(X; \mathbb{Z})
   \mathrlap{\,,}
 \end{equation}
 such as, for instance,
 \begin{equation}
   H^1\big(
     X; \Omega (B^2 \mathbb{Z} \times
      B^2 \mathbb{Z})
   \big)
   \simeq
   H^2(X; \mathbb{Z})^2
   \mathrlap{\,.}
 \end{equation}
 This is the set of ordinary \emph{electromagnetic charges} in $1+3D$, such as carried by any integer numbers of electrons and magnetic monopoles (which are 0-branes), respectively:
 \begin{equation}
   \label{2CohomologyOfMonopoleCharges}
   H^2\big(
     \mathbb{R}^3 
     \setminus
     \mathbb{R}^0
     ;\,
     \mathbb{Z}
   \big)
   \simeq
   H^2\big(S^2\big)
   \simeq
   \mathbb{Z}\,.
 \end{equation}

\item[Topological K-Cohomology]
 With $\mathcal{A} \defneq \mathbb{Z} \times B \mathrm{U}$ we obtain (complex topological) \emph{K-theory}
 \begin{equation}
   H^0\big(
     X;
     \mathbb{Z}
     \times
     B \mathrm{U}
   \big)
   \simeq
   \prod_{n \in \mathbb{Z}}
   \,
   H^1\big(X; 
     \Omega B \mathrm{U}
   \big)
   \simeq
   K(X)
   \mathrlap{\,,}
 \end{equation}
 which witnesses integer numbers of $\mathrm{D}p$-branes for even $p = 2k$:
 \begin{equation}
   \label{KOfBlackDpBrane}
   K\big(
     \mathbb{R}^9 \setminus \mathbb{R}^{p}
   \big)
   \simeq
   K\big(
     S^{8-2k}
   \big)
   \simeq
   \mathbb{Z}
   \mathrlap{\,.}
 \end{equation}

 \item[Cohomotopy]
 With $\mathcal{A} \defneq S^n$ we obtain \emph{Cohomotopy}
 \begin{equation}
   H^1(X; \Omega S^n)
   \simeq
   \pi^n(X)
   \mathrlap{\,,}
 \end{equation}
 which for $n = 4$ witnesses integer numbers of M5-branes 
 \begin{equation}
   \pi^4\big(
     \mathbb{R}^{10}
     \setminus
     \mathbb{R}^{5}
   \big)
   \simeq
   \pi^4(S^4)
   \simeq
   \mathbb{Z}
   \mathrlap{\,,}
 \end{equation}
 and integer numbers of M2-branes (together with some fractional M2-charges):
 \begin{equation}
   \pi^4\big(
     \mathbb{R}^{10}
     \setminus
     \mathbb{R}^{2}
   \big)
   \simeq
   \pi^4(S^7)
   \simeq
   \mathbb{Z}
   \times \mathbb{Z}_{/12}
   \mathrlap{\,.}
 \end{equation}
\end{description}
\end{example}

\subsubsection{The Character Map}
\label{OnTheCharacterMap}
\footnote{
  The nonabelian character map is due to \parencites[\S IV]{FSS23-Char}{SS25-TEC}, following \cite[Def. 3.7]{FSS22-Twistorial}, which in turn follows \parencites[\S 4]{FSS15-M5WZW} (the case of 4-Cohomotopy) and \cite{FSS20-H} (the case of tangentially twisted 4-Cohomotopy, reviewed in \cite[\S 12]{FSS23-Char}).
}

The second demand in \cref{TotalCharge} is then satisfied by the cohomology operation which is represented by the rationalization
unit map $\eta^{\mathbb{Q}}_{\mathcal{A}}$ followed by extension of scalars from $\mathbb{Q}$ to $\mathbb{R}$:
\begin{equation}
  \label{RationalizationUnit}
  \begin{tikzcd}[column sep=35pt]
    \mathcal{A}
    \ar[
      r,
      "{
        \eta^{\mathbb{Q}}_{\mathcal{A}}
      }"
    ]
    \ar[
      rr,
      downhorup,
      "{ 
        \eta^{\mathbb{R}}_{\mathcal{A}}
      }"{description}
    ]
    &
    L^{\mathbb{Q}} \mathcal{A}
    \ar[
      r,
      "{ \mathbb{R}\otimes_{\mathbb{Q}}(-) }"
    ]
    &
    L^{\mathbb{R}} \mathcal{A}
    \mathrlap{\,,}
\end{tikzcd}
\end{equation}
which lands in nonabelian de Rham cohomology with $\mathfrak{a}$-coefficients iff $\mathcal{A}$ is \emph{admissible} in that its Whitehead $L_\infty$-algebra is isomorphic to the Gau{\ss}-law $L_\infty$-algebra:
\begin{equation}
  \label{AdmissibilityCondition}
  \mathfrak{l}\mathcal{A}
  \simeq
  \mathfrak{a}
  \mathrlap{\,.}
\end{equation}
In this case a \emph{nonabelian de Rham theorem} \cite[\S 6]{FSS23-Char} gives the identification on the right of the following diagram \cref{TheCharacterMap} and hence defines:
\begin{definition}
The \emph{nonabelian character map} $\mathrm{ch}^{\mathcal{A}}$, from nonabelian cohomology with coefficients in $\Omega \mathcal{A}$ to nonabelian de Rham cohomology with coefficients in $\mathfrak{l}\mathcal{A}$, is the nonabelian cohomology operation induced by $\mathbb{R}$-rationalization $\eta^{\mathbb{R}}_{\mathcal{A}}$ \cref{RationalizationUnit} of the classifying space $\mathcal{A}$, as follows:
\begin{equation}
  \label{TheCharacterMap}
  \begin{tikzcd}[
    row sep=-3pt, column sep=large
  ]
  \substack{
    \scalebox{.7}{\bf\color{purple}Charge \color{darkblue}in $\Omega \mathcal{A}$-valued}
    \\
    \scalebox{.7}{\bf\color{darkblue}nonabelian cohomology}
  }
  \ar[
    rr,
    phantom,
    "{
      \scalebox{.7}{\bf\color{darkgreen}character}
    }"
  ]
  &&
  \substack{
    \scalebox{.7}{\bf\color{purple} Total flux \color{darkblue}in $\mathfrak{l}\mathcal{A}$-valued}
    \\
    \scalebox{.7}{\bf\color{darkblue}de Rham cohomology}
  }
  \\
    H^1\big(
      X; \Omega \mathcal{A}
    \big)
    \ar[
      rr,
      "{
       \mathrm{ch}^{\mathcal{A}}_X
      }"
    ]
    &&
    H^1_{\mathrm{dR}}\big(
      X;
      \mathfrak{l}\mathcal{A}
    \big)
    \\
    \rotatebox[origin=c]{-90}{$:=$}
    &&
    \\
    \pi_0
    \,
    \mathrm{Map}(X, \mathcal{A})
    \ar[
      rr,
      "{
        (\eta^\mathbb{R}_{\mathcal{A}})_\ast
      }"
    ]
    &&
    \pi_0
    \,
    \mathrm{Map}\big(
      X, 
      L^{\mathbb{R}}\!\mathcal{A}
    \big)    
    \mathrlap{\,.}
    \ar[
      uu,
      "{ \sim }"{sloped, swap, pos=.35}
    ]
  \end{tikzcd}
\end{equation}
\end{definition}
\begin{example}[Character maps] {\ }
\begin{enumerate}
  \item On ordinary cohomology, the 
  nonabelian character map \cref{TheCharacterMap} reduces to the ordinary \emph{de Rham homomorphism}:
  \begin{equation}
    \label{DeRhamHomomorphismAsCharacter}
    \begin{tikzcd}
      \mathrm{ch}^{B^n \mathbb{Z}}
      :
      H^n(X; \mathbb{Z})
      \ar[r]
      &
      H^n_{\mathrm{dR}}(X)
      \mathrlap{\,.}
    \end{tikzcd}
  \end{equation}
  \item On K-theory, the 
  nonabelian character map \cref{TheCharacterMap} reduces to the ordinary \emph{Chern character}:
  \begin{equation}
    \begin{tikzcd}
      \mathrm{ch}^{\mathbb{Z} \times B \mathrm{U}}
      :
      K(X)
      \ar[r]
      &
      \bigoplus_{k \in \mathbb{N}}  
      H^{2k}_{\mathrm{dR}}(X)
      \mathrlap{\,.}
    \end{tikzcd}
  \end{equation}
  \item On 4-Cohomotopy the nonabelian character is a novel \emph{cohomotopical character map} (whose properties are reviewed in \cite[\S 12]{FSS23-Char}):
  \begin{equation}
    \begin{tikzcd}
      \mathrm{ch}^{S^4}
      :
      \pi^4(X)
      \ar[r]
      &
      H^1_{\mathrm{dR}}\big(
        X;
        \mathfrak{l}S^4
      \big)
      \mathrlap{\,.}
    \end{tikzcd}
  \end{equation}
\end{enumerate}
\end{example}

\begin{remark}[The set of choices of flux quantizations]
  \label[remark]{SpaceOfChoicesOfFluxQuantizationLaws}
  Given a characteristic $L_\infty$-algebra $\mathfrak{a}$ there exist either \emph{none} or \emph{infinitely many} admissible choices \cref{AdmissibilityCondition} of classifying spaces $\mathcal{A}$:
  \begin{enumerate}
    \item If $\mathcal{A}$ is an admissible choice, then also $\mathcal{A} \times B K$ is admissible for all \emph{finite groups} $K \in \mathrm{Grp}(\mathrm{FinSet})$. This is just one way of many for spaces $\mathcal{A}$ to have the same \emph{rational homotopy type} $\mathfrak{l}\mathcal{A}$. 
    
    In other words: A choice of flux quantization law $\mathcal{A}$ is a choice of \emph{torsion completion} of the charges of the theory, and there are infinitely many ways that torsion may appear.

    \item The above assumption that $\mathcal{A}$ is simply connected excludes flux quantization of any theory involving 1-form flux densities. But we may assume more generally that $\mathcal{A}$ has \emph{nilpotent} fundamental group acting nilpotently on all higher homotopy groups (cf. \cite{Hilton1982,nLab:NilpotentTopSpace}), since in that case \cref{WhiteheadLInfinityAlgebras} still holds. In this generality, the Whitehead $L_\infty$-algebras $\mathfrak{l}\mathcal{A}$ that appear are precisely those which are (connected and of finite type and) \emph{nilpotent} in that its iterated brackets eventually vanish.

    Beware that this class of $L_\infty$-algebras (playing the role of coefficients of nonlinear flux densities) is essentially disjoint (away from the abelian cases) from that one is interested in when considering $L_\infty$-algebras as coefficients for gauge potentials (cf. \cref{RoleOfLInfinityAlgebras}).
\end{enumerate}
  
\end{remark}

\subsubsection{Total Flux Quantization}
\footnote{
  For more on total flux quantization, cf. \parencites[\S 3.1-2]{SS25-Flux}[\S 2]{BaSS26-MString}.
}

In summary so far, the equations of motion for higher flux densities \cref{ShorthandOfTheHigherMaxwellEquations}
are to be complemented by a \emph{flux quantization law} represented by a classifying space $\mathcal{A}$, subject to the admissibility condition that its Whitehead $L_\infty$-algebra (\cref{WhiteheadLInfinityAlgebras}) reproduces the Gau{\ss} law $L_\infty$-algebra $\mathfrak{a}$ \cref{SolutionSpaceAsClosedForms}: $\mathfrak{l}\mathcal{A} \simeq \mathfrak{a}$.

Given such a flux quantization law, the solution space of flux densities (\cref{OnShellFluxDensities}) is to be both restricted and extended:
\begin{itemize}
  \item 
  to the extent that the $\mathcal{A}$-character map $\mathrm{ch}^{\mathcal{A}}$ \cref{TheCharacterMap} is not \emph{surjective},
  \\
  the physical flux densities are just those whose total flux $[\vec B]$ is in its image;
  \item 
  to the extent that the character map is not \emph{injective}, 
  \\
  the lift of the total flux $[\vec B]$ to a charge preimage $\rchi$ is to be part of the physical field content, alongside the flux densities:
\end{itemize}
\begin{equation}
  \label{FluxesAndTheirCharges}
  \begin{tikzcd}[row sep=-2pt, column sep=7pt]
    &&
    H^1(X^d; \Omega \mathcal{A})
    \ar[
      dd,
      "{
        \mathrm{ch}^{\mathcal{A}}
      }"{swap}
    ]
    &\ni&
    \rchi
    \ar[
      dd,
      |->,
      shorten=4pt
    ]
    \\[10pt]
    \\[10pt]
    \Omega^1_{\mathrm{dR}}\big(
      X; \mathfrak{l}\mathcal{A}
    \big)_{\mathrm{cl}}
    \ar[
      rr,
      "{
        [-]
      }"
    ]
    &\phantom{-}&
    H^1_{\mathrm{dR}}\big(
      X;
      \mathfrak{l}\mathcal{A}
    \big)
    &\ni&
    \mathrm{ch}(\rchi)
    \\
    \rotatebox[origin=c]{+90}{$\in$}
    &&
    \rotatebox[origin=c]{+90}{$\in$}
    \\
    \vec B 
    \ar[
      rr,
      |->,
      shorten=10pt
    ]
      &&
    {[\vec B]}
    \ar[
     uurr,
     equals,
     bend right=25
    ]
  \end{tikzcd}
\end{equation}

For example, 
\begin{enumerate}
\item when $\mathcal{A} \defneq B^2 \mathbb{Z}$, so that the character map is the de Rham homomorphism \cref{DeRhamHomomorphismAsCharacter}, then the electromagnetic flux densities $F_2$ whose total flux $[F_2]$ is in the image of the character maps are the \emph{integral forms}.

This \emph{total flux quantization} of the magnetic field is observed in experiment: When confining magnetic flux density on a slab $X^3 = \mathbb{R}^2_{\cpt} \times [-\epsilon,+\epsilon]$ of type II superconductor material, one observes an integer number
\begin{equation}
  n \in \mathbb{Z}
  \simeq
  H^2\big(
    \mathbb{R}^2_{\cpt}
    \times [-\epsilon, +\epsilon]
    ;\mathbb{Z}
  \big)
\end{equation}
of \emph{magnetic flux quanta} penetrating the superconductor, visible via electron microscopy in the form of (an integer number of) \emph{Abrikosov vortices} that these induce in the material's electron sea. 
\end{enumerate}

So, a choice of flux quantization law $\mathcal{A}$ for given on-shell fluxes characterized by $\mathfrak{a} \simeq \mathfrak{l}\mathcal{A}$ is part of the definition of the higher gauge field theory.
In fact, below in \cref{OnTheCompletedPhaseSpace} we see that: 
\begin{standout}
The choice of $\mathcal{A}$ \emph{completes} the definition of the higher gauge field content.
\end{standout}

Therefore, the choice of $\mathcal{A}$ is part of the specification of the \emph{physical model}. If the physics to be described by that model is thought to be fixed, then the choice of $\mathcal{A}$ is a \emph{hypothesis} about the correct mathematical description of that physics, each such choice entailing \emph{predictions} about the global behavior of the given physical fields.

Some choices of $\mathcal{A}$ are mathematically more canonical than others, in that they are minimal in terms of cell complex structure, as already reflected in our notation above for the corresponding Gau{\ss}-law $L_\infty$-algebras:

\begin{example}[Some flux quantization hypotheses]
\label[example]{SomeFluxQuantizationHypotheses}
\nopagebreak \hfill
\begin{description}[
  font=\itshape,
  ref={\bf (\roman*)} in \cref{SomeFluxQuantizationHypotheses}
]
  \item[Dirac charge quantization]
  \label{DiracChargQuantization}
  For the vacuum Maxwell field \cref{DualitySymmetricMaxwellEquationOfMotion} characterized by \cref{CharacteristicOfMaxwellField}
  \begin{equation}
    \mathfrak{a} = 
    \mathfrak{l}\big(
      B^2 \mathbb{Z}
      \times
      B^2 \mathbb{Z}
    \big)
    \mathrlap{\,,}
  \end{equation}
  the evident choice of flux quantization is 
  \begin{equation}
    \begin{aligned}
    \mathcal{A}
    &:=
    B^2 \mathbb{Z}
    \times 
    B^2 \mathbb{Z}
    \\
    &
    \simeq
    B\mathrm{U}(1)
    \times
    B \mathrm{U}(1)
    \end{aligned}
  \end{equation}
  implying both magnetic as well as electric charges in integral cohomology. This is the duality symmetric form of traditional Dirac charge quantization, considered in this form in \parencites{FreedMooreSegal2007}[Rem. 2.3]{BeckerBeniniSchenkelSzabo2017}[Def. 4.1]{LazaroiuShahbazi2022a}[Def. 4.3]{LazaroiuShahbazi2022b}.

  \item[Hypothesis K]\label[example]{HypothesisK}
  For the type IIA RR-field \cref{EoMOfIIARRField} characterized by 
  \cref{CharacteristicOfIIARRField},
  \begin{equation}
    \mathfrak{a} = \mathfrak{l}(\mathbb{Z} \times B \mathrm{U})
    \mathrlap{\,,}
  \end{equation}
  the evident choice of flux quantization is 
  \begin{equation}
    \label{ClassifyingSpaceForComplexKtheory}
    \mathcal{A} := 
    \mathbb{Z} \times B \mathrm{U}
    \mathrlap{
      \;\simeq
      \mathrm{KU}_0
      \,,
    }
  \end{equation}
  implying charges in topological K-theory.

  The hypothesis that this is the ``correct'' choice of flux quantization of type IIA RR-flux \cite{MooreWitten2000,Grady2022} has received much attention in the past (cf. \parencites[\S 4.1]{SS25-Flux}{nLab:DBraneChargeQuantizationInKTheory}) and could be called \emph{Hypothesis K} to distinguish it from other hypotheses on the global nature of RR-fields (cf. \cite[Rem. 4.1]{SS23-Defect}). 

  \item[Hypothesis h]\label[example]{Hypothesish}
  For the 5D SuGra gauge field \cref{EoMOf5DMCS} characterized by \cref{CharacteristicOf5DMCSField},
  \begin{equation}
    \mathfrak{a} = \mathfrak{l}S^2
    \mathrlap{,}
  \end{equation}
  the evident choice of flux quantization is 
  \begin{equation}
    \mathcal{A} := S^2
    \mathrlap{,}
  \end{equation}
  implying charges in 2-Cohomotopy.

  Since 5D Maxwell-Chern-Simons theory has a constrained KK-reduction to abelian 3D Chern-Simons theory, one may regard this as a hypothesis about the flux quantization also of the latter \cite[\S 3]{SS25-WilsonLoops}, and as such we have called this \emph{Hypothesis h} (with lower-case ``h'') \cite{SS25-FQH}. This flux quantization turns out to imply fine detail of quantum Chern-Simons theory \cite[\S 3]{SS25-FQH} and makes novel predictions on where and how to find (nonabelian) anyons in fractional quantum (anomalous) Hall systems \parencites[Fig. D]{SS25-FQH}{SS25-FQAH}{SS25-CrystallineChern}[\S 4.2]{SS25-Orient}.

  \item[Hypothesis H]\label[example]{HypothesisH}
  For the C-field in 11D SuGra  \cref{EoMOfCFieldIn11DSugra} characterized by \cref{CharacteristicOfCField},
  \begin{equation}
    \mathfrak{a} = \mathfrak{l}S^4
    \mathrlap{\,,}
  \end{equation}
  the evident choice of flux quantization is 
  \begin{equation}
    \mathcal{A} := S^4
    \mathrlap{\,,}
  \end{equation}
  implying charges in 4-Cohomotopy.

  The hypothesis that this is at least close to the ``correct'' choice for the global completion of 11D SuGra to ``M-theory'' we have called \emph{Hypothesis H} and shown to imply a whole list of subtle topological effects expected in ``M-theory'' (cf. \parencites{FSS20-H}{FSS21-Hopf}{FSS21-StrStruc}{SS21-M5Anomaly}).
\end{description}
\end{example}

\begin{remark}
It is important to note that the above nonabelian flux quantization process concerns quantization not just of the magnetic charges but compatibly so also of the electric charges. Quantization of electric charges has not received due attention elsewhere, and it is with the quantization of the electric charges included that the flux quantization laws generically become non-abelian.
\end{remark}

But total flux quantization is not yet the whole story: The full higher gauge field involves also a \emph{gauge potential} $\widehat{A}$ (experimentally seen in Aharonov-Bohm effects) which witnesses a local form of flux quantization. But this turns out to be induced from the same flux quantization choice $\mathcal{A}$, which is what we turn to next.

\subsection{Completed Phase Space}
\label{OnTheCompletedPhaseSpace}
\footnote{
For more on the completed phase space see
\parencites{SS24-Phase}[\S 3.1]{SS25-Flux}.
}

In this section, we need to invoke some geometric homotopy theory (higher topos theory, cf. \cite{Lurie2009}) in order to describe the geometry of the phase space of higher gauge fields, globally completed by a flux quantization law $\mathcal{A}$ according to \cref{TotalFluxQuantization}. 

But a central result (\cref{ShapeOfAdiff} below) says that the ``topological sector'' of this phase space (its shape) is equivalently just (the homotopy type of) $\mathcal{A}$! Since this is all that enters the discussion of the quantization of the topological sector in \cref{QuantizationInTopologicalSector} below, the reader willing to take this step for granted may want to skip ahead to \cref{QuantizationInTopologicalSector}.

\subsubsection{Smooth Sets of Flux Densities}
\footnote{
  For more on the topic of smooth sets in physics see \parencites{Schreiber2025,GS25-FieldsI,IbortMas2025}
}

So far, we have been considering the \emph{set} $\mathrm{LocSol} \simeq \Omega^1_{\mathrm{dR}}\big(X^d; \mathfrak{a}\big)_{\mathrm{cl}}$ of on-shell flux densities (\cref{OnShellFluxDensities}), but now we need to regard this as a kind of \emph{smooth space} amenable to differential geometry. In fact, we will need to do differential geometry on a kind of \emph{moduli space} of on-shell flux densities.
But it is an open secret in mathematical physics that the differential geometry of spaces of fields of almost any kind falls outside the scope of traditional differential geometry: These spaces are generically not even infinite-dimensional manifolds.

Therefore, we need to invoke a better differential geometry where smooth spaces of flux densities do naturally exist. This is the \emph{cohesive topos theory} of \emph{smooth sets}:

\begin{definition}
\label[definition]{SiteOfCartesianSpaces}
Let $\mathrm{CartSp}$ denote the category whose objects are the Cartesian spaces $\mathbb{R}^n$ for $n \in \mathbb{N}$ and whose morphisms are the smooth functions between these. Regard this as a site via the \emph{coverage} (Grothendieck pre-topology) of \emph{differentiably good open covers} (covers by open subsets all whose non-empty finite intersections are diffeomorphic to an $\mathbb{R}^n$).
\end{definition}

Whatever a \emph{smooth set} $X$ is, it should be \emph{probeable} by these $\mathbb{R}^n$, in that we should know of $X$:
\begin{enumerate}
\item the sets of smooth \emph{probes} (``plots'') $\mathrm{Plt}(\mathbb{R}^n,X) := \big\{\begin{tikzcd}[sep=small]\mathbb{R}^n \ar[r, dashed, shorten=-2pt] & X \end{tikzcd}\big\}$, 
\item the (contravariant) functoriality of these sets under precomposition with ordinary smooth functions
$\begin{tikzcd}[sep=small] \mathbb{R}^{n_1} \ar[r, "f"] & \mathbb{R}^{n_2}\end{tikzcd}$\;:
$
  \begin{tikzcd}
    \mathrm{Plt}(\mathbb{R}^{n_2}, X)
    \ar[
      r,
      "{
        f^\ast
      }"
    ]
    &
    \mathrm{Plt}(\mathbb{R}^{n_1}, X)
    \mathrlap{\,,}
  \end{tikzcd}
$
\item and 
the fact that probes by some $\mathbb{R}^n$ may be glued together from compatible probes by covering open subsets.
\end{enumerate}

Jointly, this says that the sets of probes/plots of a generalized smooth space $X$ should form a \emph{sheaf of sets} on $\mathrm{CartSp}$ -- and that's actually all we need to know about such $X$:
\begin{definition}
  \emph{Smooth sets} are the sheaves on $\mathrm{CartSp}$:
  \begin{equation}
    \label{CategoryOfSmoothSets}
    \mathrm{SmthSet}
    :=
    \mathrm{Sh}(\mathrm{CartSp})
    \mathrlap{\,.}
  \end{equation}
  The morphisms $\begin{tikzcd}[sep=small]X \ar[r] & Y\end{tikzcd}$ of these sheaves are the \emph{smooth maps} between smooth sets.
\end{definition}

For example:
\begin{enumerate}
  \item Smooth manifolds $X$ are smooth sets via the ordinary smooth functions into them: $\mathrm{Plt}(\mathbb{R}^n, X) := C^\infty(\mathbb{R}^n, X)$.

  \item In particular, the $\mathbb{R}^n$ themselves become smooth sets this way (\emph{Yoneda embedding}) and the given plots $\begin{tikzcd}[sep=small] \mathbb{R}^n \ar[r, dashed, shorten=-2pt] & X\end{tikzcd}$ of any smooth set are equivalently the smooth maps $\begin{tikzcd}[sep=small] \mathbb{R}^n \ar[r, shorten=-2pt] & X\end{tikzcd}$ (\emph{Yoneda lemma}).

  \item For smooth sets $X_1$, $X_2$, their \emph{product smooth set} $X_1 \times X_2$ is given by $\mathrm{Plt}(\mathbb{R}^n, X_1 \times X_2) := \mathrm{Plt}(\mathbb{R}^n, X_1) \times \mathrm{Plt}(\mathbb{R}^n, X_2)$.

  \item For any pair of smooth sets $X$, $A$ the set of smooth maps between them becomes a smooth set $\mathrm{Map}(X,A)$ via $\mathrm{Plt}\big(\mathbb{R}^n, \mathrm{Map}(X,A)\big) \!:=\! \big\{\begin{tikzcd}[sep=small] \mathbb{R}^n \times X \ar[r, shorten=-2pt] & A
  \end{tikzcd}\big\}$.
\end{enumerate}
The last of these examples shows how typical spaces of fields are naturally smooth sets. 
This has evident variants, of which one of key interest to us is:
\begin{itemize}
  \item[{\bf (v)}]
  For an $L_\infty$-algebra $\mathfrak{a}$ \cref{CEAlgebraAsQuotientOfFreeDGCA} the \emph{moduli} smooth set of closed $\mathfrak{a}$-valued differential forms is 
  \begin{equation}
    \label{ModuliSpaceOfClosedaValuedForms}
    \begin{tikzcd}
     \Omega^1_{\mathrm{dR}}(-; \mathfrak{a})_{\mathrm{cl}}
     \in
     \mathrm{SmthSet}
     \mathrlap{\,}
    \end{tikzcd}
  \end{equation}
  whose $\mathbb{R}^n$-plots are just the closed $\mathfrak{a}$-valued forms on the probe space $\mathbb{R}^n$:
  \begin{equation}
    \label{PlotsOfModuliSpaceOfClosedaValuedForms}
    \mathrm{Plt}\big(
      \mathbb{R}^n
      ,\,
      \Omega^1_{\mathrm{dR}}(-;\mathfrak{a})_{\mathrm{cl}}
    \big)
    :=
    \Omega^1_{\mathrm{dR}}(\mathbb{R}^n;\mathfrak{a})_{\mathrm{cl}}
    \mathrlap{\,,}
  \end{equation}
\end{itemize}

The definition \cref{PlotsOfModuliSpaceOfClosedaValuedForms}
bootstraps, over Cartesian probe spaces, the idea that there is to be a \emph{universal} closed $\mathfrak{a}$-valued form $\vec F_{\mathrm{univ}}$ on the smooth moduli space
whose pullback along smooth classifying maps uniquely produces all others --- which then in fact follows generally, for $X$ any smooth manifold (and generally any smooth set $X$):
\begin{equation}
  \label{aValuedFormsAsMaps}
  \begin{tikzcd}[row sep=-2pt, column sep=0pt]
    \mathrm{Hom}\big(
      X, 
      \Omega^1_{\mathrm{dR}}(-;
        \mathfrak{a})_{\mathrm{cl}}
    \big)
    \ar[
      rr,
      "{
        \sim
      }"
    ]
    &&
    \Omega^1_{\mathrm{dR}}(X;
        \mathfrak{a})_{\mathrm{cl}}    
    \\
    \big(
    X \xrightarrow{\phi}
      \Omega^1_{\mathrm{dR}}(-;
        \mathfrak{a})_{\mathrm{cl}}
    \big)
    &\longmapsto&
    \phi^\ast \vec F_{\mathrm{univ}}
    \mathrlap{\,.}
  \end{tikzcd}
\end{equation}

\subsubsection{Deformation moduli of flux densities}
\footnote{
  The object $\shape \, \Omega^1_{\mathrm{dR}}(-;\mathfrak{a})_{\mathrm{cl}}$ was introduced in \parencites[\S 4.4.14.2]{Sc13-dcct}[(9.2)]{FSS23-Char}, more exposition is in \cite[p. 26]{SS25-Flux}.

  The ambient $\infty$-topos of smooth $\infty$-groupoids goes back to \parencites[\S A]{FSSt12-DiffClasses}[\S 3.3]{Sc13-dcct} following \cite{Dugger1998}. The definition is reviewed in \cite[\S 1]{FSS23-Char}, monograph accounts are \parencites[\S 1]{FSS23-Char}[\S 4]{SS25-Bun}.
}

We have seen deformation classes of closed $\mathfrak{a}$-valued forms in the form of concordances over the interval $\Delta^1_{\mathrm{top}} = [0,1]$ (in \cref{TotalFluxInNonaDRCohomology}). These deformations serve as \emph{coboundaries} between closed $\mathfrak{a}$-valued forms. But in a higher gauge theory with flux densities of higher degrees, there are in general also higher order such deformations/concordances, parameterized over the topological $k$-simplices $\Delta^k_{\mathrm{top}}$. The tower of these higher-order deformations enhances the smooth set
\cref{PlotsOfModuliSpaceOfClosedaValuedForms} of closed $\mathfrak{a}$-valued forms on $\mathbb{R}^n$ to a \emph{simplicial smooth set}
\begin{equation}
  \label{ShapeOfModuliOfClosedAValuedForms}
  \shape
  \,
  \mathbf{\Omega}^1_{\mathrm{dR}}
  (-;\mathfrak{a})_{\mathrm{cl}}
  \in
  \mathrm{SmthSet}_{\Delta}
  \mathrlap{\,,}
\end{equation}
whose simplicial sets of plots are
\begin{equation}
  \mathrm{Plt}\Big(
    \mathbb{R}^n
    \times 
    \Delta^k
    ,
    \shape
    \,
    \mathbf{\Omega}^1_{\mathrm{dR}}
    (-;\mathfrak{a})_{\mathrm{cl}}
  \Big)
  :=
  \Omega^1_{\mathrm{dR}}
  \Big(
    \mathbb{R}^n 
    \times 
    \Delta^k_{\mathrm{top}}
    ;\,
    \mathfrak{a}
  \Big)_{\mathrm{cl}}
  \mathrlap{\,.}
\end{equation}

Here, in evident generalization of \cref{CategoryOfSmoothSets}, \emph{simplicial smooth sets} are the objects of the sheaf topos
\begin{equation}
  \mathrm{SmthSet}_\Delta
  :=
  \mathrm{Sh}\big(
   \mathrm{CartSp}
   \times 
   \Delta
  \big)
\end{equation}
over the product site of $\mathrm{CartSp}$ (\cref{SiteOfCartesianSpaces}) with the simplex category (the latter equipped with the trivial Grothendieck topology), whose objects are formal products
\begin{equation}
  \mathbb{R}^n \times \Delta^k
  \in 
  \mathrm{CartSp}
  \times 
  \Delta
  \,,\;\;
  n,k \in \mathbb{N}
\end{equation}
of a Cartesian space probing smooth structure and of a cellular simplex probing higher homotopy structure.

We will regard as homotopy-invertible the simplicial smooth maps in the class 
\begin{equation}
  W := 
  \big\{
  \mbox{Local weak homotopy equivalences} 
  \big\}
  \subset
  \mathrm{Mor}\big(
    \mathrm{SmthSet}_{\Delta}
  \big)
  \mathrlap{\,,}
\end{equation}
whose simplicial maps of probes, when restricted to arbitrarily small probes (to \emph{stalks}), are simplicial weak homotopy equivalences. Doing so means to regard simplicial smooth sets as presentations of \emph{smooth $\infty$-groupoids} (cf. \parencites[\S A]{FSSt12-DiffClasses}[\S 1]{FSS23-Char}[\S 4]{SS25-Bun}):
\begin{equation}
  \label{SmoothInfinityGroupoids}
  \mathrm{SmthGrpd}_\infty
  :=
  L^W
  \,
  \mathrm{SmthSet}_{\Delta}
  \mathrlap{\,.}
\end{equation}

The plain smooth sets may be faithfully understood as \emph{simplicially constant} simplicial smooth sets, 
\begin{equation}
  \label{SimpliciallyConstantSimplicialSmoothSets}
  \begin{tikzcd}
    \mathrm{SmthSet}
    \ar[r, hook]
    &
    \mathrm{SmthSet}_{\Delta}
    \mathrlap{\,,}
  \end{tikzcd}
\end{equation}
and as such there is 
a canonical smooth simplicial inclusion morphism from the smooth moduli space of closed $\mathfrak{a}$-valued forms \cref{ModuliSpaceOfClosedaValuedForms} to this simplicial smooth moduli space \cref{ShapeOfModuliOfClosedAValuedForms} of their deformations:
\begin{equation}
  \label{ShapeUnitOnModuliOfClosedaValuedForms}
  \begin{tikzcd}
    \Omega^1_{\mathrm{dR}}(-;\mathfrak{a})
    \ar[
      r,
      hook,
      "{
        \eta^{\scalebox{.6}{$\shape$}}
      }"
    ]
    &
    \shape
    \,
    \Omega^1_{\mathrm{dR}}(-;\mathfrak{a})_{\mathrm{cl}}    
    \mathrlap{\,.}
  \end{tikzcd}
\end{equation}
This map witnesses how a fixed flux density, on the left, becomes deformable, on the right.
For example:
\begin{itemize}
  \item
  for $\mathfrak{a} = \mathfrak{l}B^n \mathbb{Z}$, then 
  \begin{equation}
    \shape
    \,
    \Omega^1_{\mathrm{dR}}(-;\mathfrak{a})_{\mathrm{cl}}    
    \simeq
    \mathrm{DK}\Big(
      \begin{tikzcd}[sep=15pt]
        0 
        \ar[r]
        &
        \Omega^0_{\mathrm{dR}}(-)
        \ar[r, "{\mathrm{d}}"]
        & 
        \cdots
        \ar[r, "\mathrm{d}"]
        &
        \Omega^n_{\mathrm{dR}}(-)_{\mathrm{cl}}
      \end{tikzcd}
    \Big)
  \end{equation}
  is equivalently, as a smooth $\infty$-groupoid \cref{SmoothInfinityGroupoids}, the simplicial incarnation of the $n$-shifted de Rham complex  (cf. \cite[Lem. 9.2]{FSS23-Char}).
\end{itemize}
Conversely this means that we may regard $\shape \, \Omega^1_{\mathrm{dR}}(-;\mathfrak{a})_{\mathrm{cl}}$ as the \emph{nonabelian de Rham complex} with  coefficients in the $L_\infty$-algebra $\mathfrak{a}$.

\subsubsection{The differential character map}
\footnote{
  The differential nonabelian character map \cref{DifferentialCharacterMap} is due to 
  \cite[Def. 9.2]{FSS23-Char} (a reinterpretation of the classical rationalization map in dg-algebraic rational homotopy theory reviewed as \cite[(5.17)]{FSS23-Char}).
  The expression \cref{ShapeAsSmoothPathGroupoid} for the shape of smooth $\infty$-groupoids is due to \cite{PavlovEtAl2024}.
}
\newline
\nopagebreak
Where smooth sets are simplicially constant simplicial smooth sets \cref{SimpliciallyConstantSimplicialSmoothSets}, on the other extreme
every topological space $A$ induces its \emph{fundamental $\infty$-groupoid} or \emph{shape} (cf. \parencites[\S 4.3.2]{SS25-Bun}[\S 9.1.1]{SS26-Orb})
\footnote{
  In \cref{ShapeOfTopologicalSpace} the functor $\gamma$ regards a simplicial smooth set as a smooth $\infty$-groupoid via the simplicial localization of  \cref{SmoothInfinityGroupoids}.
}
\begin{equation}
  \label{ShapeOfTopologicalSpace}
  \shape
  \,
  A
  \in 
  \begin{tikzcd}
    \mathrm{SmthSet}_{\Delta}
    \ar[r, "{ \gamma }"]
    &
    \mathrm{SmthGrpd}_{\infty}
    \mathrlap{\,,}
  \end{tikzcd}
\end{equation}
which is represented by a geometrically discrete (hence constant on $\mathrm{CartSp}$) simplicial smooth set, given by its \emph{singular simplicial complex}:
\begin{equation}
  \mathrm{Plt}\big(
    \mathbb{R}^n \!\times\! \Delta^k
    ,
    \shape\, X
  \big)
  =
  \mathrm{Hom}\big(
    \Delta^k_{\mathrm{top}},
    A
  \big)
  \mathrlap{\,.}
\end{equation}

The general formula for the \emph{shape} $\shape \mathbf{X}$ of a smooth $\infty$-groupoid $\mathbf{X}$ is
\begin{equation}
  \label{ShapeAsSmoothPathGroupoid}
  \mathrm{Plt}\big(
    \mathbb{R}^n \times \Delta^k,
    \,
    \shape \mathbf{X}
  \big)
  \simeq
  \mathrm{Hom}\big(
    \mathbb{R}^n \times
    \Delta^k_{\mathrm{top}},
    \mathbf{X}
  \big)
  \mathrlap{\,,}
\end{equation}
and despite the superficial dependence on the Cartesian probes in this formula, the result is, up to equivalence, a geometrically discrete $\infty$-groupoid (hence is ``pure shape'', not carrying geometry):
\begin{equation}
  \shape \mathbf{X}
  \in
  \begin{tikzcd}
    \mathrm{Grpd}_\infty
    \ar[r, hook]
    &
    \mathrm{SmthGrpd}_\infty\,.
  \end{tikzcd}
\end{equation}

In fact, all previous constructions related to a choice of classifying space ``$\mathcal{A}$'' (connected, with abelian fundamental group and of rational finite-type) only depend on its shape, whence we will understand from here on that $\mathcal{A}$ denotes the \emph{pure shape} of a classifying space:
\begin{equation}
  \label{DifferentialCharacter}
  \mathcal{A} := \shape \, A
  \in 
  \begin{tikzcd} 
    \mathrm{Grpd}_\infty
    \ar[r, hook]
    &
    \mathrm{SmthGrpd}_\infty
    \mathrlap{\,.}
  \end{tikzcd}
\end{equation}

Now, a reframing  of the \emph{fundamental theorem of dg-algebraic rational homotopy theory} (\cite[\S 5 \&  Def. 9.2]{FSS23-Char}) gives a \emph{differential nonabelian character} map of smooth $\infty$-groupoids
\begin{equation}
  \label{DifferentialCharacterMap}
  \begin{tikzcd} 
    \mathcal{A}
    \ar[
      rr,
      "{
        \mathbf{ch}^{\mathcal{A}}
      }"
    ]
    &&
    \shape\, 
    \Omega^1_{\mathrm{dR}}\big(
      -;
     \mathfrak{l}\mathcal{A}
    \big)_{\mathrm{cl}}
  \end{tikzcd}
\end{equation}
which is an equivalent incarnation in $\mathrm{SmthGrpd}_\infty$ of the rationalization map $\eta^{\mathbb{R}}_{\mathcal{A}}$ in $\mathrm{Grpd}_\infty$ that defines the nonabelian character map  \cref{TheCharacterMap}.

\subsubsection{The completed phase space}
\footnote{
  The construction of nonabelian differential cohomology is due to \cite[Def. 9.3]{FSS23-Char} and its interpretation as providing the completed phase space of higher gauge fields is made explicit in \cite[Def. 2.6]{SS24-Phase}.
}

With all this in hand, we may now refine the picture \cref{FluxesAndTheirCharges} of flux quantization from total fluxes (being nonabelian cohomology classes) to the actual flux densities. We have produced a pair of coincident maps of smooth $\infty$-groupoids
\begin{equation}
  \label{TheFluxQuantizationCoCone}
  \begin{tikzcd}[
    column sep=12pt,
    row sep=10pt
  ]
    &&
    \mathcal{A}
    \ar[
      dd,
      "{ 
        \mathbf{ch}^{\mathcal{A}} 
      }"
    ]
    \\
    \\
    \Omega^1_{\mathrm{dR}}(-;
      \mathfrak{a})_{\mathrm{cl}}
    \ar[
      rr,
      "{
        \eta^{\scalebox{.6}{$\shape$}}
      }"
    ]
    &&
    \shape\,
    \Omega^1_{\mathrm{dR}}(-;
      \mathfrak{a})_{\mathrm{cl}}
  \end{tikzcd}
  \mathrlap{
  \;\;\;\;
  \mbox{
    (where 
    $\mathfrak{a}
    \simeq \mathfrak{l}\mathcal{A}$),
  }
  }
\end{equation}
which describe universally how flux densities (on the bottom left) may correspond (on the bottom right) with $\mathcal{A}$-valued charges (on the top right).

Concretely, for $X^d$ a Cauchy surface and
\begin{itemize}
\item
$\begin{tikzcd} X^d \ar[r, "{\vec B}"] & \Omega^1_{\mathrm{dR}}(-;\mathfrak{a})_{\mathrm{cl}}\end{tikzcd}$ modulating
\cref{aValuedFormsAsMaps} on-shell flux densities $\vec B$,

\item
$\begin{tikzcd} X^d \ar[r, "{ \rchi }"] & \mathcal{A} \end{tikzcd}$ classifying a charge in nonabelian cohomology \cref{TheNonabelianCohomologySet},
\end{itemize}
their compatibility is expressed by homotopies $\widehat{A}$ filling the resulting square of smooth $\infty$-groupoids:
\begin{definition}[Global phase space and Nonabelian differential cohomology]
\label[definition]{TheNonabelianDifferentialCohomologyPullback}
Given equations of motion \cref{ShorthandOfTheHigherMaxwellEquations} for higher flux densities characterized by an $L_\infty$-algebra $\mathfrak{a}$ (\cref{WhiteheadLInfinityAlgebras}) and an admissible flux quantization law $\mathcal{A}$ \cref{AdmissibilityCondition} then a \emph{globally complete} on-shell gauge field configuration is a dashed cone of smooth $\infty$-groupoids over \cref{TheFluxQuantizationCoCone} with tip $X^d$:
\begin{equation}
  \hspace{.5cm}
  \begin{tikzcd}[
    row sep=14pt,
    column sep=15pt
  ]
    X^d
    \ar[
      rr, 
      dashed,
      "{ \rchi }",
      "{\ }"{swap, name=s}
    ]
    \ar[
      dd, 
      dashed,
      "{ \vec B }"{swap},
      "{\ }"{name=t}
    ]
    &&
    \mathcal{A}
    \ar[
      dd,
      "{ 
        \mathbf{ch}^{\mathcal{A}} 
      }"
    ]
    \ar[
      from=s,
      to=t,
      dashed,
      Rightarrow,
      "{ \widehat A }"
    ]
    \\
    \\
    \Omega^1_{\mathrm{dR}}(-;
      \mathfrak{a})_{\mathrm{cl}}
    \ar[
      rr,
      "{
        \eta^{\scalebox{.6}{$\shape$}}
      }"
    ]
    &&
    \shape\,
    \Omega^1_{\mathrm{dR}}(-;
      \mathfrak{a})_{\mathrm{cl}}
    \mathrlap{\,.}
  \end{tikzcd}
\end{equation}
Equivalently, 
if we denote the universal such cone by $\mathcal{A}_{\mathrm{diff}}$, then the choice of \emph{gauge potential} $\widehat{A}$ for given on-shell flux densities $\vec B$ is equivalently a dashed lift as shown here:
\begin{equation}
  \label{GaugeFieldsAsLifts}
  \hspace{-1cm}
  \begin{tikzcd}[
    column sep=15pt
  ]
    &[+5pt]&[+5pt]
    \mathcal{A}_{\mathrm{diff}}
    \ar[
      rr, 
      "{ \rchi_{\mathrm{univ}} }",
      "{\ }"{swap, name=s}
    ]
    \ar[
      dd, 
      "{ 
        \vec B_{\mathrm{univ}} 
      }"{description},
      "{\ }"{name=t}
    ]
    \ar[
      ddrr,
      phantom,
      "{ \lrcorner }"{pos=.02}
    ]
    &&
    \mathcal{A}
    \ar[
      dd,
      "{ 
        \mathbf{ch}^{\mathcal{A}} 
      }"
    ]
    \ar[
      from=s,
      to=t,
      Rightarrow,
      shorten >=3pt,
      "{ 
        {\widehat A}_{\mathrm{univ}}
      }"
    ]
    \\
    \\
    X^d
    \ar[
      rr,
      "{
        \vec B
      }"
    ]
    \ar[
      uurr,
      dashed,
      "{
        (\vec B, \rchi, \widehat A)
      }"
    ]
    &&
    \Omega^1_{\mathrm{dR}}(-;
      \mathfrak{a})_{\mathrm{cl}}
    \ar[
      rr,
      "{
        \eta^{\scalebox{.6}{$\shape$}}
      }"
    ]
    &&
    \shape\,
    \Omega^1_{\mathrm{dR}}(-;
      \mathfrak{a})_{\mathrm{cl}}
    \mathrlap{\,.}
  \end{tikzcd}
\end{equation}
The gauge (homotopy) equivalence classes of such data are the \emph{nonabelian differential cohomology} of $X^d$:
\begin{equation}
  \label{NonabelianDifferentialCohomology}
  H^1_{\mathrm{diff}}\big(
    X^d;
    \Omega\mathcal{A}
  \big)
  :=
  \pi_0
  \mathrm{Hom}\big(
    X^d
    ,
    \mathcal{A}_{\mathrm{diff}}
  \big)
  \mathrlap{\,.}
\end{equation}
\end{definition}

\begin{example}
\nopagebreak \hfill
\begin{enumerate}
\item
  For $\mathcal{A} = B^2 \mathbb{Z}$, the data \cref{GaugeFieldsAsLifts} are equivalently $\mathrm{U}(1)$-principal bundles with connection, the usual model for the electromagnetic field.
  
\item For $\mathcal{A} = B^3 \mathbb{Z}$, the data \cref{GaugeFieldsAsLifts} are equivalently \emph{$\mathrm{U}(1)$-bundle gerbes with connection}, the usual model for the B-field \cref{EoMForBField} \cite[Ex. 9.4]{FSS23-Char}.

\item For $\mathcal{A} = \mathbb{Z} \times B \mathrm{U}$,  \cref{GaugeFieldsAsLifts} are cocycles in ``canonical'' differential K-theory \cite[Ex. 9.2]{FSS23-Char}.

\item Generally, for $\mathcal{A} = E_n$ a stage in an \emph{$\Omega$-spectrum of spaces}, the nonabelian differential cohomology \cref{NonabelianDifferentialCohomology} reduces to ``canonical'' differential $E^n$-cohomology \cite[Ex. ]{FSS23-Char}.

\item For $\mathcal{A} = \shape S^n$, we say that \cref{NonabelianDifferentialCohomology} is (unstable/nonabelian) \emph{differential cohomotopy}.
\end{enumerate}
\end{example}

Accordingly, given an admissible flux quantization law $\mathcal{A}$, then 
\begin{equation}
  \label{ThePhaseSpace}
  \mathrm{PhsSpc}
  :=
  \mathrm{Map}\big(
    X^d,
    \mathcal{A}_{\mathrm{diff}}
  \big)
  \in
  \mathrm{SmthGrpd}_\infty
\end{equation}
is the \emph{phase space} of the higher gauge theory: the ``space'' (smooth $\infty$-groupoid) of globally completed on-shell field configurations. 

\begin{remark}[The higher gauge potentials]
\label[remark]{TheHigherGaugePotentials}
It is not obvious that the homotopy denoted $\widehat{A}$ in \cref{TheNonabelianDifferentialCohomologyPullback} locally corresponds to the usual (higher) gauge potentials seen in the physics literature, in cases such as in \cref{HigherGaugeFields,SomeFluxQuantizationHypotheses}.
That this is indeed the case is remarkable and is shown:
\begin{enumerate}
\item
for higher $\mathrm{U}(1)$-gauge fields in \cite[Ex. 9.4]{FSS23-Char},
\item
for the 11D supergravity C-field in \cite[Prop. 2.48]{GSS24-SuGra},
\item
and for the self-dual gauge field on the M5 in \cite[\S B]{SS25-Srni}.
\end{enumerate}
\end{remark}
The cases of the RR-field and of the 5D MCS field follow analogously.

In view of all these examples of higher gauge fields one is naturally left wondering about the common case of nonabelian Yang-Mills fields. Here we have the following interesting subtlety:

\begin{example}[The (non-)example of Yang-Mills fields]
\label[example]{TheNonExampleOfYangMillsFields}
Beware that Yang-Mills fields (cf. \cite{RudolphSchmidt2017,nLab:YangMillsTheory}) for a \emph{nonabelian} gauge Lie group $G$ (with Lie algebra $\mathfrak{g}$) such as $\mathrm{U}(n)$, $n \geq 2$, are \emph{not} an example of \emph{Maxwell-type higher gauge fields} in the sense of \cref{OnTheEquationsOfMotion} (as the term ``Maxwell-type'' suggests). This is because:
\begin{enumerate}
\item their flux densities 
$\big\{F^a \in \Omega^2_{\mathrm{dR}}\big(\widehat X\big) \big\}_{a = 1}^{\mathrm{dim}(\mathfrak{g})}$
in general are not globally defined on spacetime $X$ as required in \cref{TheFluxDensities}, but only on a good open cover $\begin{tikzcd}[sep=7pt]\widehat{X} \ar[r,->>, shorten=-2pt] & X, \end{tikzcd}$

\item
the Bianchi identity 
$
  \mathrm{d}
  \, 
  F^a
  =
  f^a_{b c} \, A^b \wedge F^c
$
(involving the structure constants $(f^a_{bc} \in \mathbb{R})_{a,b,c =1}^{\mathrm{dim}(\mathfrak{g})}$ of $\mathfrak{g}$)
is not a polynomial in the flux densities alone, as required in \cref{TheHigherMaxwellEquations,ShorthandOfTheHigherMaxwellEquations}, but depends also on the \emph{gauge potential} 1-forms $\big\{ A^a \in \Omega^1_{\mathrm{dR}}\big(\widehat{X}\big) \big\}$.
\end{enumerate}

Accordingly, there cannot be a classifying space $\mathcal{A}$ whose character map \cref{OnTheCharacterMap} witnesses the quantization of nonabelian Yang-Mills charges as it does exist for the higher Maxwell-type field species in \cref{SomeFluxQuantizationHypotheses}. 

But one can turn this around: There is an evident ``candidate'' classifying space, namely the traditional classifying space of $G$-principal bundles, denoted $\mathcal{A} := B G$ (cf. \parencites[Thm. 3.5.1]{RudolphSchmidt2017}[\S 3.3.1 \& Thm. 5.2.13]{SS25-Bun}), and one can ask which kinds of higher gauge fields are described by using that in \cref{TheNonabelianDifferentialCohomologyPullback}: 

Given that in the (co)limit of infinite rank  this goes to the classifying space of (reduced) topological K-theory, $\begin{tikzcd}[sep=17pt]B \mathrm{U}(n) \ar[r, "{ n \to \infty }"] & B \mathrm{U} \end{tikzcd}$ \cref{ClassifyingSpaceForComplexKtheory}, we already see that setting $\mathcal{A} := B G$ generally gives a differential refinement of the nonabelian cohomology theory sometimes called \emph{unstable K-theory} \cite{HamanakaKono2004}. The cocyles of such differential unstable K-theory have been worked out in \cite[Prop. 9.4]{FSS23-Char}, following \cite[p. 28]{HopkinsSinger2005}: They are indeed represented by principal connections, hence by Yang-Mills fields, but subject to a stronger equivalence relation than principal gauge isomorphism.

This stronger equivalence relation is of course that which in the stable case, $n \to \infty$, has been argued to reflect D-brane/antibrane pair creation/annihilation \cite[\S 3]{Witten1998-DBranesAndKTheory}.
This is an immediate consequence of the traditional \emph{Hypothesis K} which may remain underappreciated: If D-brane charge is really classified in K-theory, then the Yang-Mills fields commonly considered on coincident D-branes (following \cite[p. 7]{Witten1996-BoundStates}) are not as such actually physically observable, only their class in differential K-theory is (cf. \cite{BMSS19-GaugeEnhancement}).
\end{example}

\section{Quantization of their Topological Sector}
\label{QuantizationInTopologicalSector}

Now that we quantize, our ground field becomes the complex numbers.

\subsection{Topological Observables}
\label{OnTopologicalObservables}
\footnote{
  Our discussion of topological observables follows \cite[\S 3]{SS24-Obs} and \parencites[\S 2.5]{SS22-Conf}[\S 4]{CSS23-QuantumStates}.
}

In quantizing the topological sector of a globally completed (flux-quantized) higher gauge theory, we first consider the topological observables, in order to then pass to their quantization as an algebra of topological quantum observables, and to the determination of the topological quantum states.

\subsubsection{Ordinary topological observables}
\footnote{
  For the ordinary notion of observables as smooth functions on phase space see \cite{SimmsWoodhouse1976}. For ordinary topological observables as the 0-homology of phase space see \cite[\S 4.1.4]{SS25-WilsonLoops}.
}

Ordinary \emph{observables} $O$ are the smooth scalar functions on phase space \cref{ThePhaseSpace}:
\begin{equation}
  \label{OrdinaryObservable}
  \begin{tikzcd}
    \mathrm{PhsSpc}
    \ar[
      r,
      dashed,
      "{ O }"
    ]
    &
    \mathbb{C}\,.
  \end{tikzcd}
\end{equation}
Since $\begin{tikzcd}[sep=small]\mathbb{C} \in \mathrm{SmthSpc} \ar[r, hook] & \mathrm{SmthGrpd}_\infty\end{tikzcd}$ is 0-truncated, these are automatically \emph{gauge invariant}.

A \emph{topological observable} is one that is sensitive only to the global shape of the phase space, but not to its local geometry, hence is such a map \cref{OrdinaryObservable} that factors through the shape unit:
\begin{equation}
  \begin{tikzcd}
    \mathrm{PhsSpc}
    \ar[
      rr,
      uphordown,
      "{ O }"
    ]
    \ar[
      r,
      "{
        \eta^{\scalebox{.6}{$\shape$}}
      }"
    ]
    &
    \shape
    \,
    \mathrm{PhsSpc}
    \ar[
      r,
      dashed
    ]
    &
    \mathbb{C}
    \mathrlap{\,.}
  \end{tikzcd}
\end{equation}
As before, since $\mathbb{C}$ is 0-truncated, this factors furthermore through the 0-truncation of the shape of the phase space, which is the set of connected components $\pi_0(-)$:
\begin{equation}
  \begin{tikzcd}[column sep=large]
    \mathrm{PhsSpc}
    \ar[
      rrr,
      uphordown,
      "{ O }"
    ]
    \ar[
      r,
      "{
        \eta^{\scalebox{.6}{$\shape$}}
      }"
    ]
    &
    \shape
    \,
    \mathrm{PhsSpc}
    \ar[
      r,
      "{
        \eta^{[-]_0}
      }"
    ]
    &
    \pi_0 \shape\, \mathrm{PhsSpc}
    \ar[
      r,
      dashed
    ]
    &
    \mathbb{C}
    \mathrlap{\,.}
  \end{tikzcd}
\end{equation}
Moreover, realistic observables (those realizable in experiment) are nonvanishing only on a finite number of connected components of the phase space. Therefore the realistic topological observables form the 0th \emph{homology} of the shape of the phase space:
\begin{equation}
  \mathrm{TopObs}
  =
  H_0\big(
    \shape\, 
    \mathrm{PhsSpc}
    ;\,
    \mathbb{C}
  \big)
  \,.
\end{equation}

\subsubsection{Higher topological observables}
\label{OnHigherTopologicalObservables}
\footnote{
  The notion of higher observables as the higher homology of a higher stacky phase space appears in \parencites{CSS23-QuantumStates}[\S 4]{SS24-Obs} following \cite[\S 2.5]{SS22-Conf}.

  The homotopy theory of $\infty$-vector spaces modeled as simplicial chain complexes is established in \cite[Def. 3.2, Thm. 3.3]{SS23-EoS}.
}

Beyond its connected components, (the shape of) the phase space encodes information about higher global symmetries of the higher gauge fields, also known as ``categorified symmetries'' (cf. \cite{SchreiberSkoda2009}),  ``generalized symmetries'' (cf. \cite{GripaiosEtAl2024}) or ``higher-form symmetries'' (cf. \cite{PerezLona2025}), for more see \cite{Kaidi2026,nLab:GeneralizedGlobalSymmetry}. If we imagine (as has become common these days) that these are also observable in experiment, then we should say that we have an \emph{$\mathbb{N}$-graded} vector space of topological \emph{higher observables}, given by the higher homology groups of the shape of the phase space:
\begin{equation}
  \label{TheHigherTopologicalObservables}
  \mathrm{TopObs}_\bullet
  :=
  H_\bullet\big(
    \shape\, \mathrm{PhsSpc}
    ;\,
    \mathbb{C}
  \big)
  \mathrlap{\,.}
\end{equation}

We may further resolve this structure by passing from homology classes to chains. 
To that end, consider the category of \emph{simplicial chain complexes}
\begin{equation}
  \mathrm{Ch}\big(
    \mathrm{Vect}
  \big)_\Delta
  :=
  \mathrm{Sh}\big(
    \Delta,
    \mathrm{Ch}(\mathrm{Vect})
  \big)
\end{equation}
and force into homotopy equivalences the maps that induce quasi-isomorphisms on totalizations:
\begin{equation}
  \label{TotalQuasiIsomorphisms}
  W 
  :=
  \big\{
    \mbox{Total-quasi-isomorphisms}
  \big\}
  \subset
  \mathrm{Mor}\Big(
    \mathrm{Ch}\big(
      \mathrm{Vect}
    \big)_\Delta
  \Big)
  \mathrlap{\,,}
\end{equation}
to obtain 
the $\infty$-category of \emph{$\infty$-vector spaces} (equivalently: \emph{$H\mathbb{C}$-modules}, cf. \cite{SS26-KLoc}):
\begin{equation}
  \label{InfinityCategoryOfInfinityVectorSpaces}
  \mathrm{Vect}_\infty
  :=
  L^W \mathrm{Ch}\big(
    \mathbb{C}\mathrm{Vect}
 \big)_\Delta
 \mathrlap{\,.}
\end{equation}

Now, since simplicial chain complexes support a combinatorial simplicial model category structure with the above weak equivalences \cref{TotalQuasiIsomorphisms}, there is  a canonical ``\emph{tensoring}'' of $\infty$-vector spaces over $\infty$-groupoids
\begin{equation}
  \label{TensoringOfInfinityVectOverInfinityGrpd}
  \begin{tikzcd}
    \mathrm{Grpd}_\infty
    \times
    \mathrm{Vect}_\infty
    \ar[
      rr,
      "{ (-)\cdot(-) }"
    ]
    &&
    \mathrm{Vect}_\infty
  \end{tikzcd}
\end{equation}
which on the tensor unit
$
  \TensorUnit 
    \in 
  \mathrm{Vect}_\infty
$
sends $\infty$-groupoids to their $\mathbb{C}$-linearization or \emph{$\mathbb{C}$-motive} (cf. \parencites[\S 8.2]{Sc14-LinTypes}[Prop. 2.10]{SS25-Monadology}[\S 2.3]{SS26-KLoc}), represented by what classically is their singular chain complex $\mathbb{C}[\mathcal{X}]$:
\begin{equation}
  \mathcal{X}
  \in
  \mathrm{Grpd}_\infty
  \;\;\;\;
  \vdash
  \;\;\;\;
  \mathbb{C}[\mathcal{X}]
  :=
  \mathcal{X} \cdot \TensorUnit
  \mathrlap{\,.}
\end{equation}

Therefore, the \emph{motive of the shape of the phase space} is
\begin{equation}
  \label{MotiveOfShapeOfPhaseSpace}
  \mathbf{TopObs}
  =
  \mathbb{C}\big[
    \shape\, \mathrm{PhsSpc}
  \big]
  \in
  \mathrm{Vect}_\infty
\end{equation}
and the higher topological observables \cref{TheHigherTopologicalObservables} are its stable homotopy groups:
\begin{equation}
  \label{HigherTopObsAsHomotopyGroups}
  \mathrm{TopObs}_\bullet
  \simeq
  \pi_\bullet\big(
    \mathbf{TopObs}
  \big)
  \defneq
  \pi_\bullet\big(
    \mathbb{C}\big[
      \shape \mathrm{PhsSpc}
    \big]
  \big)
  \mathrlap{\,.}
\end{equation}

\subsubsection{The shape of phase space}
\footnote{
  For more on the shape of smooth $\infty$-groupoids see \parencites[\S 4.3]{SS25-Bun}. The smooth Oka principle (\cref{SmoothOkaPrinciple}) is  due to \parencites{PavlovEtAl2024}[\S 1.1.2]{SS25-Bun}
}

Hence, to proceed with analyzing the topological observables, we need a good handle on the shape of the completed phase space of flux-quantized higher gauge fields. 

\begin{lemma}[Smooth Oka principle]
  \label[proposition]{SmoothOkaPrinciple}
  For $\begin{tikzcd}[sep=small]X \in \mathrm{SmthMfd} \ar[r, hook] & \mathrm{SmthGrpd}_\infty \end{tikzcd}$ and any $\mathbf{A} \in \mathrm{SmthGrpd}_\infty$ we have natural equivalences
  \begin{equation}
    \label{SmoothOkaPrincipleEquivalence}
    \begin{aligned}
    \shape\, 
    \mathrm{Map}\big(
      X, \mathbf{A}
    \big)
    & \simeq
    \mathrm{Map}\big(
      \shape X, \shape \mathbf{A}
    \big)
    \\
    & \simeq
    \mathrm{Map}\big(
      X, \shape \mathbf{A}
    \big)
    \mathrlap{\,.}
    \end{aligned}
  \end{equation}
\end{lemma}

\begin{proposition}
\label[proposition]{ShapeOfAdiff}
The shape of the nonabelian differential cohomology moduli $\mathcal{A}_{\mathrm{diff}}$ \cref{GaugeFieldsAsLifts} is the classifying shape $\mathcal{A}$:
\begin{equation}
  \label{ShapeOfDifferentialModuliIsCalA}
  \shape \, \mathcal{A}_{\mathrm{diff}}
  \;\simeq\;
  \mathcal{A}
\end{equation}
and under this equivalence the shape unit is the universal charge map:
\begin{equation}
  \label{ShapeUnitOnADiffIsUniversalCharge}
  \eta
    ^{\scalebox{.6}{$\shape$}}
    _{ \mathcal{A}_{\mathrm{diff}} }
  \simeq
  \rchi_{\mathrm{univ}}
  \mathrlap{\,.}
\end{equation}
\end{proposition}
\begin{proof}
  Using that the shape modality is idempotent, whence $\shape \eta^{\scalebox{.6}{$\shape$}} \simeq \mathrm{id}$ and $\shape \mathbf{ch}^{\scalebox{.6}{$\mathcal{A}$}} \simeq \mathbf{ch}^{\scalebox{.6}{$\mathcal{A}$}}$, and that it preserves pullbacks over pure-shape objects we obtain the following pullback on the right
  \begin{equation}
    \shape
    \left(
    \begin{tikzcd}[column sep=18pt]
      \mathcal{A}_{\mathrm{diff}}
      \ar[
        r,
        "{ \rchi_{\mathrm{univ}} }"
      ]
      \ar[
        d,
        "{
          \vec B_{\mathrm{univ}}
        }"{swap}
      ]
      \ar[
        dr,
        phantom,
        "{ \lrcorner }"{pos=.1}
      ]
      &
      \mathcal{A}
      \ar[
        d,
        "{ \mathbf{ch}^{\mathcal{A}} }"
      ]
      \\
      \mathbf{\Omega}^1_{\mathrm{dR}}\big(
        -;
        \mathfrak{a}
      \big)_{\mathrm{cl}}
      \ar[
        r,
        shorten=-2pt,
        "{ \eta^{\scalebox{.6}{$\shape$}} }"
      ]
      &
      \shape
      \,
      \mathbf{\Omega}^1_{\mathrm{dR}}\big(
        -;
        \mathfrak{a}
      \big)_{\mathrm{cl}}
    \end{tikzcd}
    \right)
    \simeq\;\;\;
    \begin{tikzcd}[column sep=14pt]
      \shape\mathcal{A}_{\mathrm{diff}}
      \ar[
        r,
        "{ \shape \rchi_{\mathrm{univ}} }"
      ]
      \ar[
        d,
        "{
          \shape \vec B_{\mathrm{univ}}
        }"{swap}
      ]
      \ar[
        dr,
        phantom,
        "{ \lrcorner }"{pos=.1}
      ]
      &
      \mathcal{A}
      \ar[
        d,
        "{ \mathbf{ch}^{\mathcal{A}} }"
      ]
      \\
      \shape
      \mathbf{\Omega}^1_{\mathrm{dR}}\big(
        -;
        \mathfrak{a}
      \big)_{\mathrm{cl}}
      \ar[
        r,
        equals,
        shorten=-2pt
      ]
      &
      \shape
      \,
      \mathbf{\Omega}^1_{\mathrm{dR}}\big(
        -;
        \mathfrak{a}
      \big)_{\mathrm{cl}}
      \mathrlap{\,.}
    \end{tikzcd}
  \end{equation}
  Since equivalences are preserved under pullback, this shows that $\shape \rchi_{\mathrm{univ}}$ is an equivalence, proving the first statement. 

  Moreover, the universal property of the shape implies that a map to a pure shape object like $\mathcal{A}$ factors through the shape unit of its domain, so that 
  \begin{equation}
    \begin{tikzcd}
      \mathcal{A}_{\mathrm{diff}}
      \ar[
        r,
        "{
          \eta^{\scalebox{.6}{$\shape$}}
        }"
      ]
      \ar[
        rr,
        uphordown,
        "{
          \rchi_{\mathrm{univ}}
        }"
      ]
      &
      \shape \mathcal{A}_{\mathrm{diff}}
      \ar[r, "{ F }"]
      &
      \mathcal{A}
      \mathrlap{\,,}
    \end{tikzcd}
  \end{equation}
  for some pure shape map $F$.
  But with $\shape \rchi_{\mathrm{univ}}$ and $\shape \eta^{\scalebox{.6}{$\shape$}}$ being equivalences, also $\shape F \simeq F$ is an equivalence, which proves the second claim.
\end{proof}

\begin{corollary}
\label[corollary]{TheShapeOfPhaseSpace}
The shape of the phase space of $\mathcal{A}$-flux-quantized higher gauge fields is that of the mapping space of the Cauchy surface $X^d$ into $\mathcal{A}$:
\begin{equation}
  \begin{aligned}
  \shape
  \,
  \mathrm{PhsSpc}
  & 
  \underset{
    \mathclap{\scalebox{.7}{\cref{ThePhaseSpace}}
    }
  }{\defneq}
  \shape\, \mathrm{Map}\big(
    X^d,
    \mathcal{A}_{\mathrm{diff}}
  \big)
  \underset{
    \mathclap{\scalebox{.7}{\cref{SmoothOkaPrincipleEquivalence}}
    }
  }{\simeq}
  \mathrm{Map}\big(
    X^d,
    \shape\, \mathcal{A}_{\mathrm{diff}}
  \big)
  \\
  & 
  \underset{
    \mathclap{\scalebox{.7}{\cref{ShapeOfDifferentialModuliIsCalA}}
    }
  }{\simeq}
  \mathrm{Map}\big(
    X^d
    ,
    \mathcal{A}
  \big)
  \mathrlap{\,,}
  \end{aligned}
\end{equation}
and hence the higher topological observables \cref{HigherTopObsAsHomotopyGroups} are the homology groups of this mapping space (the homotopy groups of its ``motive'' \cref{MotiveOfShapeOfPhaseSpace}):
\begin{equation}
  \label{HigherTopOsAsHomotopyGroupsOfMaps}
  \mathrm{TopObs}_\bullet
  \simeq
  \pi_\bullet\Big(
    \mathbb{C}\big[
      \mathrm{Map}(X^d, \mathcal{A})
    \big]
  \Big)
  \mathrlap{\,.}
\end{equation}
\end{corollary}

\subsubsection{Solitonic fields}
\label{SolitonicFields}
\footnote{
  Discussion of solitonic charges classified by pointed maps is in \parencites[\S 2.2]{SS25-Flux}[\S A.2]{SS24-Obs}. The solitonic phase space is discussed in \cite[\S A.2 (91)]{SS25-WilsonLoops}.

  For the classical notion of \emph{one-point compactification} $(-)_{\cpt}$ cf. \parencites[pp. 150]{Kelley1975}[p. 199]{Bredon1993}{nLab:OnePointCompactification}.
}

Generally, the field configurations of topological interest are \emph{solitonic}, meaning that their flux densities and charges \emph{vanish at infinity}. 
To model this, take: 
\begin{enumerate}

\item the charge classifying space $\mathcal{A}$ to be equipped with a (``base'') point, to be thought of as the classifier for zero-charge:
\begin{equation}
  \label{TheZeroCharge}
  \begin{tikzcd}
    \{0\}
    \ar[r, "{ \iota_0 }"]
    &
    \mathcal{A}
    \mathrlap{\,,}
  \end{tikzcd}
\end{equation}

\item
the spatial domain $X^d_{\cpt}$ to be the \emph{one-point compactification}  of a manifold $X^d$:
\begin{equation}
  \label{OnePointCompactification}
  X^d_{\cpt}
  :=
  X^d/\mathrm{ends}
  \,,
\end{equation}
identifying the \emph{ends} of a non-compact space (such as the evident two ends $\pm \infty$ of the real line) with a \emph{point at infinity}
\begin{equation}
  \label{ThePointAtInfinity}
  \begin{tikzcd}
    \{\infty\}
    \ar[r, "{ \iota_{\infty} }"]
    &
    X^d_{\cpt}
    \mathrlap{\,,}
  \end{tikzcd}
\end{equation}
which we take to be equipped with a \emph{neighborhood of infinity},
$\begin{tikzcd}[sep=small] O_\infty \ar[r, hook] & X^d_{cpt}  \end{tikzcd}$.

\end{enumerate}

Then the \emph{solitonic phase space} $\mathrm{PhsSpc}^{\mathrm{sol}}$ is the variant of the phase space \cref{ThePhaseSpace} where the fluxes and charges are constrained to vanish on the neighborhood of infinity. Its shape is computed in corresponding variation of \cref{TheShapeOfPhaseSpace} and comes out as the \emph{pointed mapping space} $\mathrm{Map}^\ast(-,-)$ out of the one-point compactification:
\begin{equation}
  \label{ShapeOfSolitonicPhaseSpace}
  \shape \mathrm{PhsSpc}^{\mathrm{sol}}
  \simeq
  \mathrm{Map}^\ast\big(
    X^d_{\cpt},
    \mathcal{A}
  \big)
  \mathrlap{\,.}
\end{equation}

To appreciate this: When modeled via the compact-open topology on the set of continuous maps to a representative pointed topological space $A$ (with $\mathcal{A} \simeq \shape A$) out of a CW-complex $X^d_{\cpt}$, then the pointed mapping space \cref{ShapeOfSolitonicPhaseSpace} is the space of maps $\rchi$ that \emph{literally vanish at infinity} in that they take the value $0 \in A$ when evaluated at $\infty \in X^d$:
\begin{equation}
  \rchi \in 
  \mathrm{Map}^\ast\big(
    X^d_{\cpt}
    ,\,
    A
  \big)
  \;\;\;
  \Rightarrow
  \;\;\;
  \begin{tikzcd}[
    row sep=15pt
  ]
    X^d_{\cpt}
    \ar[
      rr,
      dashed,
      "{ \rchi }"
    ]
    &&
    A
    \\
    \{\infty\}
    \ar[
      u,
      hook,
      "{ \iota_{\infty} }"
    ]
    \ar[
      rr
    ]
    &&
    \{0\}
    \mathrlap{\,.}
    \ar[
      u,
      hook,
      "{ \iota_0 }"
    ]
  \end{tikzcd}
\end{equation}

Accordingly, in generalization of \cref{HigherTopOsAsHomotopyGroupsOfMaps} have the \emph{solitonic topological observables}
\begin{equation}
  \label{HigherSolitonicTopologicalObservables}
  \begin{aligned}
  \mathrm{TopObs}^{\mathrm{sol}}_\bullet
  & :=
  \pi_\bullet\big(
    \mathbf{TopObs}^{\mathrm{sol}}
  \big)
  =
  \pi_\bullet\Big(
    \mathbb{C}\big[
      \shape 
      \mathrm{PhsSpc}^{\mathrm{sol}}
    \big]
  \Big)
  \\
  & =
  \pi_\bullet\Big(
    \mathbb{C}\big[
      \mathrm{Map}^\ast(
        X^d_{\cpt}, 
        \mathcal{A}
      )
    \big]
  \Big)
  \mathrlap{\,.}
  \end{aligned}
\end{equation}

When $X^d$ is already compact, then its point at infinity is disjoint
\begin{equation}
  \mbox{$X$ compact}
  \;\;\;
  \Rightarrow
  \;\;\;
  X_{\plus}
  :=
  X \sqcup \{\infty\}
  \mathrlap{\,,}
\end{equation}
in which case the pointed maps reduce to unpointed maps:
\begin{equation}
  \mathrm{Map}^\ast\big(
    X^d_{\cpt}
    ,
    \mathcal{A}
  \big)
  \simeq
  \mathrm{Map}\big(
    X^d
    ,
    \mathcal{A}
  \big)
  \mathrlap{\,.}
\end{equation}
Therefore we will suppress the superscript $(-)^{\mathrm{sol}}$ in the following and understand ``topological field sectors'' and ``topological observables'' to refer to solitonic fields.

For example:
\footnote{
  The following statements about the one-point compactification apply to locally compact Hausdorff spaces.
}
\begin{enumerate}
\item The end of Euclidean space $\mathbb{R}^d$ is the proverbial ``$(d-1)$-sphere at infinity'' whose identification with  a single point $\infty$ closes $\mathbb{R}^d$ to the $d$-sphere:
\begin{equation}
  \mathbb{R}^d_{\cpt}
  \simeq
  S^d
  \mathrlap{\,.}
\end{equation}

\item
If  a space is already compact, then its infinity is disjoint:
\begin{equation}
  \mbox{$X$ compact}
  \;\;\;\;
  \Rightarrow
  \;\;\;\;
  X_{\cpt} \simeq X_{\plus}
  \,.
\end{equation}

\item The one-point-compactification of a product space is the \emph{smash product} $(-)\wedge(-)$ of the compactifications of the factors:
\begin{equation}
  \label{OnePointCompactificationOfProduct}
  \big(
    X
    \times
    Y
  \big)_{\cpt}
 \! \simeq
  \big(X_{\cptIndex{X}}\big) 
    \wedge 
  \big(Y_{\cptIndex{Y}}\big)
  \!:=
  \frac{
    \rule[-6pt]{0pt}{5pt}%
    X_{\cptIndex{X}} 
      \times 
    Y_{\cptIndex{Y}}
  }{
    \{\infty_{{}_X}\!\}
    \!\!\times\!\! 
    Y_{\cptIndex{Y}}
    \cup
    X_{\cptIndex{X}} 
    \!\!\times\!\!
    \{\infty_{{}_Y}\!\}
  }
  \mathrlap{\,.}
\end{equation}

\item The pointed maps out of a smash product with a circle are based loops of maps out of the remaining factor:
\begin{equation}
  \label{PointedMapsOutOfSmashWithCircle}
  \mathrm{Map}^\ast\big(
    S^1 \wedge X
    ,
    \mathcal{A}
  \big)
  \simeq
  \Omega
  \,
  \mathrm{Map}^\ast\big(
    X
    ,
    \mathcal{A}
  \big)
  \mathrlap{\,.}
\end{equation}

\end{enumerate}

\subsubsection{Higher-dimensional solitons}
\label{OnSolitonicBranes}
\footnote{
  The content of this section follows \cite[\S 2.2]{SS25-Flux}. Further discussion of different \emph{pointed spatial domains} for the same Cauchy surfaces corresponding to different dimensional soliton species is in \parencites[\S 2.1]{SS22-Conf}[\S 3.2]{BaSS26-MString}.
}

Yet more generally, a higher-dimensional soliton is a field configuration which is constrained (only) to vanish at \emph{transverse infinity}, 
namely transverse to its higher-dimensional soliton core. 

This is mathematically modeled by adjoining the point-at-infinity to the transverse directions but keeping it disjoint to the longitudinal directions. For instance $\mathbb{R}^p_{\plus} \wedge \mathbb{R}^{d-p}_{\cpt}$ is the pointed domain space on which to measure ``cartesian'' or ``flat'' $p$-solitonic  charge in $d$-dimensional space.

\begin{figure}[htb]
\caption{\label{AbrikosovVortexSchematics}
  \emph{Abrikosov vortices}
  (cf. \cite{FranzEtAl1996,nLab:VortexString}, figure adapted from \cite[Fig. 1]{LoudonMidgley2009}),
  in flat slabs of 2-dimensional type II superconductors, are the \emph{solitonic 1-branes} of ordinary electromagnetism (cf. \cref{AbrikosovVorticesAsSolitonic1Branes}).
}
\centering
\adjustbox{
  rndfbox=4pt
}{
\begin{tikzpicture}
\clip
  (-6,+1.5) rectangle 
  (+6,-3);
\node at (0,0) {
\includegraphics[width=11.5cm]{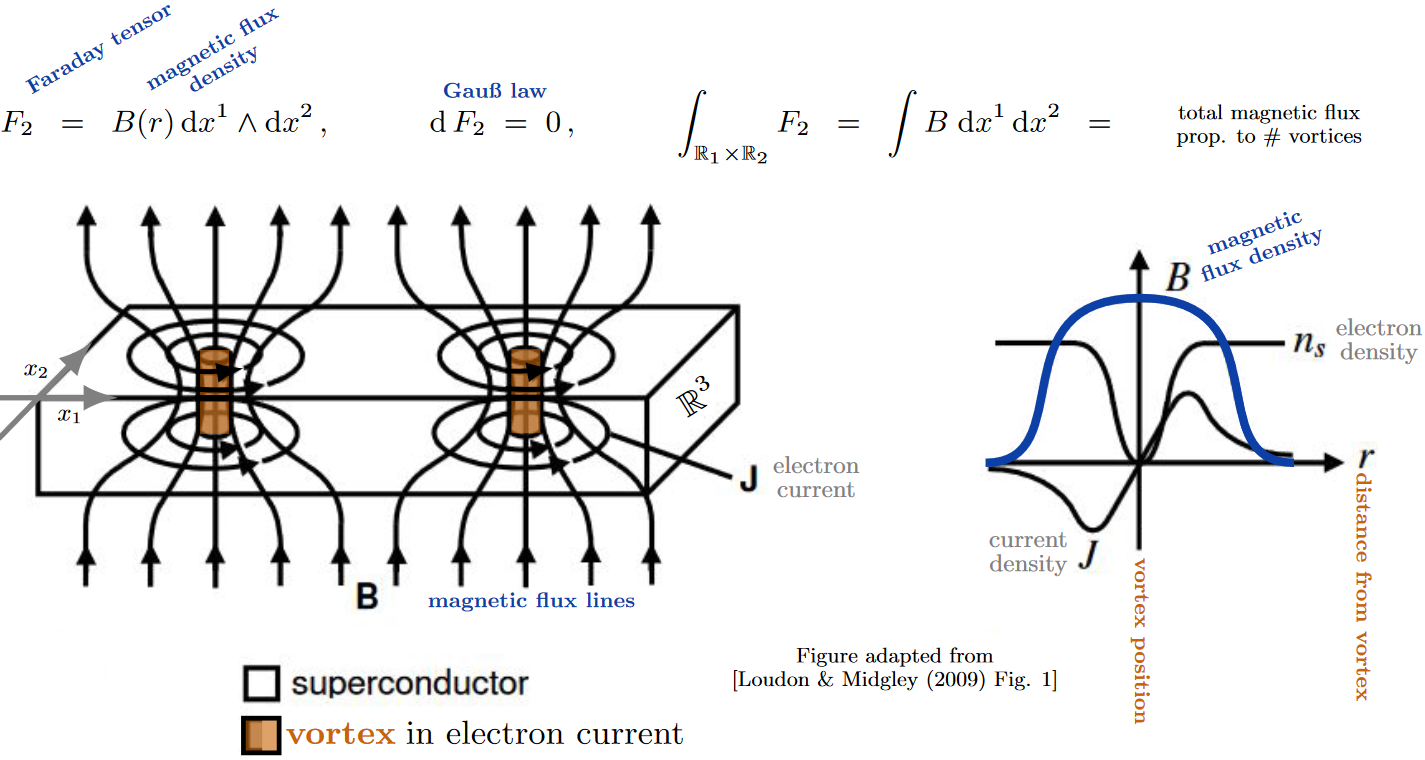}};

\draw[
  line width=15pt,
  draw=white
]
  (0,-2.3) --
  (2.8,-2.3);

\end{tikzpicture}
}

\end{figure}

\begin{example}
\label[example]{AbrikosovVorticesAsSolitonic1Branes}
The solitonic 1-branes of ordinary Maxwell electromagnetism are the \emph{Abrikosov vortices} (cf. \cref{AbrikosovVortexSchematics}): Here spacetime is the wedge product of the locus $\mathbb{R}^2_{\cpt}$ of an essentially 2-dimensional (type II) superconducting slab of material on which transverse magnetic flux is \emph{localized} (not escaping to infinity), with the remaining $\mathbb{R}^{1,1}_{\plus}$. With magnetic flux quantized in $B \mathrm{U}(1)$ it follows that there are any integer number of 1D magnetic flux concentrations 
\begin{equation}
  \begin{aligned}
    \pi_0\, 
    \mathrm{Map}^\ast\big(
      \mathbb{R}^{1,1}_{\plus}
      \wedge
      \mathbb{R}^2_{\cpt}
      ,\,
      B \mathrm{U}(1)
    \big)
    & 
    \simeq
    \pi_0\, 
    \mathrm{Map}^\ast\big(
      S^2
      ,\,
      B \mathrm{U}(1)
    \big)
    \simeq
    \pi_2\big(B \mathrm{U}(1)\big)
    \\
    & \simeq
    \mathbb{Z}
    \mathrlap{\,,}
  \end{aligned}
\end{equation}
which is famously what is seen in experiment, as indicated in \cref{AbrikosovVortexSchematics}.

(This is the actual experimental test that Dirac charge quantization, \cref{DiracChargQuantization}, of the magnetic field is the correct hypothesis for the global completion of the electromagnetic field; while the traditional argument via magnetic monopoles remains hypothetical.)
\end{example}

Therefore, in full generality we are to consider any
\begin{equation}
  \label{TheSpatialDomain}
  \big(
    X^d_{\mathrm{dom}},
    \infty
  \big)
  \in
  \mathrm{TopSp}^\ast
  \,,
  \;\;\;
  \mbox{
    with
    $X^d \simeq X^d_{\mathrm{dom}} \!\setminus\! \{\infty\}$
  }
\end{equation}
as the \emph{spatial domain} on which to observe the topological quantum behavior of higher-dimensional solitons in $\Omega \mathcal{A}$-cohomology.

As before in \cref{TransverseSpace}, in the base case of no topology-change in the light-front direction this specializes to
\begin{equation}
  X^{d}_{\mathrm{dom}}
  \simeq
  \mathbb{R}^1_{\cpt} \wedge X^{d-1}_{\mathrm{dom}}
  \mathrlap{\,.}
\end{equation}

\subsection{Light Front Quantization}
\label{OnLightFronQuantization}

We ask now for the quantization of the higher solitonic topological observables  \cref{HigherSolitonicTopologicalObservables} to \emph{quantum observables}, hence for the structure of a generally non-commutative complex star-algebra on the graded vector space $\mathrm{TopObs}_{\bullet}$ that qualifies as the appropriate graded \emph{star-algebra of quantum observables}. 

Since there is no established non-perturbative quantization prescription for higher phase space stacks of globally completed higher gauge fields, we will propose one (in \cref{OnTopologicalQuantumObservables}), at least for the solitonic topological sector in light-front form, and justify it by demonstrating the following plausibility checks:
\begin{description}
  \item[\cref{OnTopologicalQuantumObservables}] it reflects the (light-front) time-ordering phenomenon characterizing quantum operator products (recalled in \cref{OnTheLightFront,OnAlgebrasOfQuantumObservables});

  \item[\cref{The2DTQFT}]
  seen after KK-reduction on the transverse space $X^{d-1}$ it yields a (Chas-Sullivan-Godin string topology) open  
  $\mathrm{TQFT}_2[X^{d-1}]$ on the remaining $\mathbb{R}^{1,1}$;

  \item[\cref{TheddimExtednedTQFT}] seen as a $d$-dimensional Euclidean TQFT on $X^d$ it fully extends to an $(\infty,d)$-cobordism representation (\cref{TheddimExtednedTQFT});

  \item[\cref{OnMaxwellFluxObservables}] it reproduces the non-perturbative topological quantum observables of Maxwell theory;

  \item[\cref{OnMaxwellCSFluxObservables}] it reproduces the phase space structure of topological flux observables in 5D and 11D Maxwell-Chern-Simons theory.
  
\end{description}

\subsubsection{The light front}
\label{OnTheLightFront}
\footnote{
  The \emph{light front form} (also known as: ``infinite momentum frame'') of Hamiltonian evolution in relativistic field theory was first highlighted by \cite[\S 5]{Dirac1949}; for review of traditional light front quantization cf. \cite{Burkardt1996,nLab:LightFrontQuantization}. 

  The suggestion to think of what we consider in \cref{OnTopologicalQuantumObservables} as the globally completed and topological version of light front quantization originates with \cite[\S 4]{SS24-Obs}.
}

In order to couple the purely topological nature of the observables from \cref{OnTopologicalObservables} to time evolution on a globally hyperbolic spacetime $X^{1,d} \simeq \mathbb{R}^{1,0} \times X^d$ \cref{GloballyHyperbolicSpacetime},
we consider now quantization in \emph{light front form}  hence for evolution along a \emph{lightlike foliation} of $X^{1,d}$ (cf. \cref{TheLightFront}). In this situation, as the parameter $t \in \mathbb{R}^{1,0}$ evolves, the system is \emph{necessarily evolving monotonically in space} $X^d$, no matter which local Lorentz frames are chosen in \cref{GloballyHyperbolicSpacetime}.

Of course, light front quantization of 11D supergravity has famously been suggested \parencites{BFSS1997}{Susskind1997}{nLab:BFSSMatrixModel} as a non-perturbative quantization capturing aspects of ``M-theory'', but has traditionally only been applied to local degrees of freedom (in fact mostly to graviton scattering amplitudes). Our aim here is to give the topological form of light front quantization which applies to solitons in the globally completed higher gauge field content.

\begin{SCfigure}[2.6][htb]
\caption{
  \label{TheLightFront} 
  The \emph{light front form} of relativistic field theory regards evolution along a fixed \emph{lightlike} foliation of spacetime (such as by the wave fronts of a plane wave of electromagnetic radiation). 
  With respect to a temporal foliation by spacelike (Cauchy) hypersurfaces $X^d$, this means that as the actual timelike parameter evolves (along $\mathbb{R}^{1,0}$ in the graphics) the system is \emph{monotonically translating in space}, along a foliation of $X^d$. Therefore, time-ordered products of light-front observables are equivalently spatially ordered along a foliation of $X^d$.
}
\centering
\begin{tikzpicture}
  \draw[
    -Stealth,
    line width=.9,
    gray
  ]
  (0,-1) to
  (0,2);
  \draw[
    -Stealth,
    line width=.9,
    gray
  ]
  (-1,0) to
  (2,0);

  \foreach \n in {-2,...,5} {
    \draw[
      dashed,
      rotate=-45,
      shift={(0,\n*.2)},
    ]
      (-1.5,0) to (1.5,0);
  }

 \node[
   scale=.8
 ] at 
   (.4,1.7) {$\mathbb{R}^{1,0}$};
 \node[
   scale=.8
 ] at 
   (1.7, .3) {$X^{d}$};

\draw[
  black!70,
  -Stealth
]
  (45:-.2) -- 
    node[
      pos=.9, 
      below,
      scale=.9
    ] {$\tau$}
  (45:1.5);
  
\end{tikzpicture}
\end{SCfigure}
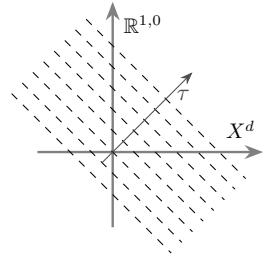

In the simplest base case, the light front evolution proceeds without ``topology change'' in the spatial direction, along a cylinder (cf. \cref{TheMFiber}):
\begin{equation}
  \label{TransverseSpace}
  \mathbb{R}^{1,d}
  \defneq
  \mathbb{R}^{1,1} \times X^{d-1}
  \,,
  \;\;\;\;\;\;
  X^d
  \defneq
  \mathbb{R}^1 
    \times 
  X^{d-1}
  \mathrlap{.}
\end{equation}
(Here $X^{d-1}$ is the \emph{transverse space} whose notational suppression in \cref{TransverseSpace} yields exactly the picture in \cref{TheMFiber}.) The quantum observables in this situation are discussed in \cref{OnAlgebrasOfQuantumObservables,The2DTQFT}.

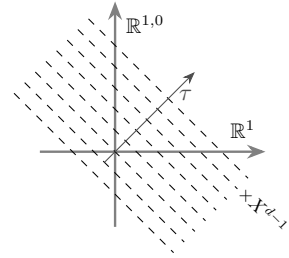
\begin{SCfigure}[2.6][htb]
\caption{
  \label{TheMFiber} 
  In the case that spacetime and with it the Cauchy surface has a line factor, $X^d = \mathbb{R}^1 \times X^{d-1}$ \cref{TransverseSpace},
  the light-front foliation in \cref{TheLightFront} may be taken to be the product of $X^{d-1}$ with the canonical light-front folia of 2D Minkowski spacetime $\mathbb{R}^{1,1}$. We refer to the factor $\mathbb{R}^1$ here as the \emph{M-fiber}, since it plays the role of the extra dimension that opens up in describing branes in 10D SuGra via 11D SuGra \parencites{DuffHoweInamiStelle1987}{Witten1995}{BFSS1997}.
}
\centering
\begin{tikzpicture}
  \draw[
    -Stealth,
    line width=.9,
    gray
  ]
  (0,-1) to
  (0,2);
  \draw[
    -Stealth,
    line width=.9,
    gray
  ]
  (-1,0) to
  (2,0);

  \foreach \n in {-2,...,5} {
    \draw[
      dashed,
      rotate=-45,
      shift={(0,\n*.2)},
    ]
      (-1.5,0) to (1.5,0);
  }

 \node[
   scale=.8
 ] at 
   (.4,1.7) {$\mathbb{R}^{1,0}$};
 \node[
   scale=.8
 ] at 
   (1.7, .3) {$\mathbb{R}^1$};

\draw[
  black!70,
  -Stealth
]
  (45:-.2) -- 
    node[
      pos=.9, 
      below,
      scale=.9
    ] {$\tau$}
  (45:1.5);

 \node[
   rotate=-45,
   scale=.7
 ]
  at (2-.03,-.8+.03) {$\times X^{d-1}$};
 
\end{tikzpicture}
\end{SCfigure}

More generally, there may be ``topology change'' in that $X^d$ is a \emph{cobordism} between (ends which are) different transverse spaces $X^{d-1}_{\mathrm{in}}$ and $X^{d-1}_{\mathrm{out}}$. The generalization of the quantization to this more general situation we discuss in \cref{TheddimExtednedTQFT}.

\subsubsection{Algebras of Quantum Observables}
\label{OnAlgebrasOfQuantumObservables}
\footnote{
  The observation that the quantum (operator) product of observables is their ordinary product after shifting their time domains into operator order is due to \cite[p. 381-2]{Feynman1948}, there argued with the time-discretized path integral in quantum mechanics, in this form recalled in \parencites[\S 7.3]{FeynmanHibbsStyer2010}[pp. 33]{Nagaosa1999}.
  
  We may observe that Feynman's ``very important relation'' \cite[(45-6)]{Feynman1948} is (the 1d discretized version of) what is now called the \emph{Schwinger-Dyson equation} (cf. \cite[(15.25)]{HenneauxTeitelboim1992}), which has a rigorous formulation in general relativistic field theory (not relying on path integral heuristics, cf. \cite[Rem. 7.7]{Rejzner2016}). This may be used to generalize Feynman's old observation to quantum field theory (cf. \href{https://physics.stackexchange.com/a/685812/5603}{\tt physics.SE:a/685812}), including its light-front formulation.
}

In order to understand what we are looking for, here we highlight a field-theoretic generalization of an old insight of \cite[p. 381-2]{Feynman1948} that has echoes in modern discussion of time-ordered products of quantum fields, but which in itself may be underappreciated. As a slogan:
\begin{standout}
  The quantum (operator) product $O_2(\tau) \star O_1(\tau)$ of observables on the $\tau$th folium (\cref{TheLightFront}) is their ordinary product $O_2(\tau + \epsilon) \cdot O_1(\tau)$ after slightly  shifting ($\epsilon \to 0$) their folium domain into operator order. 
\end{standout}

Since this may not be citable from the literature, we briefly indicate how to see it (for more details and pointers cf. \cite[\S 3]{nLab:QuantumObservable}). To that end, we  now use standard notation and terminology from Lagrangian field theory (such as found in \cite{HenneauxTeitelboim1992}).

Consider:
\begin{itemize}
\item a Lagrangian field theory on a coordinate chart $\simeq \mathbb{R}^{1+d}$,

\item with a regular Lagrangian density $L$ (including all of the usual examples, just excluding pathological field dependencies);

\item $\phi$ one of its field species (which may be a scalar field or a component of a more complex field);

\item $\big\langle \cdots \big\rangle$ the \emph{path integral expectation value} of fields observables (in some state), if it exists or can be imagined, or else the \emph{time-ordered product}, of its field operator argument.

\end{itemize}

The key fact now is the \emph{Schwinger-Dyson equation} for field insertion $\phi$, which is the statement that:
\begin{equation}
  \label{SchwingerDysonEquation}
  \left\langle
  \left(
  \frac{
    \partial L
  }{
    \partial\phi
  }
  -
  \partial_\mu 
  \frac{
    \partial L
  }{
    \partial(\partial_\mu \phi)
  }
  \right)
  \!
  (x)
  \cdot
  \phi(y)
  \right\rangle
  =
  \mathrm{i}\hbar
  \left\langle
  \frac{
    \partial \phi(y)
  }{
    \partial \phi
  }
  (x)
  \right\rangle
  =
  \mathrm{i}\hbar 
  \,
  \delta^{1+d}(x-y)
  \mathrlap{\,.}
\end{equation}

Multiply this equation with a smearing function $f \in C^\infty(\mathbb{R}^d)$ (a smooth function of compact support on a spatial slice), integrate the result over all of space and over a time interval $(\tau -\epsilon, \tau + \epsilon)$ for $\tau := y^0$ and $\epsilon \in \mathbb{R}_{> 0}$, and take the limit $\epsilon \to 0$. In this limit all bulk summands in the left term in \cref{SchwingerDysonEquation} vanish (using here our assumption that $L$ depends regularly on the fields) and what remains, under integration by parts, is just the temporal boundary evaluation of the term involving the \emph{canonical momentum}:
\begin{equation}
  \label{CanonicalFieldMomentum}
  \pi 
    := 
  \frac
    {\partial L}
    {\partial (\partial_0 \phi)}
  \mathrlap{\,.}
\end{equation}
The result is the value of the following distributional equation on the arbitrary smearing function $f$, where we now decompose coordinates as $x = (x^0 , \vec x\,)$:
\begin{equation}
  \label{PathIntegralCanonicalCommutationRelation}
  \underset{\underset{\epsilon \to 0}{\longrightarrow}}{\mathrm{lim}}
  \Big\langle
  \pi\big(\tau \!+\! \epsilon, \vec x\, \big)
  \cdot
  \phi\big(\tau, \vec y\, \big)
  -
  \phi\big(\tau, \vec y\, \big)
  \cdot
  \pi\big(\tau \!-\! \epsilon, \vec x\, \big)
  \Big\rangle
  =
  -\mathrm{i}\hbar
  \,
  \delta^d\big(\vec x - \vec y\, \big)
  \mathrlap{\,.}
\end{equation}
If we think of $\langle \cdots \rangle$ as the expectation value computed by a path integral, then its arguments $\pi \cdot \phi$ are ordinary (pointwise, commutative) products of observables, and yet this equation \cref{PathIntegralCanonicalCommutationRelation} expresses the expectation value of the equal-time \emph{canonical commutation relation} between field operators and their canonical momenta:
\begin{equation}
  \label{CanonicalCommutatorRelation}
  \pi\big(\vec x\, \big)
    \star 
  \phi\big(\vec y\, \big)
  -
  \phi\big(\vec y\, \big)
    \star 
  \pi\big(\vec x\, \big)
  =
  - \mathrm{i}\hbar
  \,
  \delta^d\big(
    \vec x - \vec y
  \, \big)
  \mathrlap{\,.}
\end{equation}
As advertized, we see that the operator product order in \cref{CanonicalCommutatorRelation} is reflected in the temporal order in \cref{PathIntegralCanonicalCommutationRelation}.

This derivation applies verbatim also after passage to light front form (\cref{TheLightFront}): In this case the integral is over an $x^+$-interval
(where
$ 
  x^{\pm}
  :=
  \tfrac{1}{\sqrt{2}}\big(
    x^0 \mp x^1
  \big)
$),
erected over a lightlike hypersurface $x^+ = \tau$, and yields the following identity,
\footnote{
  The factor of $\sfrac{1}{2}$ on the right of \cref{LightFrontPathIntegralCanonicalCommutationRelation} originates as a factor of 2 on the left due to the fact that Lorentz invariant Lagrangians have kinetic terms of the form $L = (\partial_+ \phi) \cdot (\partial_- \phi) + \cdots$ for which
  \begin{equation}
    \begin{aligned}
    \partial_\mu 
    \frac{ \partial L }{
      \partial(\partial_\mu)
    }
    &
    =
    \partial_-(\partial_+ \phi)
    +
    \partial_+(\partial_- \phi)
    + 
    \cdots
    =
    2 \, \partial_+ (\partial_- \phi)
    \mathrlap{\,.}
    \end{aligned}
  \end{equation}
  This exactly matches a corresponding factor of $\sfrac{1}{2}$ in the canonical derivation of \cref{LightFrontCanonicalCommutatorRelation}, whose derivation, however, requires more work, see \cite[\S A]{Burkardt1996}.
}
where we now decompose coordinates as $x = \big( x^+, x^-, \vec x_\perp \big)$:
\begin{equation}
  \label{LightFrontPathIntegralCanonicalCommutationRelation}
  \begin{aligned}
  &
  \underset{\underset{\epsilon \to 0}{\longrightarrow}}{\mathrm{lim}}
  \Big\langle
  \pi\big(
    \tau \!+\! \epsilon, 
    x^-,
    \vec x \,
  \big)
  \cdot
  \phi\big(
    \tau,
    x^-,
    \vec y\,
  \big)
  -
  \phi\big(
    \tau, \vec y
  \big)
  \cdot
  \pi\big(
    \tau \!-\! \epsilon, \vec x \,
  \big)
  \Big\rangle
  \\
  & \qquad =
  -\mathrm{i}\hbar
  \,
  \tfrac{1}{2}\delta\big(x^- - y^-\big)
  \delta^{d-1}\big(\vec x - \vec y\, \big)
  \mathrlap{\,.}
  \end{aligned}
\end{equation}
Here $\pi$ is the canonical light-front momentum
\begin{equation}
  \pi :=
  \frac{
    \partial L
  }{
    \partial(\partial_+ \phi)
  }
\end{equation}
and even though its nature in canonical formalism is quite different (being a second class constraint, cf. \cite{Burkardt1996}) from that of the ordinary momentum \cref{CanonicalFieldMomentum}, its operator commutator is again exactly of this form (cf. \cite[Table 2.1]{Burkardt1996}), under replacing time-ordered ordinary products with operator products:
\begin{equation}
  \label{LightFrontCanonicalCommutatorRelation}
  \begin{aligned}
  &
  \pi\big(x^-, \vec x_\perp\big)
    \star 
  \phi\big(y^-, \vec y_{\perp}\big)
  -
  \phi\big(y^-, \vec y_\perp\big)
    \star 
  \pi\big(x^-, \vec x_\perp\big)
  \\
  &
  \qquad =
  - \mathrm{i}\hbar
  \,
  \tfrac{1}{2}\delta\big(x^- - y^-\big)
  \delta^d\big(
    \vec x_\perp - \vec y_\perp
  \big)
  \mathrlap{\,.}
  \end{aligned}
\end{equation}
In conclusion, we see again that the operator product order in \cref{LightFrontCanonicalCommutatorRelation} is reflected in the light-front parameter order of ordinary products in \cref{LightFrontPathIntegralCanonicalCommutationRelation}.

This suggests that the quantum product on topological observables should similarly be their pointwise product after suitably shifting their light-front domains into the intended product order.

\subsubsection{Topological Quantum Observables}
\label{OnTopologicalQuantumObservables}
\footnote{
  The discussion of Pontrjagin algebras as algebras of quantum observables follows \parencites[\S 3-4]{SS24-Obs}[\S 4]{CSS23-QuantumStates}[p. 18]{SS22-Conf}[Fig. 11]{SS25-Srni}.
}

In the base case that the light front evolution proceeds without ``topology change'', hence along a cylinder
$
  \mathbb{R}^{1,d}
  \defneq
  \mathbb{R}^{1,1} \times X^{d-1}
  X^d
  \defneq
  \mathbb{R}^1 
    \times 
  X^{d-1}
$
\cref{TransverseSpace},
the solitonic topological observables \cref{HigherSolitonicTopologicalObservables} are the homology of the based loop space of a topological \emph{phase space in codimension=1},
\begin{equation}
  \label{PhaseSpaceInCodimension1}
  M 
    :=
  \shape \mathrm{PhsSpc}_{d-1}
    :=
  \mathrm{Map}^\ast\Big(
    X^{d-1}_{\cpt}
    ,
    \mathcal{A}
  \Big)
  \mathrlap{\,,}
\end{equation}
namely:
\begin{equation}
  \label{SolTopObsAsHomologyOfLoopSpace}
  \begin{aligned}
  \mathrm{TopObs}_{\bullet} \;
  & 
  \underset{\mathclap{\scalebox{.6}{%
    \cref{HigherSolitonicTopologicalObservables}
  }}}{\defneq}
  \;
  H_\bullet\Big(
    \mathrm{Map}^\ast\big(
      X^d_{\cpt},
      \mathcal{A}
    \big)
    ;
    \mathbb{C}
  \Big)
  \\
  &
  \underset{\mathclap{\scalebox{.6}{%
    \cref{TransverseSpace}
  }}}{\simeq}
  \;
  H_\bullet\Big(
    \mathrm{Map}^\ast\big(
      (\mathbb{R}^1 \times X^{d-1})_{\cpt},
      \mathcal{A}
    \big)
    ;
    \mathbb{C}
  \Big)
  \\
  & 
  \underset{\mathclap{\scalebox{.6}{%
    \cref{OnePointCompactificationOfProduct}
  }}}{\simeq}
  \;
  H_\bullet\Big(
    \mathrm{Map}^\ast\big(
      S^1 \wedge X^{d-1}_{\cpt},
      \mathcal{A}
    \big)
    ;
    \mathbb{C}
  \Big)
  \\
  &
  \underset{\mathclap{\scalebox{.6}{%
    \cref{PointedMapsOutOfSmashWithCircle}
  }}}{\simeq}
  \;
  H_\bullet\Big(
    \Omega
    \,
    \mathrm{Map}^\ast\big(
      X^{d-1}_{\cpt},
      \mathcal{A}
    \big)
    ;
    \mathbb{C}
  \Big)
  \\
  & 
  \underset{\mathclap{\scalebox{.6}{%
    \cref{PhaseSpaceInCodimension1}
  }}}{\defneq}
  \;
  H_\bullet\big(
    \Omega \, M
    ;\,
    \mathbb{C}
  \big)
  \mathrlap{\,.}
  \end{aligned}
\end{equation}
As such, the observables carry the graded \emph{Pontrjagin algebra} structure induced by pushforward in homology
\begin{equation}
  \label{ThePontrjaginProduct}
  \begin{tikzcd}[
    column sep=14pt
  ]   
    H_\bullet\big(
      \Omega M
    \big)
    \otimes
    H_\bullet\big(
      \Omega M
    \big)
    \ar[
      r,
      "{ \sim }"
    ]
    \ar[
      rrr,
      downhorup,
      "{ \star }"{description}
    ]
    &
    H_\bullet\big(
      (\Omega M)
      \times
      (\Omega M)
    \big)
    \ar[
      rr,
      "{
        H_\bullet(\star)
      }"
    ]
    &&
    H_\bullet\big(
      (\Omega M)
      \times
      (\Omega M)
    \big)
  \end{tikzcd}
\end{equation}
along concatenation of loops 
\begin{equation}
  \begin{tikzcd}[row sep=-2pt, column sep=0pt]
    \ell 
    &:& 
    \mathbb{R}^1_{\cpt} 
    \ar[rr] 
    && 
    M 
    \\
    &&
    x^1 
      &\longmapsto&
    \ell(x^1)
    \mathrlap{\,,}
\end{tikzcd}
\end{equation}
given by:
\begin{equation}
  \begin{tikzcd}[
    sep=-3pt
  ]
    \Omega M 
    \times 
    \Omega M
    \ar[
      rr,
      "{ \star }"
    ]
    &&
    \Omega M
    \\
    (\ell_2, \, \ell_1)
    &\longmapsto&
    \ell_2 \star \ell_1
    &:&
    x^1 \mapsto
    \left\{
    \begin{aligned}
      \ell_1\big(
        \ln(-\tfrac{1}{x^1})
      \big) & \mbox{if $x^1 \leq 0$,}
      \\
      \ell_2\big(\ln(+x^1)\big)
      &
      \mbox{if $x^1 \geq 0$.}
    \end{aligned}
    \right.
  \end{tikzcd}
\end{equation}
Here we are notationally (ab)using the happy coincidence that both concatenation of loops as well as quantum products of observables are commonly denoted by the symbol ``$\star$''.

Let's see what this means under the identification \cref{SolTopObsAsHomologyOfLoopSpace}:
\begin{enumerate}
\item 
A loop 
\[
  \ell \in \Omega M
\]
is the spatial component (the temporal component being topologically trivial) of a light-front \emph{field history} (``path'') of solitonic charges.

As such, this is best understood, equivalently, as a \emph{path} through moduli space which happens to be constant on the vanishing field configuration in the far past and future.

\item Topological observables are sensitive exactly to the homotopy class of such loops:
\[
  [\ell] 
    \in 
  \pi_0\big(\Omega M\big)
    \simeq
  \pi_1(M)
  \mathrlap{\,.}
\]

\item
Specifically, the corresponding homology class
\[
  \delta[\ell]
  \in 
  \mathbb{C}\big[
    \pi_1(M)
  \big]
  \simeq
  H_0 \big( \Omega M; \mathbb{C} \big)
\]
is equivalently the \emph{characteristic function} of $[\ell]$ among the compactly supported functions 
\begin{equation}
  H_0 \big( \Omega M; \mathbb{C} \big)
  \simeq
  \mathrm{Map}\big(
   \pi_1(M)
   ,
   \mathbb{C}
  \big)_{\mathrm{cpt}}
\end{equation}
and as such is to be understood as being the binary observable (classically: the projection operator) which asks:
\begin{quote}
``Is the solitonic light-front charge history of form $\ell$, up to continuous deformation?''
\end{quote}

\item The Pontrjagin product of a pair of such observables,
\[
  \delta[\ell_2] \star \delta[\ell_1]
  =
  \delta[ \ell_2 \star \ell_1 ]
\]
is hence the observable which asks:
\begin{quote}
``Is the solitonic field history, up to continuous deformation,  \emph{first} of form $\ell_1$ and \emph{then} of the form $\ell_2$?'' 
\end{quote}
This is like the product of these characteristic functions after shifting their light-front domain into product order. 
\end{enumerate}
In this way, the Pontrjagin product \cref{ThePontrjaginProduct} of solitonic topological observables is the analog of the time-ordered product of local observables (from \cref{OnAlgebrasOfQuantumObservables}) and therefore a good candidate for a quantum ``operator'' product.

Therefore, we regard the graded Pontrjagin algebra structure \cref{ThePontrjaginProduct} on the solitonic topological observables as the \emph{algebra of topological quantum observables} on solitonic fields: 
\begin{equation}
  \big(
    \mathrm{TopObs}_\bullet,
    \ast
  \big)
  \in
  \mathrm{Alg}\big(
    \mathbb{C}\mathrm{Vect}^{\mathbb{Z}}
  \big)
  \,.
\end{equation}
Here the ordinary topological observables (those in degree=0) form a subalgebra
\begin{equation}
  \label{AlgebraOfOrdinaryQuantumObservables}
  \big(
    \mathrm{TopObs}_0,
    \star
  \big)
  \in
  \mathrm{Alg}\big(
    \mathrm{Vect}
  \big)
  \,,
\end{equation}
which is the \emph{group algebra of the fundamental group} $\pi_1(-)$ of the codimension=1 phase space (based at the zero-charge configuration):
\begin{equation}
  \label{OrdinaryTopObsAsGroupAlgebras}
  \mathrm{TopObs}_0
  \simeq
  \mathbb{C}\Big[
    \pi_1 
    \mathrm{Map}^\ast\big(
      X^{d-1}_{\mathrm{dom}},
      \mathcal{A}
    \big)
  \Big]
  \mathrlap{\,.}
\end{equation}
On the other extreme, before passage to homology we have $A_\infty$-algebra structure  on the phase space motive \cref{MotiveOfShapeOfPhaseSpace}, which suggestively we will still denote by ``$\star$'':
\begin{equation}
  \label{MotiveAlgebra}
  \big(
    \mathbf{TopObs},
    \star
  \big)
  \in
  \mathrm{Alg}\big(
    \mathrm{Vect}_\infty
  \big)
  \mathrlap{\,,}
\end{equation}
being the \emph{$\infty$-group $\infty$-algebra} structure (cf. \parencites{AbadSchaetz2016}{nLab:InfinityGroupAlgebra}) on the based loop $\infty$-group of the codim=1 phase space:
\begin{equation}
  \label{SolTopObsInDegreeZeroAsGroupAlgebra}
  \mathbf{TopObs}
  =
  \mathbb{C}\big[
    \Omega\, 
    \mathrm{Map}^\ast\big(
      X^{d-1}_{\mathrm{dom}}
      ,
      \mathcal{A}
    \big)
  \big]
  \mathrlap{\,.}
\end{equation}

Finally, we note that a respectable algebra of quantum observables needs to be a complex \emph{star-algebra}, hence equipped with a anti-$\mathbb{C}$-linear involution $(-)^\ast$ reflecting time reversal. Such a structure is naturally present on the Pontrjagin algebra \cref{ThePontrjaginProduct}, given by complex conjugation of coefficients combined with pushforward in homology: 
\begin{equation}
  \begin{tikzcd}[
    column sep=30pt
  ]
    H_\bullet(\Omega M; \mathbb{C})
    \ar[
      rr,
      "{
        H_\bullet\big(
          \overline{(-)};\,
          \mathbb{C}
        \big)
      }"
    ]
    \ar[
      rrrr,
      downhorup,
      "{
        (-)^\ast
      }"{description}
    ]
    &&
    H_\bullet(\Omega M; \mathbb{C})
    \ar[
      rr, 
      "{  
        H_\bullet\big(
          \Omega M;\, 
          \overline{(-)}
        \big)
      }"
    ]
    &&
    H_\bullet(\Omega M; \mathbb{C})
  \end{tikzcd}
\end{equation}
along reversal of loops:
\begin{equation}
  \begin{tikzcd}[row sep=-2pt, column sep=0pt]
    \Omega M
    \ar[
      rr,
      "{ \overline{(-)} }"
    ]
    &\phantom{--}&
    \Omega M
    \\
    \ell &\longmapsto&
    \overline{\ell}
    \mathrlap{
      \;
      : s \mapsto
      \ell(1-s)
      \mathrlap{\,.}
    }
  \end{tikzcd}
\end{equation}

In fact, the $\infty$-group structure of $\Omega M$ induces on $\mathrm{TopObs}_\bullet$ the structure of a graded cocommutative \emph{Hopf algebra} whose antipode is this star-involution.

\subsubsection{The 2D TQFT of the M-fiber}
\label{The2DTQFT}
\footnote{
  The open/closed TQFT perspective on Chas-Sullivan-type string topology operations was proposed by \cite{Godin2007} and worked out in \cite{Kupers2011}. The observation of \cref{PontrjaginAlgebraViaOpenStringTopology} (not surprising but noteworthy) that in the open string 0-brane sector the string topology product coincides with the Pontrjagin product \cref{ThePontrjaginProduct} is due to \cite[\S 4.1.5]{SS25-WilsonLoops} using \cite[p. 136-7]{Kupers2011}. 
}

The usual global Hamiltonian formulation of a QFT means to regard it as quantum mechanical evolution along the parameter $\tau$ indexing the spacetime folia (\cref{TheLightFront}). For a topological field theory, this retains \emph{no} evolution information unless and until spatial topology changes with the parameter $\tau$.

At the other extreme, the \emph{extended functorial} formulation of $1+d$-dimensional QFT keeps track of its evolution in all $1+d$ spacetime directions. This retains more information about topological field theories, even locally. In fact, the \emph{cobordism hypothesis} asserts that a fully extended topological field theory is completely characterized by its restriction to any spacetime point.

Interpolating between these extremes, Kaluza-Klein-style reduction over a product decomposition $X^{1,d} = X^{1,d'} \times X^{d-d'}$ means to describe a $(1+d)$-dimensional QFT via a($n$-extended) $\mathrm{QFT}_{1+d'}[X^{d-d'}]$ in $(1+d')$-dimensions indexed by the fiber space $X^{d-d'}$.

Considering this for $d'=1$ in the cylindrical situation \cref{TransverseSpace}, we expect to obtain a 2-dimensional $\mathrm{TQFT}_{1+1}[X^{d-1}]$ which describes the propagation of a ``string'' (the M-theory fiber) in the moduli space of fields on $X^{d-1}$.
Specifically, due to the restriction that solitonic fields \emph{vanish at infinity}, this is effectively an ``open string'' whose endpoints are stuck on the ``0-brane'' (the base point) of vanishing field moduli in the moduli space $\mathcal{X} \defneq \mathrm{Map}^\ast\big(X^{d-1}_{\mathrm{dom}}, \mathcal{A}\big)$, cf. \cref{HCFT}.

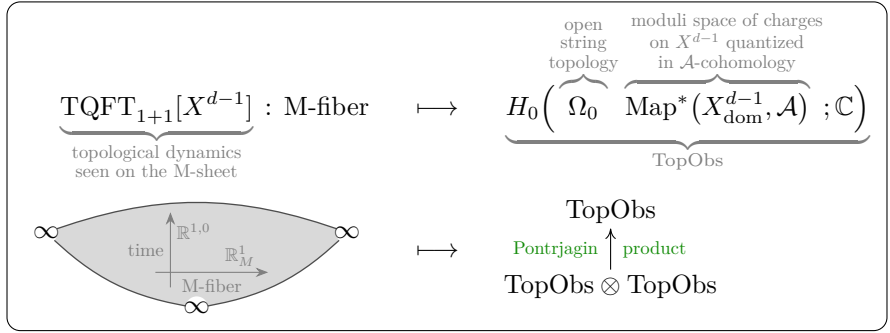
\begin{figure}[htb]
\caption{
  \label{HCFT} 
  The topological dynamics (of higher gauge fields flux-quantized in $\hotype{A}$-cohomology) on a globally hyperbolic spacetime $X^{1,d}$ with an ``M-fiber'', $X^{1,d} \simeq \mathbb{R}^{1,1} \times X^{d-1}$ (\cref{TheMFiber}), yields a $D = 1+1$ genus=0 open TQFT on the ``M-sheet'' $\mathbb{R}^{1,0} \times \mathbb{R}^1_M$, which is the ``HCFT'' induced by \emph{open string topology operations} over the moduli space of topological charges on $X^{d-1}$. 
  Here the solitonic nature of the topological charges, namely their \emph{vanishing at infinity}, entails that the spatial ends $\infty$ of the ``M-sheet'' $\mathbb{R}^{1,0} \times \mathbb{R}^1_M$ are stuck on the ``0-brane'' locus $\infty \mapsto \{\ast\} \hookrightarrow \mathrm{Map}^\ast\big(X^{d-1}_{\mathrm{dom}},\hotype{A}\big)$ in the moduli space, which reduces (\cref{PontrjaginAlgebraViaOpenStringTopology}) the string product to the Pontrjagin product of topological quantum observables from \cref{OnTopologicalQuantumObservables}.
}

\centering

\adjustbox{
  rndfbox=5pt
}{
$
  \begin{aligned}
  \grayunderbrace{
    \mathrm{TQFT}_{1+1}[X^{d-1}]
  }{
    \mathclap{
    \substack{
      \scalebox{.7}{topological dynamics}
      \\
      \scalebox{.7}{seen on the M-sheet}
    }
    }
  }
  :
  \scalebox{1}{M-fiber}
  &
  \;\;\;\longmapsto\;\;\;
  \grayunderbrace{
  H_0\Big(
  \grayoverbrace{
    \Omega_0
    \mathclap{\phantom{X^{d-1}_{\mathrm{dom}}}}
  }{
    \mathclap{
    \substack{
      \scalebox{.7}{open}
      \\
      \scalebox{.7}{string}
      \\
      \scalebox{.7}{topology}
    }
    }
  }
  \grayoverbrace{
  \mathrm{Map}^\ast\big(
    X^{d-1}_{\mathrm{dom}}
    ,
    \hotype{A}
  \big)
  }{\;\;
    \substack{
      \scalebox{.7}{moduli space of charges}
      \\
      \scalebox{.7}{on $X^{d-1}$ quantized}
      \\
      \scalebox{.7}{in $\hotype{A}$-cohomology}
    }
  }
  ;
  \mathbb{C}
  \Big)
  }{
    \mathrm{TopObs}
  }
  \\
\begin{tikzpicture}[
  baseline=(current bounding box.center)
]

\def\width{2}

\draw[
  line width=.5,
  draw=black!70,
  fill=gray!30
] 
  (-\width, 0) 
    to[bend left=20]
  (+\width,0) 
    to[bend left=18]
  (0, -1)
    to[bend left=18]
  (-\width,0); 

\fill[white] (-\width,0) circle (.15);
\fill[white] (+\width,0) circle (.15);
\fill[white] (0,-1) circle (.15);
\node[scale=1] at (-\width,0) {$\infty$};
\node[scale=1] at (+\width,0) {$\infty$};
\node[scale=1] at (0,-1) {$\infty$};

\begin{scope}[
  shift={(-.3-.05,-.48-.05)}
]
\draw[
  -Stealth,
  gray
] 
  (-.2,0) 
    to[
      "$\mathbb{R}^1_M$"{
        near end,
        scale=.7
      },
      "{ M-fiber }"{
        swap,
        scale=.7
      }
    ]
  (1.3,0);
\draw[
  -Stealth,
  gray
] 
  (0,-.2) 
    to[
      "{ time }"{
        scale=.7
      },
      "{
        $\mathbb{R}^{1,0}$
      }"{
        swap,
        near end,
        scale=.7
      }
    ]
  (0,.8);
\end{scope}
\end{tikzpicture}
&
\;\;\;\longmapsto\;\;\;
\begin{tikzcd}[
  row sep=12pt
]
  \mathrm{TopObs}
  \\
  \mathrm{TopObs}
  \otimes
  \mathrm{TopObs}
  \ar[
    u,
    shorten=-2pt,
    "{\color{darkgreen} 
      \scalebox{.7}{
        Pontrjagin
      }
    }",
    "{\color{darkgreen} 
      \scalebox{.7}{
        product
      }
    }"{swap},
  ]
\end{tikzcd}
  \end{aligned}
$
}

\end{figure}

Open-closed 2d TQFTs describing propagation on bare target space manifolds $\mathcal{X}$ are known as Chas-Sullivan-Godin \emph{string topology operations}. In the case of open strings attached to 0-branes these exist in the further generality that $\mathcal{X}$ is any topological space and one has (using \cite[p. 136-7]{Kupers2011}):
\begin{proposition}
  \label[proposition]
   {PontrjaginAlgebraViaOpenStringTopology}
  The string-topology TQFT of open strings attached to a 0-brane $b_0$ in any topological space $\mathcal{X}$ exists and assigns:
  \begin{enumerate}
    \item to the interval the homology $H_\bullet(\Omega_{b_0}\mathcal{X})$ of the base loop space of $\mathcal{X}$;
    \item to the concatenation of intervals the Pontrjagin product on $H_\bullet(\Omega_{b_0}\mathcal{X})$.
  \end{enumerate}
\end{proposition}

For our case, $\mathcal{X} = \mathrm{Map}^\ast(X^{d-1}, \mathcal{A})$, this is exactly the topological quantum data that we obtained in \cref{OnTopologicalQuantumObservables}.

\subsection{Topological Quantum States}
\footnote{
  The following discussion of topological quantum states on solitonic field histories goes back to \parencites[\S 3.5]{SS22-Conf}[\S 4]{CSS23-QuantumStates} with further discussion in \parencites[\S 2.1]{SS25-FQH}.
}

With the solitonic topological quantum observables understood (in \cref{OnLightFronQuantization}), we turn to the corresponding (solitonic topological) pure \emph{quantum states}. We understand them first as modules over the algebra of quantum observables for trivial light-cone topology and then in more detail as values of a holographic TQFT functor.

\subsubsection{Local systems of topological quantum states}
\label{OnLocalSystemsOfTopologicalQuantumStates}
\footnote{
  For background on local systems in our context, see \parencites[Lit. 2.22]{MySS2024}[Rem. 2.12]{SS25-FQH}[Fig. 1]{SS25-CrystallineChern}.
}

In the case $X^d = \mathbb{R}^1 \times X^{d-1}$ \cref{TransverseSpace}, consider the algebra of ordinary solitonic topological quantum observables $\mathrm{TopObs}_0$ according to \cref{AlgebraOfOrdinaryQuantumObservables}. 
The corresponding \emph{spaces of ordinary topological pure quantum states} must be modules over this algebra of observables. But this being a group algebra \cref{OrdinaryTopObsAsGroupAlgebras}, these modules are equivalently the linear representations of the group, and that group being a fundamental group of a space, these are equivalently the \emph{local systems} of vector spaces on the connected component $\mathrm{Map}^\ast_0(-,-)$ of the codim=1 phase space:
\begin{equation}
  \label{SpacesOfTopQStatesAsLocalSystems}
  \begin{aligned}
  \mathrm{TopQStSpcs}
  & 
  :=
  \big(
  \mathrm{TopObs}_0
  \big)
  \mathrm{Mod}\big(
    \mathrm{Vect}
  \big)
  \\
  & 
  \,\simeq
  \big(
    \pi_1\,
    \mathrm{Map}^\ast(X^{d-1}_{\cpt},\mathcal{A})
  \big)
  \mathrm{Rep}\big(
    \mathrm{Vect}
  \big)
  \\
  & 
  \,\simeq
  \mathrm{Vect}^{
    \mathrm{Map}^\ast_0(%
      X^{d-1}_{\cpt},\mathcal{A}%
    )
  }
  \mathrlap{\,.}
  \end{aligned}
\end{equation}
Here in the last line,
\begin{equation}
  \hspace{-.8cm}
  \mathcal{X}
  \in
  \mathrm{Grpd}_\infty
  \;\;\;
  \vdash
  \;\;\;
  \mathrm{Vect}^{\mathcal{X}}
  :=
  \mathrm{Func}\big(
    \mathcal{X},
    \mathrm{Vect}
  \big)
  \in
  \mathrm{Cat}
\end{equation}
denotes the category of local 1-systems on an $\infty$-groupoid $\mathcal{X}$ (the shape of some topological space), being equivalently the category of functors $\begin{tikzcd}[sep=small]\mathcal{X} \ar[r, dashed] & \mathrm{Vect}\end{tikzcd}$ to the 1-category of vector spaces.

Over general base spaces, such local systems are equivalently tuples of linear representations of the fundamental groups of the base space, indexed by the connected components $\pi_0$, which in our case are the total charge sectors in nonabelian cohomology $\tilde H^1(X^{d-1}_{\cpt},\Omega \mathcal{A})$ \cref{TheNonabelianCohomologySet}:
\begin{equation}
  \label{LocalSystemsAsFundGrpReps}
  \mathrm{Vect}^{
    \mathrm{Map}(X^{d-1}_{\cpt}, \mathcal{A})
  }
  \simeq
  \qquad \;\;
  \prod_{\mathclap{
    [x] \in 
    \tilde H^1\big(X^{d-1}_{\cpt}, \Omega\mathcal{A}\big)
  }}
  \qquad \;\;
  \big(
    \pi_1 
    \mathrm{Map}(X^{d-1}_{\cpt}, \mathcal{A})
  \big)
  \mathrm{Rep}(\mathrm{Vect})
  \mathrlap{\,.}
\end{equation}
This means (cf. \cref{AsymptoticChargeSector}) that in generalizing quantum state spaces from local systems over the connected component of 0-charge in \cref{SpacesOfTopQStatesAsLocalSystems} to local systems over all of $M$ \cref{PhaseSpaceInCodimension1} we are allowing for all possible topological \emph{asymptotic boundary conditions}
in terms of \emph{asymptotic charge sectors} (not necessarily that of vanishing charge) in the longitudinal direction. (In the picture of \cref{HCFT}, this means that we are allowing the M-fiber string to be attached to any ``0-brane'' in the moduli space.)

\begin{SCfigure}[.71][htb]
\caption
[
  For an inclusion $\iota$ of an asymptotic boundary of space (for us: at longitudinal infinity), its nonabelian cohomology classifies the \emph{asymptotic boundary charges} that may serve as \emph{asymptotic boundary conditions} for fields.
]
{
  \label{AsymptoticChargeSector}
  For an inclusion\footnotemark\; $\iota$ of an asymptotic boundary of space (for us: at longitudinal infinity), its nonabelian cohomology classifies the \emph{asymptotic boundary charges} that may serve as \emph{asymptotic boundary conditions} for fields.
}
\centering
$
\adjustbox{raise=42pt}{
  \begin{tikzcd}[
    row sep=5pt,
    column sep=23pt
  ]
    X^{d-1} 
    \ar[
      rr,
      hook,
      "{
        \iota
      }"
    ]
    &&
    X^d
    \\
    \pi_0
    \mathrm{Map}^\ast\!\big(
      X^{d-1}_{\cpt},
      \mathcal{A}
    \big)
    \ar[
      rr,
      <-,
      "{ \;
        \mathrm{Map}^\ast(%
          \iota,%
          \mathcal{A}%
        )
      }"
    ]
    \ar[
      d,
      equals
    ]
    &&
    \pi_0
    \mathrm{Map}^\ast\!\big(
      X^{d}_{\cpt},
      \mathcal{A}
    \big)
    \ar[
      d,
      equals
    ]
    \\[+4pt]
    \tilde H^1\big(
      X^{d-1}_{\cpt},
      \Omega\mathcal{A}
    \big)
    \ar[
      rr,
      <-,
      "{ 
        (\iota_{\cpt})^\ast 
      }"{pos=.6}
    ]
    &&
    \tilde H^1\big(
      X^{d}_{\cpt},
      \Omega\mathcal{A}
    \big)
  \end{tikzcd}
  }
$
\end{SCfigure}
\footnotetext{
  We are using that the one-point compactification $(-)_{\cpt}$ is functorial on proper maps (between locally compact Hausdorff spaces), such as on boundary inclusions of manifolds: Such maps extend continuously to the point at infinity. Therefore, while  the middle map in \cref{AsymptoticChargeSector} is really $\mathrm{Map}^\ast\big( \iota_{\cpt}, \mathcal{A}\big)$, we are notationally suppressing the subscript $(-)_{\cpt}$ on $\iota$, here and in the following, for notational brevity.
}
In a key example \cite{SS23-Defect}, the  codim=1 phase space for a suitable theory with \emph{intersecting solitonic branes} (a situation slightly outside the scope that we have reviewed here) is the \emph{configuration space} of some number of ordered points in the plane (reflecting the positions of defect branes in their transverse space), whence the corresponding spaces of solitonic topological quantum states \cref{SpacesOfTopQStatesAsLocalSystems} are pure \emph{braid representations} witnessing (nonabelian) anyonic ``statistics'', a hallmark property of topological quantum states in 2 dimensions. 
The variant of this example describing specifically the anyons of fractional quantum Hall systems we discuss in more detail in \cref{OnFQAHMaterials}.

\subsubsection{$\infty$-Local systems of higher topological quantum states}
\label{OnInfinityLocalSystemsOfHigherTopQStates}
\footnote{
  In dg-categorical language, $\mathbb{C}$-linear $\infty$-local systems go back to \cite{BlockSmith2014}.
  A model category presentation of $\mathbb{C}$-linear  $\infty$-local systems modeled as simplicial functors is in \cite[\S 3.2]{SS23-EoS}.
  Discussion in quasi-category language is in \cite[\S 4.3]{HopkinsLurie2013}, with relevant points surveyed, in our context, in \cite[Exs. 3.11, 5.6]{Sc14-LinTypes}.
}

As the full motive $\mathbf{TopObs}$ \cref{MotiveOfShapeOfPhaseSpace} refines the topological observables, so the local systems of quantum states \cref{OnLocalSystemsOfTopologicalQuantumStates} refine to $\infty$-local systems of higher topological quantum states (cf. \cite[Rem. 3.11]{SS23-EoS}):
\begin{equation}
  \label{HigherSpacesOfTopQStatesAsLocalSystems}
  \begin{aligned}
  \mathbf{TopQStSpc}
  & 
  :=
  \big(
    \mathbf{TopObs}
  \big)
  \mathrm{Mod}\big(
    \mathrm{Vect}_\infty
  \big)
  \\
  & 
  \,\simeq
  \big(
    \Omega\,
    \mathrm{Map}^\ast(X^{d-1}_{\cpt},\mathcal{A})
  \big)
  \mathrm{Rep}\big(
    \mathrm{Vect}_\infty
  \big)
  \\
  & 
  \,\simeq
  \mathrm{Vect}_\infty^{
    \mathrm{Map}^\ast_0(%
      X^{d-1}_{\cpt},\mathcal{A}%
    )
  }
  \mathrlap{,}
  \end{aligned}
\end{equation}
where in the last line
\begin{equation}
  \label{FormingInfinityLocalSystems}
  \hspace{-.8cm}
  \mathcal{X}
  \in
  \mathrm{Grpd}_\infty
  \;\;\;
  \vdash
  \;\;\;
  \mathrm{Vect}_\infty^{\mathcal{X}}
  :=
  \mathrm{Func}\big(
    \mathcal{X},
    \mathrm{Vect}_\infty
  \big)
  \in
  \mathrm{Cat}_\infty
\end{equation}
is the $\infty$-category of \emph{$\infty$-local systems} on an $\infty$-groupoid $\mathcal{X}$, consisting of the $\infty$-functors $\begin{tikzcd}[sep=small]\mathcal{X} \ar[r, dashed] & \mathrm{Vect}_\infty\end{tikzcd}$ to the $\infty$-category of $\infty$-vector spaces \cref{InfinityCategoryOfInfinityVectorSpaces}.

Since every homotopy type $\mathcal{X}$ is equivalently the disjoint union of the deloopings $B \mathcal{G}_x$ of its loop $\infty$-groups $\Omega_x \mathcal{X}$ as the base point $x$ ranges through its connected components, the $\infty$-local systems on $\mathcal{X}$ are equivalently tuples of linear $\infty$-representations of these $\infty$-groups (in refinement of \cref{LocalSystemsAsFundGrpReps}, cf. \cite[Rem. 3.11]{SS23-EoS}):
\begin{equation}
  \label{InfinityLocalSystemsAsInfinityRepresentations}
  \mathrm{Vect}_\infty^{\mathcal{X}}
  \simeq
  \prod_{[x] \in \pi_0(\mathcal{X})}
  (\Omega_x \mathcal{X}) \mathrm{Rep}(\mathrm{Vec}_\infty)
  \mathrlap{\,.}
\end{equation}

Incidentally, this makes manifest that the algebraic K-theory of $\mathbf{TopQStSpc}$ \cref{HigherSpacesOfTopQStatesAsLocalSystems} is the $H\mathbb{C}$-linear \emph{Waldhausen A-theory} of the shape of the phase space (\cite{HessShipley2016,nLab:ATheory}). 

\subsubsection{The holographic TQFT}
\label{TheddimExtednedTQFT}
\footnote{
  The notion of $(\infty,n)$-extended topological field theory is of course due to \cite{Lurie2008}. The construction of ``classical'' examples via  mapping spaces into classifying spaces is considered in  \parencites[\S 3]{FHLT2010}[v2 \S5.2.18.3]{Sc13-dcct}. Their quantization by pull-push of coefficients via \emph{motivic yoga} (6-functor formalism \cite{Scholze2025}, here in ``Wirthm{\"u}ller form'' where $f^! \simeq f^\ast$, cf. \cite{FauskHuMay2003}), including the Beck-Chevalley condition \cref{BeckChevalleyCondition}, is considered in \parencites[Cor. 5.10, Def. 7.6]{Sc14-LinTypes}[v2 \S5.2, Def. 7.6]{Sc13-dcct}{Stefanich2025-Categorification}. 
  \nopagebreak
  That this \emph{motivic yoga} may be understood as axiomatics for dependent \emph{linear homotopy types} which formalize the logic of quantum systems in general was envisioned in \cite[\S 3.2]{Sc14-LinTypes}, syntactically realized in \cite[\S 2.4]{Riley2022}, and developed for application to (topological) quantum information/computation (as in \cref{OnFQAHMaterials}) in \cite{SS25-Monadology}
}

We have seen:
\begin{description}
  \item[in \cref{OnTopologicalObservables}] that
  associated with the full codim=0 phase space space is the motive $\mathbf{TopObs}$ of (solitonic) topological observables;

  \item[in \cref{OnLocalSystemsOfTopologicalQuantumStates,OnInfinityLocalSystemsOfHigherTopQStates}] that, associated with the codim=1 phase space (in the case of trivial lightcone topology) is an $\infty$-category $\mathbf{TopQStSpc}$ of spaces of quantum states given by the $\infty$-local systems.  
\end{description}

But in fact, also the first item is naturally expressed in terms of $\infty$-local systems: The motive of the shape of the phase space, $\mathbf{TopObs}$ \cref{MotiveOfShapeOfPhaseSpace},  is equivalently the result of \emph{pull-pushing} the tensor unit $\infty$-local system through this correspondence (cf. \parencites[Ex. 4.1]{Sc14-LinTypes}[Prop. 2.7]{SS25-Monadology}):
\begin{equation}
\label{PhaseSpaceMotiveAsPullPush}
\hspace{-2cm}
  \begin{tikzcd}[ 
    row sep=-3pt, 
    column sep=40pt
  ]
    & 
    \mathrm{Map}^\ast(X^d_{\cpt}, \mathcal{A})
    \ar[
      dl,
      "{ p }"{swap}
    ]
    \ar[
      dr,
      "{ p }"
    ]
    \\[10pt]
    \ast
    &&
    \ast
    \\
    \mathrm{Vect}^{\ast}_\infty
    \ar[
      r,
      "{ 
        \BigDelta_p
      }"
    ]
    &
    \mathrm{Vect}
      ^{\!\!\mathrm{Map}^\ast(X^d_{\cpt}, \mathcal{A})}
      _\infty
    \ar[
      r,
      "{ 
        \sum_p 
      }"
    ]
    &
    \mathrm{Vect}^{\ast}_\infty
    \\
    \TensorUnit
    \ar[
      rr,
      |->,
      shorten <=5pt,
      shorten >=30pt
    ]
    &&
   \;\;\; \mathclap{
    \mathbb{C}\big[
      \mathrm{Map}^\ast(X^d_{\cpt}, \mathcal{A})
    \big]
    \mathrlap{
      \defneq 
      \mathbf{TopObs}.
    }
    }
  \end{tikzcd}
\end{equation}
Here we are using that for a map $f$ of $\infty$-groupoids there is an induced adjoint pair of base change functors of $\infty$-local systems:
\begin{equation}
  \label{BaseChangeAdjunctionForInfinityLocalSystems}
  \begin{tikzcd}[
    row sep=4pt,
    column sep=40pt
  ]
    \mathcal{X}
    \ar[
      rr,
      "{ f }"
    ]
    &&
    \mathcal{Y}
    \\
    \mathrm{Vect}^{\mathcal{X}}_\infty
    \ar[
      rr,
      shift left=9pt,
      "{ 
        \sum_f \,\defneq\, f_!
      }"{description}
    ]
    \ar[
      rr,
      phantom,
      "{
        \bot
      }"{scale=.6, pos=.6}
    ]
    \ar[
      rr,
      <-,
      shift right=9pt,
      "{
       \;
        \BigDelta_f
        \,\defneq\,
        f^\ast 
      }"{description}
    ]
    &&
    \mathrm{Vect}^{\mathcal{Y}}_\infty
    \mathrlap{\,,}
  \end{tikzcd}
\end{equation}
where the right adjoint $\BigDelta_f$ is pullback, while the left adjoint $\sum_f$ is given by forming $\infty$-colimits of $\infty$-vector spaces over the homotopy fibers of $f$ --- which for constant systems reduces to the tensoring \cref{TensoringOfInfinityVectOverInfinityGrpd}.

However, from this perspective, the system of observables and quantum states generalizes to spacetimes $X^{1,d}$ which are not necessarily topologically trivial in the lightcone direction as in \cref{TransverseSpace}, but where $X^d$ may be a non-trivial \emph{cobordism}, hence\footnote{
  We are allowing the cobordisms and their in/out-boundaries to be non-compact manifolds, but we only ever map out of their one-point compactification, cf. \cref{AsymptoticChargeSector}. This slightly non-standard construction is the relevant one for application to actual physics.
}
a $d$-manifold with ``longitudinal'' boundary $I^{d-1} \sqcup O^{d-1} := \partial X^d$, included as:
\begin{equation}
  \begin{tikzcd}[
    row sep=-2pt,
    column sep=50pt
  ]
    & 
    X^d
    \\
    I^{d-1}
    \ar[
      ur,
      hook,
      "{ i }"
    ]
    &&
    O^{d-1}
    \mathrlap{\,.}
    \ar[
      ul,
      hook',
      "{ o }"{swap}
    ]
  \end{tikzcd}
\end{equation}
Associated with such data is first, under $\mathrm{Map}^\ast\big( (-)_{\cpt}, \mathcal{A} \big)$ (\cref{AsymptoticChargeSector}), the \emph{correspondence} which restricts the topological field configurations to their in/out-going values, and then, under $\mathrm{Vect}_\infty^{(-)}$ \cref{FormingInfinityLocalSystems} and pull-push
\cref{BaseChangeAdjunctionForInfinityLocalSystems}, a linear $\infty$-functor between categories of in/out-spaces of higher topological quantum states (cf. \cite[Def. 7.6]{Sc14-LinTypes}):
\begin{equation}
  \label{PullPushOfInfinityLocalSystems}
  \begin{tikzcd}[ 
    row sep=-4pt,
    column sep=42pt
  ]
    & 
    \mathrm{Map}^\ast\big(X^d_{\cpt}, \mathcal{A}\big)
    \ar[
      dl,
      shorten=-3pt,
      "{ 
        \mathrm{Map}^\ast(
          i,
          \mathcal{A}
        )
      }"{sloped, pos=.5}
    ]
    \ar[
      dr,
      shorten=-3pt,
      "{ 
        \mathrm{Map}^\ast(
          o,
          \mathcal{A}
        )
      }"{sloped, pos=.5}
    ]
    \\[7pt]
    \mathrm{Map}^\ast\big(I^{d-1}_{\cpt}\!, \mathcal{A}\big)
    &&
    \mathrm{Map}\big(O^{d-1}_{\cpt}\!, \mathcal{A}\big)
    \\[10pt]
    \mathrm{Vect}
      ^{\!\!\mathrm{Map}^\ast(I^{d-1}_{\cpt}\!,\mathcal{A})}
      _\infty
    \ar[
      r,
      "{
        \underset
        {\mathrm{Map}^\ast(i, \mathcal{A})}
        {\BigDelta}
      }"
    ]
    &
    \mathrm{Vect}
      ^{\!\!\mathrm{Map}^\ast(X^d_{\cpt}, \mathcal{A})}
      _\infty
    \ar[
      r,
      "{
        \underset
        {%
          \mathrm{Map}^\ast(o, \mathcal{A})%
        }
        {\sum}
      }"
    ]
    &
    \mathrm{Vect}
      ^{\!\!\mathrm{Map}(O^{d-1}_{\cpt}\!,\mathcal{A})}
      _\infty .
  \end{tikzcd}
\end{equation}

Due to the \emph{Beck-Chevalley condition} satisfied by $\infty$-local systems (cf. \cite[Prop. 4.3.3]{HopkinsLurie2013}): 
\begin{equation}
  \label{BeckChevalleyCondition}
  \begin{tikzcd}[
    row sep=5pt, 
    column sep=15pt
  ]
    & 
    \mathcal{X}
    \times_{\mathcal{B}} 
    \mathcal{Y}
    \ar[
      dl,
      "{
        \mathrm{pr}_{\mathcal{X}}
      }"{pos=.2, swap}
    ]
    \ar[
      dr,
      "{
        \mathrm{pr}_{\mathcal{Y}}
      }"{pos=.3}
    ]
    \ar[
      dd,
      phantom,
      "{ \lrcorner }"{
        rotate=-45,
        pos=.1
      }
    ]
    \\
    \mathcal{X}
    \ar[
      dr,
      "{ f }"{swap}
    ]
    &&
    \mathcal{Y}
    \ar[
      dl,
      "{ g }"
    ]
    \\
    & 
    \mathcal{B}
  \end{tikzcd}
  \;\;\;
  \Rightarrow
  \;\;\;
  \begin{tikzcd}[
    row sep=2pt, 
    column sep=22pt
  ]
    & 
    \mathrm{Vect}_\infty
     ^{{
       \mathcal{X}
         \times_{\mathcal{B}} 
        \mathcal{Y}
      }}
    \ar[
      dl,
      <-,
      shorten <=-3pt,
      "{
         \BigDelta
           _{\mathrm{pr}_{\mathcal{X}}}
      }"{pos=.2, swap}
    ]
    \ar[
      dr,
      shorten <=-7pt,
      "{
        \sum_{%
        \mathrm{pr}_{\mathcal{Y}}%
        }
      }"{pos=.2}
    ]
    \\
    \mathrm{Vect}_\infty^{\mathcal{X}}
    \ar[
      dr,
      "{ 
        \sum_f 
      }"{swap}
    ]
    &&
    \mathrm{Vect}_\infty^{\mathcal{Y}}
    \mathrlap{\,,}
    \ar[
      dl,
      <-,
      "{ 
        \BigDelta_g
      }"
    ]
    \\
    & 
    \mathrm{Vect}_\infty^{\mathcal{B}}
  \end{tikzcd}
\end{equation}
this construction \cref{PullPushOfInfinityLocalSystems} is in fact (symmetric-monoidal) functorial under composition of correspondences (cf. \cite[Cor. 5.10]{Sc14-LinTypes}) to be denoted:
\begin{equation}
  \label{TheTQFTd}
  \mathbf{TopObs}(-)
  \;:\;
  \begin{tikzcd}[
    column sep=10pt
  ]
    I^{d-1}
    \ar[
      d,
      hook,
      "{
        i
      }"
    ]
    &\longmapsto&
    \mathrm{Vect}_\infty^{\!\!
      \mathrm{Map}(I_{\cpt},%
      \mathcal{A})
    }
    \ar[
      dd,
      "{
        \underset{%
        \mathrm{Map}(%
          o%
          ,%
          \mathcal{A}%
        )%
        }{\sum}
        \circ
        \underset{
        \mathrm{Map}(%
          i%
          ,%
          \mathcal{A}%
        )
        }{\BigDelta}
      }"{description, pos=.5}
    ]
    \\
    X^d
    &\longmapsto&
    \\
    O^{d-1}
    \ar[
      u,
      hook',
      "{
        o
      }"{swap}
    ]
    &\longmapsto&
    \mathrm{Vect}_\infty^{\!\!
      \mathrm{Map}(O^{d-1}_{\cpt},%
      \mathcal{A})
    }
    \mathrlap{\,.}
  \end{tikzcd}
\end{equation}
In fact, \cite[Thm. 1.2]{Stefanich2025-Categorification} suggests that this construction extends to a symmetric monoidal $(\infty,d)$-functor on $d$-dimensional cobordisms and as such makes a ``fully extended topological quantum field theory'' in the sense of \cite{Lurie2008}.

But, while this is technically a TQFT in the mathematical sense, it is \emph{not yet} the description of the  (light front) time evolution of our topological quantum states in the Schr{\"o}dinger picture (cf. \cite[\S 1, Table 1]{Schreiber2009}): The latter should assign specific (Hilbert) spaces $\HilbertSpace$ of quantum states in dimension $d-1$, while $\mathbf{TopObs}(-)$ instead assigns here \emph{categories of all possible spaces of states}, namely the  representation categories of the algebras of observables, hence equivalently the \emph{Morita classes of algebras of observables}.
\footnote{
  Functorial field theories which assign algebras of observables in codim=1 (instead of spaces of quantum states) have recently been highlighted \parencites[(4.7)]{BunkSchenkel2025a}[(5.3b)]{BunkSchenkel2025b} as the Functorial Quantum Field Theory (FQFT) equivalent to the local nets of observables from Algebraic Quantum Field Theory (AQFT).
}

Also note that the coefficients of $\mathbf{TopObs}(-)$ 
are really those of a $(d+1)$-dimensional field theory, which suggests that the physical $d$-dimensional QFT should be a \emph{boundary field theory} of $\mathbf{TopObs}(-)$.

We make this precise next.

\subsubsection{The Schr{\"o}dinger picture TQFT}
\footnote{
  The observation that ``twisted'' or ``anomalous'' $(d-1)$-dimensional TFTs may be understood as natural transformations between extended $d$-dimensional TFTs we first communicated in 2006 (cf. \cite{nLab:TwistedFunctorialFieldTheory}). The insight was popularized under the name ``twisted field theory'' by \cite[\S 5]{StolzTeichner2011} and then again under the name ``relative field theory'' by \cite{FreedTeleman2014}, from where the broader theoretical physics community picked it up, now under names such as ``anomaly TFT'' or ``symmetry TFT'' (cf. \cite{ApruzziEtAl2021,nLab:TwistedFunctorialFieldTheory}). 

  The special case relevant here, of such transformations into the kind of anomaly TQFTs from \cref{TheddimExtednedTQFT}, we had discussed in \cite[\S 7.3]{Sc14-LinTypes}.
}
\nopagebreak \newline 
\indent The choice of an \emph{actual} space $\HilbertSpace({X^{d-1}})$ of topological quantum states, among the $\infty$-category $\mathbf{TopQStSpc}$ \cref{HigherSpacesOfTopQStatesAsLocalSystems} of all possible such spaces, may naturally and equivalently be encoded as a choice of an $\infty$-functor which is $\mathrm{Vect}_\infty$-linear and of the following form (therefore determined by its value on the tensor unit $\TensorUnit \in \mathrm{Vect}_\infty$):
\begin{equation}
  \begin{tikzcd}[row sep=-2pt, column sep=0pt]
    \mathrm{Vect}_\infty
    \ar[
      rr
    ]
    &&
    \mathrm{Vect}_\infty^{
      \mathrm{Map}(X^{d-1}, \mathcal{A})
    }
    \\
    \TensorUnit 
      &\longmapsto&
    \HilbertSpace(X^{d-1})
    \mathrlap{\,.}
  \end{tikzcd}
\end{equation}
But this is the form of a component morphism over objects of a natural transformation $\mathbf{TopQSts}(-)$ into the functor $\mathbf{QObs}(-)$ \cref{TheTQFTd}
out of the unit functor, $\mathbf{1}$ (constant on $\mathrm{Vect}_\infty$): 
\begin{equation}
  \mathbf{TopQSts}(-)
  \;:\;
  \begin{tikzcd}[
    sep=small
  ]
    \mathbf{1}
    \ar[
      rr
    ]
    &&
    \mathbf{TopObs}(-)
    \mathrlap{\,.}
  \end{tikzcd}
\end{equation}
The codim=0 component of $\mathbf{TopQSts}(X^d)$ is itself a linear natural transformation  (cf. \cite[Prop. 7.9]{Sc14-LinTypes}),
\begin{equation}
  \label{QStatesTransformation}
  \mathbf{TopQSts}(X^d)
  \;:\;
  \begin{tikzcd}[
    column sep=2pt
  ]
    I^{d-1}
    \ar[
      d,
      hook,
      "{ i }"
    ]
    &\longmapsto&
    \mathrm{Vect}_\infty
    \ar[
      rr,
      "{
        \TensorUnit
        \,\mapsto\,
        \HilbertSpace(I^{d-1})
      }"
    ]
    \ar[
      dd,
      equals
    ]
    &\phantom{-----}&
    \mathrm{Vect}_\infty^{%
      \mathrm{Map}(I^{d-1}_{\cpt}\!,\mathcal{A})
    }
    \ar[
      dd,
      "{
        \underset{%
        \mathrm{Map}^\ast(o,%
          \mathcal{A}%
        )%
        }{\sum}
        \circ
        \underset{%
        \mathrm{Map}^\ast(i,%
          \mathcal{A}%
        )
        }
        {\BigDelta}
      }"{description}
    ]
    \ar[
      ddll,
      Rightarrow,
      shorten=17pt,
      "{
        \widetilde{U}(X^d)
      }"{sloped}
    ]
    \\
    X^d &\longmapsto&
    \\
    O^{d-1}
    \ar[
      u,
      hook',
      "{ o }"{swap}
    ]
    &\longmapsto&
    \mathrm{Vect}_\infty
    \ar[
      rr,
      "{
        \TensorUnit
        \,\mapsto\,
        \HilbertSpace(O^{d-1})
      }"{swap}
    ]
    &&
    \mathrm{Vect}_\infty^{%
      \mathrm{Map}(O^{d-1}_{\cpt}\!,\mathcal{A})
    }
    \mathrlap{\,,}
  \end{tikzcd}
\end{equation}
determined by its tensor unit component:
\begin{equation}
  \label{UnitComponentOfTransformationComponentOfQStates}
  \widetilde{U}_{\TensorUnit}(X^d)
  : \; 
  \begin{tikzcd}[sep=small]
    \underset{%
      \mathrm{Map}^\ast(o, \mathcal{A})
    }{\sum}
    \circ
    \underset{%
      \mathrm{Map}^\ast(i, \mathcal{A})
    }{\BigDelta} 
  \HilbertSpace(I^{d-1})
  \ar[rr]
  &&
  \HilbertSpace(O^{d-1}) 
  \mathrlap{\,.}
  \end{tikzcd}
\end{equation}
Adjointly \cref{BaseChangeAdjunctionForInfinityLocalSystems}, this is a block multi-matrix (a categorified integral kernel, cf. \cite[Def. 4.12]{Sc14-LinTypes}),  
\begin{equation}
  \label{QStatesTransformationAsAKernel}
    U_{\TensorUnit}(X^d)
    \;:\;\;
    \begin{tikzcd}
      \underset{
        \mathrm{Map}^\ast(i,\mathcal{A})
      }{\BigDelta}
      \HilbertSpace(I^{d-1})
      \ar[r]
      &
      \underset{
        \mathrm{Map}^\ast(o,\mathcal{A})
      }{\BigDelta}
      \HilbertSpace(O^{d-1})
      \mathrlap{\,,}
    \end{tikzcd}
\end{equation}
indexed by asymptotic boundary conditions (cf. \cref{AsymptoticChargeSector}),
\footnote{
 Namely,  \cref{InfinityLocalSystemsAsInfinityRepresentations} shows that 
  $U_\TensorUnit(X^d)$
  \cref{QStatesTransformationAsAKernel}
  it is equivalently a block matrix of size $\pi_0\mathrm{Map}(X^{d-1}_{\mathrm{in}}, \mathcal{A}) \times \pi_0\mathrm{Map}(X^{d-1}_{\mathrm{out}}, \mathcal{A})$ (hence indexed by the asymptotic charge sectors, \cref{AsymptoticChargeSector}) of homomorphisms between $\infty$-representations, or rather a block \emph{multi-matrix} (the incidence matrix of a multi-graph) with entry sets given by the fibers of $\pi_0 \big(\mathrm{Map}(\iota_{\mathrm{in}}, \mathcal{A}) \times \mathrm{Map}(\iota_{\mathrm{out}}, \mathcal{A})\big)$: There is one linear quantum evolution operator per topological sector of field trajectories between given asymptotic topological charge sectors.

  One may consider ``quantizing further'' by summing over these charge sectors to obtain a single linear evolution map. This is considered in \cite[Def. 7.10]{Sc14-LinTypes}. While interesting in itself (\cite[Prop. 7.11]{Sc14-LinTypes}), such ``further quantization'' is not what the physics and applications suggest to do in the present case: Pushing the $\infty$-local systems on $\mathrm{Map}(X^{d-1}, \mathcal{A})$ to the point means to form the \emph{coinvariants} of the corresponding representations and thereby discards exactly the action of our topological observables. 
}
and the pasting composition of the natural transformations \cref{QStatesTransformation} corresponds to the matrix multiplication of these kernels \cref{QStatesTransformationAsAKernel}, cf. \cite[\S 6.1]{Sc14-LinTypes}. It is hereby that a choice of transformation $\mathbf{TopQSts}(-)$ in \cref{QStatesTransformation} encodes, in addition to the spaces $\HilbertSpace(X^{d-1})$ of quantum states themselves, also their quantum evolution $U(X^d)$.

Of course, a globally consistent (``anomaly free'') choice of $U(-)$ is delicate. We next consider this in the very special but important case of evolution just along diffeomorphisms. 


\subsubsection{The modular functor}
\label{TheModularFunctor}
\footnote{
  Our discussion of equivariance of topological quantum states under diffeomorphisms and mapping class groups (\emph{general covariance} in homotopy theory, cf. \cite{Dul2023}) is due to \parencites[Def. 2.11, \S 2.2]{SS25-FQH}{SS25-Srni}. 
}

We consider now the restriction of the above TQFTs from general cobordisms to those that reflect \emph{diffeomorphisms}. This is known (at least for 3d TQFT) as restricting to the \emph{modular functor}, which records the spaces of states in codim=1 and the action of diffeomorphisms on them.

If we think of $\mathbf{TopObs}(-)$ \cref{TheTQFTd} as extended to an $(\infty,d)$-functor then this means (\cite[Warning 2.2.8]{Lurie2008} \footnote{
  This is inside a \emph{warning} in \cite{Lurie2008} to highlight that, for $d > 3$, the equivariance group may be even larger than the diffeomorphism group, due to the existence of exotic h-cobordisms.  
}) to record also the higher isotopies of diffeomorphisms; 
namely to record the topological action of the diffeomorphism group $\mathrm{Diff}(-)$, whence the moduli space of charges is promoted to the \emph{homotopy quotient} (aka \emph{Borel construction})
\begin{equation}
  \mathrm{Map}\big(
    X^{d-1},
    \mathcal{A}
  \big)
  \sslash 
  \mathrm{Diff}\big(X^{d-1}\big)
  \;\sim\;
  \mathrm{Map}\big(
    X^{d-1},
    \mathcal{A}
  \big)
  \;\;\;
  \underset{
    \mathclap{\mathrm{Diff}(X^{d-1})}
  }{\otimes}
  \;\;\;
  E \mathrm{Diff}\big(X^{d-1}\big)
  \mathrlap{\,.}
\end{equation}
The correspondingly \emph{generally covariantized} quantum observables are hence
\begin{equation}
  \begin{aligned}
  \mathrm{TopObs}_0
  & \defneq 
  \;
  \mathbb{C}\bigg[
  \pi_1
  \Big(
  \mathrm{Map}\big(
    X^{d-1},
    \mathcal{A}
  \big)
  \sslash 
  \mathrm{Diff}\big(X^{d-1}\big)
  \Big)
  \bigg]
  \\
  &
  \underset{\mathclap{
    \scalebox{.6}{\begin{minipage}{35pt}\cite[P. 2.24]{SS25-FQH}\end{minipage}}}
  }{\simeq}
  \;
  \mathbb{C}\bigg[
  \Big(
  \pi_1
  \mathrm{Map}\big(
    X^{d-1},
    \mathcal{A}
  \big)
  \Big)
  \rtimes
  \Big(
  \pi_0
  \mathrm{Diff}\big(X^{d-1}\big)
  \Big)
  \bigg]
  \mathrlap{\,,}
  \end{aligned}
\end{equation}
where on the right we have the \emph{mapping class group}
\begin{equation}
  \mathrm{MCG}\big(X^{d-1}\big)
  :=
  \pi_0 \mathrm{Diff}\big(
    X^{d-1}
  \big)
  \,.
\end{equation}
This entails that the generally covariantized quantum states are still representations of the fundamental group of the charge moduli as in \cref{SpacesOfTopQStatesAsLocalSystems}, but now need to admit and be equipped with extensions to the semidirect action of the mapping class group.

\section{Application to Topological Quantum Materials}
\label{ExamplesAndApplications}

Here we connect the general theory of topological quantization of higher gauge fields, as laid out in \cref{GlobalCompletionOfHigherGaugeFields,QuantizationInTopologicalSector}, to experimental physics and questions of engineering of topologically ordered quantum materials.

\subsection{Canonical Flux Observables}
\label{OnCanonicalFluxObservables}
\footnote{
  On basics of Poisson brackets of observables in field theory, derived from Lagrangian densities, see for instance \cite[\S 9]{Nastase2019}.
  The role of Pontrjagin algebras in uniformly reformulating the quantization of topological flux observables was indicated in \cite{SS24-Obs}.
}

First, to demonstrate backwards compatibility, we begin by showing how our general formula \cref{OrdinaryTopObsAsGroupAlgebras} for ordinary topological flux observables  reproduces quantization of traditional ``canonical'' Poisson brackets of topological flux observables --- first for 4D Yang-Mills theory (\cref{OnMaxwellFluxObservables}), then for 5D Maxwell-Chern-Simons theory \cref{OnMaxwellCSFluxObservables}. The case of 11D MSC theory follows analogously.

\subsubsection{Yang-Mills Flux Observables}
\label{OnMaxwellFluxObservables}
\footnote{
  The Poisson brackets of fluxes in ordinary Yang-Mills theory (\cref{PoissonBracketsOfYMFluxes}) are due to \cite[Thm 1.1, \S A.1]{SS24-Obs}, building on the key observation of \cite{CattaneoPerez2017}. The non-perturbative quantization of these brackets in the topological sector is \cite[Thm. 1.2]{SS24-Obs}, by appeal to Rieffel deformation quantization \parencites[\S 2]{Hawkins2008}{nLab:CStarAlgebraicDeformationQuantization} of Lie-Poisson structures by group convolution algebras \parencites{Rieffel1990}[Ex. 11.1]{LandsmanRamazan2001}[Ex. 2]{Landsman1999}.
}

We begin with comparing to the quantization of fluxes in ordinary Yang-Mills theory and show that for the special case of Maxwell theory (to which our setup applies, by \cref{NonlinearFluxDensities}) a traditional quantization of the topological observables coincides with the result of our general formula \cref{OrdinaryTopObsAsGroupAlgebras}.
To that end, consider:
\begin{itemize}
\item
$\mathfrak{g}$ a Lie algebra with $\mathrm{ad}$-invariant metric $\langle -,-\rangle$,

\item
$\Sigma^2$ a closed orientable surface, not necessarily connected,

\item
$X^{1,3} := \mathbb{R}^{1,1} \times \Sigma^2$
a 4D spacetime of the form \cref{TransverseSpace},
\end{itemize}
and consider $\mathfrak{g}$-valued functions
$
  \alpha \in C^\infty(\Sigma, \mathfrak{g})
$
as observables of $\mathfrak{g}$-Yang-Mills flux through $\Sigma$ given by evaluating flux densities $B_2, E_2$ inside 
\begin{equation}
  \begin{tikzcd}[
    sep=14pt
  ]
  \int_\Sigma \langle \alpha, -\rangle
  :
  \Omega^2_{\mathrm{dR}}\big( \mathbb{R}^1 \times \Sigma; \mathfrak{g} \big)
  \ar[rr]
  &&
  \mathbb{R}
  \ar[r, hook]
  &
  \mathbb{C}\;.
  \end{tikzcd}
\end{equation}

A careful analysis of the standard Poisson brackets of fluxes in $\mathfrak{g}$-Yang-Mills theory reveals the following fact (where $\mathfrak{g}_{\hbar}$ denotes $\mathfrak{g}$ with Lie bracket rescaled by $\hbar \in \mathbb{R}_{>0}$):
\begin{proposition}
  \label[proposition]
    {PoissonBracketsOfYMFluxes}
  The phase space of electromagnetic fluxes through $\Sigma$ in $\mathfrak{g}$-Yang-Mills theory is the Lie-Poisson manifold corresponding to the Fr{\'e}chet Lie algebra
  \begin{equation}
    C^\infty\big(
      \Sigma
      ,
      \mathfrak{g}_{\hbar} 
      \rtimes_{{}_{\mathrm{ad}}}
      \mathfrak{g}_0
    \big)
    = 
    C^\infty\big(
      \Sigma,
      \mathfrak{g}_{\hbar}
    \big)
    \rtimes_{{}_{\mathrm{ad}}}
    C^\infty\big(
      \Sigma,
      \mathfrak{g}_{0}
    \big)
  \end{equation}
  \textup{(where $\mathfrak{g}_{\hbar}$ denotes $\mathfrak{g}$ with Lie bracket rescaled by $\hbar \in \mathbb{R}$)}.
\end{proposition}

Non-perturbative quantization of such brackets is given by a choice of Lie group $G \rtimes (\mathfrak{g}_0/\Lambda)$ integrating $\mathfrak{g}_\hbar \rtimes_{{}_{\mathrm{ad}}} \mathfrak{g}_0$ (hence with $\Lambda \subset \mathfrak{g}$ a $\mathrm{G}_{\mathrm{ad}}$ invariant lattice) with the quantum observables forming the  convolution $C^\ast$-algebra of the Fr{\'e}chet Lie group
\begin{equation}
    C^\infty\big(
      \Sigma
      ,
      G
      \rtimes
      \mathfrak{g}_0/
      \Lambda
    \big)
    = 
    C^\infty\big(
      \Sigma,
      G
    \big)
    \rtimes
    C^\infty\big(
      \Sigma,
      \mathfrak{g}_{0}/\Lambda
    \big)
    \,.
\end{equation}
Given this, the topological quantum observables therefore form the group algebra (cf. \cite[p. 51,197]{Balachandran2010}) of the group of connected components:
\begin{equation}
  \label{TopologicalObservablesOfYMFluxes}
  \mathrm{TopObs}
  =
  \mathbb{C}\big[
    \pi_0
    \,
    C^\infty\big(
      \Sigma,
      G
      \rtimes
      \mathfrak{g}_{0}/\Lambda
    \big)
  \big]
    \,.
\end{equation}

Specializing this to electromagnetism (Maxwell theory) with $\mathfrak{g} = \mathfrak{u}(1) \simeq \mathbb{R}$ an evident choice of integrating group is
\begin{equation}
  \mathrm{U}(1) \times \mathbb{R}/\mathbb{Z}
  \simeq
  \mathrm{U}(1) \times \mathrm{U}(1)
\end{equation}
in which case the topological observables \cref{TopologicalObservablesOfYMFluxes} become:
\begin{equation}
  \begin{aligned}
    \mathrm{TopObs}
    &
    \simeq
    \mathbb{C}\Big[
      \pi_0\,
      C^\infty\big(
        \Sigma,
        \mathrm{U}(1)^2
      \big)
    \Big]
    \\
    &
    \simeq
    \mathbb{C}\Big[
      \pi_0\,
      \mathrm{Map}\big(
        \Sigma,
        \mathrm{U}(1)
      \big)^2
    \Big]    
    \\
    &
    \simeq
    \mathbb{C}\Big[
      \pi_1\,
      \mathrm{Map}^\ast\big(
        \Sigma_{\cpt},
        B \mathrm{U}(1)
      \big)^2
    \Big]    
    \mathrlap{\,.}
  \end{aligned}
\end{equation}
This result for the topological observables on Maxwell flux, derived via an established quantization procedure, coincides with our general formula \cref{SolTopObsInDegreeZeroAsGroupAlgebra} for the choice of flux quantization of Maxwell theory from \cref{SomeFluxQuantizationHypotheses}.

For example, over the torus $\Sigma = T^2$, the magnetic sector gives:
\begin{equation}
  \label{OrdinaryMagneticObservablesOnTorus}
  \begin{aligned}
  \pi_1\, 
  \mathrm{Map}^\ast\big(
    T^2_{\cpt},
    B \mathrm{U}(1)
  \big)
  & \simeq
  \pi_0
  \,
  \mathrm{Map}^\ast\big(
    T^2_{\cpt},
    \pi_1 \, B \mathrm{U}(1)
  \big)
  \\
  & \simeq
  \pi_0
  \,
  \mathrm{Map}\big(
    T^2,
    \mathrm{U}(1)
  \big)
  \\
  & \simeq
  \pi_0
  \,
  \mathrm{Map}^\ast\big(
    T^2,
    B \mathbb{Z}
  \big)
  \\
  & \simeq
  H^1(T^2;\, \mathbb{Z})
  \\
  & \simeq
  \mathbb{Z}_a \times \mathbb{Z}_b
  \mathrlap{\,.}
  \end{aligned}
\end{equation}
The generators 
\begin{equation}
  \begin{tikzcd}[sep=small]
    W_{a/b}
    \coloneqq
    1
    \in 
    \mathbb{Z}_{a/b}
    \ar[r, hook]
    &
    \mathbb{Z}_a \times \mathbb{Z}_b
    \ar[r, hook]
    &
    \mathbb{C}\big[
      \mathbb{Z}_a \times \mathbb{Z}_b
    \big]
  \end{tikzcd}
\end{equation}
may be understood as the \emph{Wilson loop observables} for a pair of basis 1-cycles on the torus. In Maxwell theory,
these commute with each other.

Next, we see how this situation changes as the classifying space $B \mathrm{U}(1)$ is deformed (in \cref{OnSolitonicStringsIn5D}).

\subsubsection{Maxwell-CS Flux Observables}
\label{OnMaxwellCSFluxObservables}
\footnote{
  The Lagrangian density of 5D Maxwell-Chern-Simons theory is discussed in the context of minimal 5D SuGra in \cite[(0.1)]{ChamseddineNicolai2018} and in the context of quantum electrodynamics in \cite[(2)]{Hill:2006}. Our observation here on the Poisson structure among the canonical flux observables is a straightforward consequence, but may not be citable from the literature.
}

The local Lagrangian density for 5D Maxwell-Chern-Simons theory (\cref{HigherGaugeFields}) is
\begin{equation}
  L
  =
  \tfrac{1}{2}
  F_2 \wedge \star F_3
  -
  \tfrac{1}{6}
  A \wedge \wedge F_2 \wedge F_2 
  \mathrlap{\,,}
  \;\;\;
  \mbox{where}\;
  \left\{
  \begin{aligned}
    F_2 & := \mathrm{d}A
    \\
    F_3 & :=  \star F_2
    \mathrlap{\,,}
  \end{aligned}
  \right.
\end{equation}
whence the canonical momentum to the field observable $A$ on a Cauchy surface $\begin{tikzcd}[sep=small]\iota \colon X^d \ar[r, hook] & \mathbb{R}^{1,0} \times X^d\end{tikzcd}$ \cref{InclusionOfCauchySurface} is, cf. \cref{LocalSolutionsViaInitialValues}:
\begin{equation}
  \begin{aligned}
  \widetilde B_3
  & :=
  \frac{\partial L}{\partial(\partial_0 A)}
  \\
  & =
  B_3
  -
  \tfrac{1}{3} A \wedge B_2
  \mathrlap{\,,}
  \end{aligned}
  \quad 
  \mbox{where}
  \;\;
  \left\{
  \begin{aligned}
    B_2 & = \iota^\ast F_2
    \\
    B_3 & = \iota^\ast F_3
    \mathrlap{\,.}
  \end{aligned}
  \right.
\end{equation}
with Poisson bracket
\begin{equation}
  \bigg\{
  {\int_{X^4}}
  \omega_1 
    \wedge 
  \widetilde B_3
  \; ,\,
  A(x)
  \bigg\}
  =
  \omega_1(x)
  \,,
  \;\;\;\;\;
  \mbox{for}\;
  \omega_1 \in \Omega^1_{\mathrm{dR}}(X^d)
  \mathrlap{\,.}
\end{equation}
The gauge invariant momentum observable is however
\begin{equation}
  B_3 
  =
  \widetilde B_3 + \tfrac{1}{3}A \wedge B_2
  \,,
\end{equation}
which therefore has a non-trivial Poisson bracket with itself: 
\begin{equation}
  \label{FluxMomentumSelfCommutatorin5DMSC}
  \bigg\{
  \int_{X^4}
  \,
  \omega_1 \wedge B_3
  \;,\,
  \int_{X^4}
  \,
  \omega' \wedge B_3
  \bigg\}
  =
  \tfrac{4}{3}
  \int_{X^4}
  \omega_1 \wedge \omega'
  \wedge B_2
  \mathrlap{\,.}
\end{equation}

The \emph{topological} flux observables should be those for which the smearing form $\omega$ is closed,
\begin{equation}
  \mathrm{d} \omega_p \overset{!}{=} 0
  \mathrlap{\,.}
\end{equation}
On these, the remaining mixed Poisson bracket vanishes:
\begin{equation}
  \bigg\{
  \int_{X^4}
  \omega_1 \wedge B_3
  \;,\,
  \int_{X^4}
  \omega_2 \wedge B_2
  \bigg\}
  =
  \int_{X^4}
  \omega_2 \wedge 
  \grayunderbrace{\mathrm{d}\omega_1}{=0}
  \mathrlap{\,,}
\end{equation}
so that the topological Poisson structure \cref{FluxMomentumSelfCommutatorin5DMSC} is that of a Heisenberg algebra. 
\footnote{
  As we further pass to the connected component of topological observables understood as the de Rham class their closed smearing coefficient $\omega$, then \cref{FluxMomentumSelfCommutatorin5DMSC} is a Heisenberg extension of $H^1_{\mathrm{dR}}(X^d)$ by $H^2_{\mathrm{dR}}(X^d)$.
  On a 2-torus, the bracket structure \cref{FluxMomentumSelfCommutatorin5DMSC} coincides with the degree=0 Lie algebra of the \emph{toroidified} $L_\infty$-algebra $\mathrm{tor}^2 \mathfrak{l} S^2$, cf. \parencites[p. 10]{SatiVoronov2025}[(83)]{GSS25-TD}.
}

After an electric/magnetic flux quantization law is imposed (which is unclear how to properly do in the traditional ``canonical'' formalism), we are therefore to expect that the topological quantum observables on fluxes in 5D MSC theory form a discrete Heisenberg algebra. This is exactly what we will find below in \cref{AnyonsOnTheTorus} via the complete construction of flux-quantized flux quantum observables from \cref{QuantizationInTopologicalSector}.

A completely analogous argument, just up to degree shifts, applies to the C-field gauge sector of 11D SuGra, and again we find the suggested discrete Heisenberg algebra structure brought out, below in \cref{OnSolitonic5Branes}, by the complete construction of flux-quantized flux quantum observables from \cref{QuantizationInTopologicalSector}.

\subsection{Anyons in FQ(A)H-Materials}
\label{OnFQAHMaterials}
\footnote{
  The construction of topologically ordered anyonic quantum states of FQ(A)H systems via flux quantization in 2-Cohomotopy (Hypothesis h in  \cref{SomeFluxQuantizationHypotheses}) is due to 
  \parencites{SS25-FQH}{SS25-AbelianAnyons}{SS25-FQAH}{SS25-CrystallineChern}, exposition is in \parencites{SS25-ISQS29}{SS25-WilsonLoops}.
}

After a long history as a theoretical curiosity in mathematical physics, (abelian) \emph{anyons} (cf. \cite{Goldin2023,nLab:BraidGroupStatistics}) have consistently been experimentally reported in recent years (starting with \cite{Nakamura2020}, see \cite{Ghosh2025} for the state of the art), in the guise of surplus magnetic flux quanta (and their induced ``quasi-hole'' vortices) in cold 2-dimensional electron gases penetrated by transverse magnetic flux with an exact integer (or rational) multiple of background flux quanta per electron (cf. \cref{FQHAnyons}). 

\begin{figure}[htb]
\caption{%
  \label{FQHAnyons}%
  Anyons in fractional quantum Hall (FQH) systems are (quasi-hole vortices in a 2D electron gas induced by) surplus magnetic flux quanta on top of an exact \emph{filling fraction} of some integer/rational number $K$ of flux quanta per electron. The adiabatic braiding of their worldlines causes the quantum state of the whole system to change by a global phase. This phenomenon is experimentally well established, but challenges traditional theory which does not readily apply to strongly-coupled and global topological phenomena.  
}
\centering
\adjustbox{
  rndfbox=4pt
}{
\adjustbox{
  raise=-1.5cm,
  scale=.9,
}{
\hspace{-.5cm}
\begin{tikzpicture}[
  scale=.75
]

\node
  at (.3,.55+.8)
  {
    \adjustbox{
      bgcolor=white,
      scale=.7
    }{
      \color{darkblue}
      \bf
      \def\arraystretch{.9}
      \begin{tabular}{c}
        surplus
        \\
        flux quantum:
        \\
        \color{purple}
        quasi-hole
        \\
        \color{purple}
        vortex
      \end{tabular}
    }
  };

\draw[
  dashed,
  fill=lightgray
]
  (0,0)
  -- (5,0)
  -- (7+.3-.1,2+.3)
  -- (2.8+.3+.1,2+.3)
  -- cycle;

\begin{scope}[
  shift={(2.4,.5)}
]
\shadedraw[
  draw opacity=0,
  inner color=olive,
  outer color=lightolive
]
  (0,0) ellipse (.7 and .3);
\end{scope}

\begin{scope}[
  shift={(4.5,1.5)}
]

\begin{scope}[
 scale=1.8
]
\shadedraw[
  draw opacity=0,
  inner color=olive,
  outer color=lightolive
]
  (0,0) ellipse (.7 and .25);
\end{scope}

\begin{scope}[
 scale=1.45
]
\shadedraw[
  draw opacity=0,
  inner color=olive,
  outer color=lightolive
]
  (0,0) ellipse (.7 and .25);
\end{scope}

\shadedraw[
  draw opacity=0,
  inner color=olive,
  outer color=lightolive
]
  (0,0) ellipse (.7 and .25);

\begin{scope}[
  scale=.2
]
\draw[
  fill=black
]
  (0,0) ellipse (.7 and .25);
\end{scope}

\end{scope}

\draw[
  white,
  line width=2
]
  (1.3, 1.8) .. controls 
  (2,2.2) and 
  (2.2,1.5) ..
  (2.32,.7);
\draw[
  -Latex,
  black!70
]
  (1.3, 1.8) .. controls 
  (2,2.2) and 
  (2.2,1.5) ..
  (2.32,.7);

\node
  at (1.3,2.7)
  {
    \adjustbox{
      scale=.7
    }{
      \color{darkblue}
      \bf
      \def\arraystretch{.9}
      \def\tabcolsep{-5pt}
      \begin{tabular}{c}
        $K$ flux-quanta
        absorbed
        \\
        by each electron:
      \end{tabular}
    }
  };

\draw[
 line width=2.5pt,
  white
]
  (2.4, 3.1) .. controls 
  (2.8,3.3) and 
  (4,3.5) ..
  (4.3,1.8);

\draw[
  -Latex,
  black!70
]
  (2.4, 3.1) .. controls 
  (2.8,3.3) and 
  (4,3.5) ..
  (4.3,1.8);

\node at 
  (5.9,2.6)
  {
   \scalebox{.8}{
     \color{gray}
     (cf. \cite[Fig. 16]{Stormer99})  
   }
  };

\node[
  gray,
  rotate=-20,
  scale=.73
] 
  at (4.8,+.3) {$\Sigma^2$};

\end{tikzpicture}
}
\hspace{-1cm}
  \adjustbox{raise=-1.4cm}{
  \begin{tikzpicture}[
    xscale=.7
  ]
    \draw[
      gray!30,
      fill=gray!30
    ]
      (-4.6,-1.5) --
      (+1.8,-1.5) --
      (+1.8+3-.5,-.4) --
      (-4.6+3+.5,-.4) -- cycle;

    \begin{scope}[
      shift={(-1,-1)},
      scale=1.2
    ]
    \shadedraw[
      draw opacity=0,
      inner color=olive,
      outer color=lightolive
    ]
      (0,0) ellipse (.7 and .1);
    \end{scope}

    \draw[
     line width=1.4
    ]
      (-1,-1) .. controls
      (-1,0) and
      (+1,0) ..
      (+1,+1);

  \begin{scope}
    \clip 
      (-1.5,-.2) rectangle (+1.5,1);
    \draw[
     line width=7,
     white
    ]
      (+1,-1) .. controls
      (+1,0) and
      (-1,0) ..
      (-1,+1);
  \end{scope}
  
    \begin{scope}[
      shift={(+1,-1)},
      scale=1.2
    ]
    \shadedraw[
      draw opacity=0,
      inner color=olive,
      outer color=lightolive
    ]
      (0,0) ellipse (.7 and .1);
    \end{scope}
    \draw[
     line width=1.4
    ]
      (+1,-1) .. controls
      (+1,0) and
      (-1,0) ..
      (-1,+1);

  \node[
    rotate=-25,
    scale=.7,
    gray
  ]
    at (3.2,-.58) {
      $\Sigma^2$
    };

  \draw[
    -Latex,
    gray
  ]
    (-3.4,-1.35) -- 
    node[
      near end, 
      sloped,
      scale=.7,
      yshift=7pt
      ] {time}
    (-3.4, 1.2);

  \node[
    scale=.7
  ] at 
    (0,-1.3)
   {\bf \color{darkblue} 
   surplus flux quanta};

  \node[
    scale=.7
  ] at 
    (1.5,0)
   {\bf \color{darkgreen} braiding};

  \node[
    fill=white,
    scale=.8
  ] at (-2.5,-1) {$
    \vert \Psi \rangle
  $};

  \node[
    fill=white,
    scale=.8
  ] at (-2.5,.7) {$
    e^{\tfrac{\pi \mathrm{i}}{K}}
    \vert \Psi \rangle
  $};

  \draw[
    |->,
    black!80,
    line width=.5
  ]
    (-2.5,-.6) --
    (-2.5, .3);

  \end{tikzpicture}
  }
}
\end{figure}
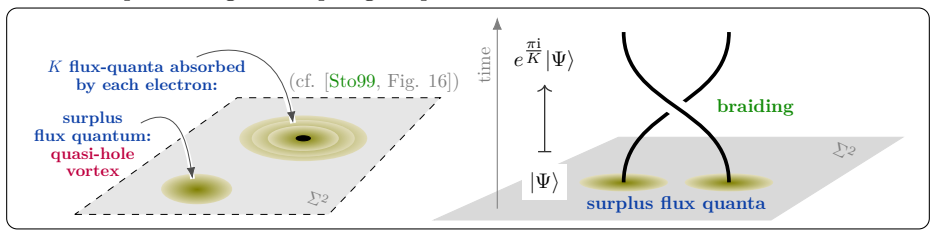

These \emph{fractional quantum Hall systems} (FQH, cf. \cite{Stormer99,nLab:QuantumHallEffect}) have thereby become the leading (currently the only) known hardware platform potentially supporting topological quantum operations plausibly necessary for non-trivial quantum computing technology. However, for that application one needs more than the mere \emph{presence} of solitonic anyons that has so far been established, one will crucially need fine-grained \emph{control} over their movement and hence a detailed understanding of their fundamental nature.

However, the theoretical understanding of these FQH anyons has remained immature. The traditional description by effective Chern-Simons field theory is fraught with internal inconsistency at the global level and in general relies on an outdated hierarchical picture of FQH anyons (cf. \cite[Rems. A.1-3]{SS25-FQH}).

Turning this situation right-side-up, since FQH anyons are, remarkably, magnetic flux quanta (in a strongly coupled electron gas environment), the description of their quantum topology calls for a globally consistent description of flux of the kind we have discussed in \cref{GlobalCompletionOfHigherGaugeFields}. By the results there, we are thus to ask for a classifying space $\mathcal{A}$ which serves as a deformation of the ordinary magnetic classifying space $B \mathrm{U}(1)$ from \cref{OnMaxwellFluxObservables}, now modelling the effective behavior of surplus flux in FQH systems.

We turn to $\mathcal{A} = S^2$ the 2-sphere.

\subsubsection{Solitonic strings in 5D}
\label{OnSolitonicStringsIn5D}
\footnote{
  An old suggestion that strings in the $\mathrm{AdS}_5$-factor of type IIB 10D SuGra on $\mathrm{AdS}_5 \times S^5$ may be higher-dimensional analogs of abelian anyons is given in \cite[\S 3]{Hartnoll2006}.
  We follow the argument in \cite[\S 3]{SS25-WilsonLoops}, for 5D supergravity, which is at least nominally different.
}

We are going to model the anyonic solitons in FQH materials as the KK-compactification of anyonic strings in 5D Maxwell-Chern-Simons theory flux-quantized according to \emph{Hypothesis h} (\cref{Hypothesish}).

By the general discussion of solitonic branes in \cref{OnSolitonicBranes}, this means that we are to take our spatial domain \cref{TheSpatialDomain} to be 
\begin{equation}
  \label{SpatialDomainForSolitonicStrings}
  \begin{aligned}
    X^4_{\mathrm{dom}}
    &
    :=
    \mathbb{R}^1_{\plus}
    \wedge
    \big(
      \mathbb{R}^1 \times \Sigma^2
    \big)_{\cpt}
    \,,
  \end{aligned}
\end{equation}
whence the string-solitonic topological quantum observables (in the 0-charge sector) are \cref{OrdinaryTopObsAsGroupAlgebras}:
\begin{equation}
  \begin{aligned}
    \mathrm{TopObs}_0
    & \simeq
    \mathbb{C}\bigg[
    \pi_0
    \,
    \mathrm{Map}^\ast\Big(
      \mathbb{R}^1_{\plus}
      \wedge
      \big(
        \mathbb{R}^1 \times \Sigma^2
      \big)_{\cpt},
      \mathcal{A}
    \Big)
    \bigg]
    \\
    & \simeq
    \mathbb{C}\bigg[
    \pi_0
    \,
    \mathrm{Map}^\ast\Big(
      \big(
        \mathbb{R}^1 \times \Sigma^2
      \big)_{\cpt},
      \mathcal{A}
    \Big)    
    \bigg]
    \\
    & \simeq
    \mathbb{C}
    \bigg[
    \pi_0
    \,
    \mathrm{Map}^\ast\Big(
      \mathbb{R}^1_{\cpt} 
        \wedge 
      \Sigma^2_{\cpt}
      ,\,
      \mathcal{A}
    \Big)    
    \bigg]
    \\
    & \simeq
    \mathbb{C}\Big[
    \pi_1
    \,
    \mathrm{Map}^\ast\big(
      \Sigma^2_{\cpt}
      ,\,
      \mathcal{A}
    \big)    
    \Big]
    \,.
  \end{aligned}
\end{equation}
Moreover, if $\Sigma^2$ is already compact, then $\Sigma^2_{\cpt} = \Sigma^2_{\plus}$ and the above reduces to
\begin{equation}
  \mbox{$\Sigma^2$ compact}
  \;\;\;
  \Rightarrow
  \;\;\;
  \mathrm{TopObs}_0
  \simeq
  \mathbb{C}\Big[
    \pi_1
    \,
    \mathrm{Map}\big(
      \Sigma^2
      ,\,
      \mathcal{A}
    \big)
  \Big]
  \mathrlap{\,.}
\end{equation}

\subsubsection{Flux quantization in 2-Cohomotopy}
The 2-sphere classifying space $\mathcal{A} := S^2$ here is usefully regarded as the 2-skeleton of the ordinary magnetic classifying space:
\begin{equation}
  \label{SphereAs2SkeletonOnB2Z}
  \mathcal{A}
  :=
  \begin{tikzcd}
    S^2
    \ar[
      r, 
      hook
    ]
    &
    B \mathrm{U}(1)
    \mathrlap{\,.}
  \end{tikzcd}
\end{equation}
Given that we are looking for a global completion which  at least \emph{locally} mimics abelian Chern-Simons theory, we may already see that this is a promising choice by noting that the further Bianchi identity encoded by the corresponding Whitehead $L_\infty$-algebra \cref{CElOfEvenDimensionalSphere},
\begin{equation}
  \mathfrak{a}
  =
  \begin{tikzcd}
  \mathfrak{l}S^2
  \ar[r]
  &
  \mathfrak{l}B\mathrm{U}(1)
  \mathrlap{\,,}
  \end{tikzcd}
\end{equation}
is exactly the relation characterizing an abelian Chern-Simons form ${H_3}$ for a given abelian flux form $B_2$:
\begin{equation}
  \begin{tikzcd}[
    sep=0pt,
    ampersand replacement=\&
  ]
    \Omega^1_{\mathrm{dR}}(-; \mathfrak{l}S^2)
    \ar[rr]
    \&\&
    \Omega^1_{\mathrm{dR}}(-; \mathfrak{l}B \mathrm{U}(1))
    \\
    \left\{
    \begin{aligned}
      \mathrm{d}\, B_2 & = 0
      \\
      \mathrm{d}\, H_3 
      & = B_2 \wedge B_2
    \end{aligned}
    \right\}
    \&\longmapsto\&
    \big\{
      \mathrm{d}\, B_2  = 0
    \big\}
    \mathrlap{\,.}
  \end{tikzcd}
\end{equation}

In the following, we may just take this choice $\mathcal{A} \defneq S^2$ at face value and find its consequences (as in \cite{SS25-FQH}).
But according to \cref{HigherGaugeFields}
we may also think of this step as passing from 4D Maxwell theory to 5D Maxwell-Chern-Simons theory, the latter flux-quantized according to ``Hypothesis h'' from \cref{SomeFluxQuantizationHypotheses}, and thought of as dimensionally reduced to 3D abelian Chern-Simons theory according to \cite[\S 3]{SS25-WilsonLoops}.

\subsubsection{FQH Anyons on the Torus}
\label{AnyonsOnTheTorus}
\footnote{
  The result that $\pi_1\, \mathrm{Map}(T^2, S^2)$ is an integer Heisenberg group is due to \cite{Hansen1974}, and that it is specifically the one at level=2 is due to \cite{LarmoreThomas1980}, see also \cite[Prop. 1.5, Cor. 7.6]{Kallel2001}. 
  (Here we give a much shorter and more transparent re-derivation of this fact, due to \cite{KSS26-HigherAnyons}.) 
  The observation that this relates to quantum observables of anyons on the torus is due to \cite{SS25-FQH}.
}

Remarkably, this deformation of the classifying space $\begin{tikzcd}[sep=small] \mathcal{A} := S^2 \ar[r, hook] & B \mathrm{U}(1)\end{tikzcd}$ \cref{SphereAs2SkeletonOnB2Z} deforms the situation \cref{OrdinaryMagneticObservablesOnTorus} of the magnetic flux observables on the torus as follows:
\begin{proposition}
  \label[proposition]
    {FundamentalGroupOfMapsFromTorusToSphere}
  The fundamental group of maps from the torus to the sphere
  \begin{equation}
  \label{IdentifyingFundamentalGroupOfMapsFromTorusToSphere}
    \pi_1\, \mathrm{Map}\big(T^2, S^2\big)
    \simeq
    \widehat{\mathbb{Z}^2}
  \end{equation}
  is the \emph{integer Heisenberg group} at level=2.
\end{proposition}
\begin{definition}[Integer Heisenberg group at level=2, {cf. \parencites[Ex. 11 on p. 35]{DummitFoote2003}{nLab:IntegerHeisenbergGroup}}]
\label[definition]{IntegerHeisenbergGroupAtLevel2}
The underlying set is
\begin{equation}
  \label{UnderlyingSetOfIntegerHeisenbergGrp}
  \widehat{\mathbb{Z}^2}
  :=_{\mathrm{Set}}
  (\mathbb{Z}_a \times \mathbb{Z}_b) \times \mathbb{Z}
\end{equation}
with generating elements
\begin{equation}
  \label{GeneratorsOfHeisenbergGrp}
  \left.
  \begin{aligned}
    W_a & := \big((1,0), 0\big)
    \\
    W_b & := \big((0,1), 0\big)
    \\
    \zeta & := \big((0,0), 1\big)
  \end{aligned}
  \right\}
  \in
  (\mathbb{Z}_a \times \mathbb{Z}_b) 
    \times \mathbb{Z}
\end{equation}
on which the only non-trivial group commutator is
\[
  [W_a, W_b] = \zeta^2
  \mathrlap{.}
\]
\end{definition}

We spell out a proof of \cref{FundamentalGroupOfMapsFromTorusToSphere}. To that end, recall the following facts from basic homotopy theory:
\begin{enumerate}

\item The \emph{Samelson product} $[ -,-]_{\mathrm{Sam}}$ on a loop space $\Omega X$ (being an \emph{H-group}) is the ``H-group commutator'' descended from the Cartesian product to the smash product (cf. \cite[\S X.5]{Whitehead1978}):
\begin{equation}
  \label{SamelsonProductOnLoopSpace}
  \begin{tikzcd}[
    row sep=0pt,
    column sep=25pt
  ]
  \;\;  (\ell_2, \ell_1)
 \quad   \ar[
      rr,
      shorten=5pt,
      |->
    ]
    &&
    (\ell_2 
      \star 
    \ell_1)
      \star
    (\overline{\ell_2} \star
    \overline{\ell_1})
    \\
    (\Omega X)
    \times 
    (\Omega X)
    \ar[
      rr,
      "{
        \star
      }"
    ]
    \ar[
      d,
      ->>
    ]
    &&
    \Omega X
    \\[15pt]
    (\Omega X)
    \wedge
    (\Omega X)
    \ar[
      urr,
      "{
        [ -,- ]_{\mathrm{Sam}}
      }"{swap, sloped, pos=.4}
    ]
  \end{tikzcd}
\end{equation}

\item 
The Samelson product induced on $\pi_\bullet(\Omega X)$ relates to the \emph{Whitehead bracket} $[-,-]_{\mathrm{Wh}}$ on $\pi_{\bullet+1}(X)$ as (cf. \cite[Thm. 7.10 on p. 476]{Whitehead1978})
\begin{equation}
  \label{WhiteheadProductAsSamelsonProduct}
  \big[
    \widetilde{\alpha}_1
    ,\,
    \widetilde{\alpha}_2 
    \,
  \big]_{\mathrm{Sam}}
  = 
  (-1)^{\mathrm{deg}(\alpha_1)} 
  \,
  \widetilde{%
    [\alpha,\alpha_2]%
  }_{\mathrm{Wh}}
  \mathrlap{\,,}
\end{equation}
where $
  \begin{tikzcd}[]
    \pi_{\bullet + 1}(X)
    \ar[
      <->,
      r, 
      "{\sim}"{swap},
      "{ \widetilde{(-)} }"
    ]
    &
    \pi_\bullet(\Omega X)
  \end{tikzcd}
$ is induced from the hom-isomorphism of the $\Sigma \dashv \Omega$-adjunction on pointed spaces:
\begin{equation}
  \label{SuspensionLoopingHomIsomorphism}
  \Big\{
  \begin{tikzcd}
    \Sigma Y 
      \ar[
        r,
        dashed,
        "{
          f
        }"
      ]
    &
    X
  \end{tikzcd}
  \Big\}
  \begin{tikzcd}
    {}
    \ar[
      r,
      <->,
      "{ \widetilde{(-)} }",
      "{ \sim }"{swap}
    ]
    &
    {}
  \end{tikzcd}
  \Big\{
  \begin{tikzcd}
    Y 
      \ar[
        r,
        dashed,
        "{
          \widetilde{f}
        }"
      ]
    &
    \Omega X
  \end{tikzcd}
  \Big\}
  \mathrlap{\,.}
\end{equation}

\item 
The Whitehead bracket of the generator $\iota_n := 1 \in \mathbb{Z} \simeq \pi_n(S^n)$ satisfies (cf. \cite[Thm. 2.5 on p. 495]{Whitehead1978}):
\begin{equation}
\label{WhiteheadSquareRootsOfHopfFibrations}
\begin{aligned}
  \big[
    \iota_2, \iota_2
  \big]_{\mathrm{Wh}}
  & 
    = 2 [h_{\mathbb{C}}]
    =: 2 \eta
  \\
  \big[
    \iota_4, \iota_4
  \big]_{\mathrm{Wh}}
  & 
    = 2 [h_{\mathbb{H}}]
    =: 2 \nu
  \\
  \big[
    \iota_8, \iota_8
  \big]_{\mathrm{Wh}}
  & 
    = 2 [h_{\mathbb{O}}]
    =: 2 \sigma
  \mathrlap{\,,}
\end{aligned}
\end{equation}
where $\begin{tikzcd}[sep=small] h_\mathbb{K} : S(\mathbb{K}^2) \ar[r] & \mathbb{K}P^1\end{tikzcd}$ denotes the $\mathbb{C}$/$\mathbb{H}$/$\mathbb{O}$-Hopf fibration (cf. \cite[\S 3.2.3]{SS25-Orient}).
\end{enumerate}
Moreover we need:
\begin{itemize}
\item
The suspension of $T^2_{\plus}$ is homotopy equivalent to (\cite[Prop. 4I.1]{Hatcher2002}) a wedge sum of spheres:
\begin{equation}
  \label{StableSplittingOfSuspendedTorus}
  \begin{tikzcd}[
    row sep=0pt,
    column sep=40pt
  ]
  &
  S^2_a
  \\[+2pt]
  \Sigma\big(T^2_{\plus}\big)
  \underset{
  \mathrm{hmtp}
  }{\simeq}
  S^1 \vee S^2_a \vee S^2_b \vee S^3
  \ar[
    ur,
    "{ 
      \Sigma \mathrm{pr}_a 
    }"{description}
  ]
  \ar[
    r, 
    "{ 
      \Sigma \mathrm{pr}_c 
    }"{description}
  ]
  \ar[
    dr,
    shorten <=-18pt,
    "{ 
      \Sigma \mathrm{pr}_b 
    }"{
      description,
      pos=.3
    }
  ]
  &
  S^3
  \\[-2pt]
  & S^2_b
  \mathrlap{\,,}
  \end{tikzcd}
\end{equation}
via the suspensions of the maps $\mathrm{pr}_{a/b}$ onto the 1-cells and $\mathrm{pr}_c$ onto the 2-cell of the standard CW-complex structure of $T^2$.
\end{itemize}

With this, we have:
\begin{proof}[Proof of \cref{FundamentalGroupOfMapsFromTorusToSphere}]
  The underlying set \cref{UnderlyingSetOfIntegerHeisenbergGrp} follows immediately from \cref{StableSplittingOfSuspendedTorus}; the generators \cref{GeneratorsOfHeisenbergGrp} being represented by the projectors \cref{StableSplittingOfSuspendedTorus}, respectively. 
  Then, by \cref{SamelsonProductOnLoopSpace}, the group commutator $[W_a, W_b]$ is represented by the outer composite map of the following diagram:
  \begin{equation}
    \label{FirstDiagramForDerivingHeisenberg}
    \begin{tikzcd}[
      column sep=40pt
    ]
      T^2_{\plus}
      \ar[
        d,
        "{ \mathrm{diag} }"{swap}
      ]
      \ar[
        rr,
        "{
          [W_a, W_b]
        }"
      ]
      \ar[
        dr,
        "{
          \mathrm{pr}_c
        }"{description}
      ]
      &&
      \Omega S^2
      \\
      T^2_{\plus}
      \wedge
      T^2_{\plus}
      \ar[
        r,
        "{
          \mathrm{pr}_a
          \wedge 
          \mathrm{pr}_b
        }"{swap}
      ]
      &
      S^1 \wedge S^1
      \ar[
        r,
        "{
          \widetilde{\iota}_2
          \,\wedge\,
          \widetilde{\iota}_2
        }"{swap}
      ]
      \ar[
        ur,
        "{
          2 \, \widetilde{\eta}
        }"{description}
      ]
      &
      (\Omega S^2)
      \wedge
      (\Omega S^2)
      \ar[
        u,
        "{
          [-,-]_{\mathrm{Sam}}
        }"{swap}
      ]
    \end{tikzcd}
  \end{equation}
  and we identify the left diagonal map by inspection via \cref{StableSplittingOfSuspendedTorus} and the right diagonal map by \cref{WhiteheadProductAsSamelsonProduct,WhiteheadSquareRootsOfHopfFibrations}. This shows that $[W_a, W_b] = \zeta^2$.

  Similarly, $[W_{a/b}, \zeta]$ is represented by the outer composite map of the following diagram
  \begin{equation}
    \label{SecondDiagramForDerivingHeisenberg}
    \begin{tikzcd}[
      column sep=40pt
    ]
      T^2_{\plus}
      \ar[
        d,
        "{ \mathrm{diag} }"{swap}
      ]
      \ar[
        dr,
        "{ 0 }"{description}
      ]
      \ar[
        rr,
        "{
          [W_{a/b}, \zeta]
        }"
      ]
      &&
      \Omega S^2
      \\
      T^2_{plus}
      \wedge
      T^2_{\plus}
      \ar[
        r,
        "{
          \mathrm{pr}_{a/b}
          \wedge 
          \mathrm{pr}_c
        }"{swap}
      ]
      &
      S^1 \wedge S^2
      \ar[
        r,
        "{
          \widetilde{\iota_2}
          \,\wedge\,
          \widetilde{\eta}
        }"{swap}
      ]
      &
      (\Omega S^2)
      \wedge
      (\Omega S^2)
      \mathrlap{\,,}
      \ar[
        u,
        "{
          [-,-]_{\mathrm{Sam}}
        }"{swap}
      ]
    \end{tikzcd}
  \end{equation}
  but the diagonal map is null(homotopic), as indicated, already by degree reasons. This shows that $[W_{a/b}, \zeta] = \mathrm{e}$ and hence completes the proof.
\end{proof}

But this statement of \cref{FundamentalGroupOfMapsFromTorusToSphere} is remarkable, because the group algebra of the integer Heisenberg group at level 2, which we have thereby found to be the algebra of topological quantum observables \cref{OrdinaryTopObsAsGroupAlgebras} on 2-cohomotopical flux over the torus,
\begin{equation}
  \label{GroupAlgebraOfIntegerHeisenberg}
  \mathbb{C}\big[\,
    \widehat{\mathbb{Z}^2}
  \, \big]
  \underset{\scalebox{.7}{
\cref{IdentifyingFundamentalGroupOfMapsFromTorusToSphere}}}{\simeq}
  \mathbb{C}\big[
    \pi_1
    \mathrm{Map}(T^2, S^2)
  \big]
  \underset{\scalebox{.7}{\cref{OrdinaryTopObsAsGroupAlgebras}}}{\defneq}
  \mathrm{TopObs}_0
  \mathrlap{\,,}
\end{equation}
is exactly the algebra of observables of, equivalently (\parencites[(4.9)]{WenNiu1990}[(4.14)]{IengoLechner1992}[(4.21)]{Fradkin2013}[(5.28)]{Tong2016}):
\begin{enumerate}
  \item 
  abelian Chern-Simons Wilson loops,

  \item
  abelian FQH anyons
\end{enumerate}
on the torus. This shows that:
\begin{standout}
  Changing the flux quantization law of magnetic flux from ordinary cohomology to 2-Cohomotopy makes magnetic flux quanta on the torus become anyonic in just the way expected of surplus flux quanta (attached to quasi-holes) in fractional quantum Hall materials.
\end{standout}

This indicates that 2-Cohomotopy serves as the \emph{effective} non-abelian flux quantization law for surplus flux (above the exact rational filling fraction) in FQH systems. This turns out to be verified from a variety of angles:

\subsubsection{Wilson loop observables}
\label{OnWilsonLoopObservables}
\footnote{
  For background on framed knots and links cf. \cite[pp. 15]{Ohtsuki2001}. For the \emph{writhe} of framed links (the signed total number of crossings in their diagrams) cf.  \parencites[pp 152, 185]{Adams1994}[p. 523]{Ohtsuki2001}.
  
  The traditional definition of renormalized Wilson loops in abelian Chern-Simons goes back to \cite{Polyakov1988}, and its now widely accepted ``renormalization'' via framing was proposed in \cite[\S 2.1]{Witten1989}, reviewed in \cite[\S 2]{SS25-WilsonLoops}. 
  \nopagebreak
  Our systematic (re-)derivation via 2-cohomotopical flux quantization is due to \cite{SS25-AbelianAnyons,SS25-WilsonLoops} with further exposition in \parencites[\S 3]{SS25-Srni}{SS26-Rickles}.
  The proof of \cref{DerivingCSWilsonLoopObs} given there makes crucial use of the classical result \cite[Thm. 1]{Segal1973} and its more recent elaboration provided by \cite{Okuyama2005}.
}
\newline \nopagebreak 
To see more transparently how these observables \cref{GroupAlgebraOfIntegerHeisenberg} reflect anyons, consider the same situation but just on the plane $\mathbb{R}^2$ instead of the torus $T^2$, so that we may disentangle the pure anyon braiding phases from their \emph{Zak phases} associated with the nontrivial topology of the torus.

First of all, we find that the quantum algebra of ordinary topological observables \cref{OrdinaryTopObsAsGroupAlgebras} on solitonic 2-cohomotopical flux through the  plane is simply the group algebra of the integers:
\begin{equation}
  \label{2CohomotopicalObservablesOnPlane}
  \begin{aligned}
  \mathrm{TopObs}_0
  & \defneq
  \mathbb{C}\big[
  \pi_1 
  \,
  \mathrm{Map}^\ast\big(
    \mathbb{R}^2_{\cpt}
    ,\,
    S^2
  \big)
  \big]
  \simeq
  \mathbb{C}\big[
  \mathrm{Map}^\ast\big(
    \mathbb{R}^3_{\cpt}
    ,\,
    S^2
  \big)
  \big]
  \\
& \simeq
  \mathbb{C}\big[
    \pi_3(S^2)
  \big]
 \simeq
  \mathbb{C}\big[\mathbb{Z}\big]
  \mathrlap{\,.}
  \end{aligned}
\end{equation}

\begin{SCfigure}[1.2][htb]
\caption{
  \label{ThePontrjaginConstruction}
  The unstable \emph{Pontrjagin theorem} (cf. \parencites[\S II.16]{Bredon1993}[\S IX]{Kosinski1993}) physically identifies (cf. \parencites[\S 2.1]{SS20-Tad}[\S\S 2.2, 3.2]{SS23-Mf}) 2-cohomotopical charge with topological classes of trajectories of flux quanta whose density is the pullback of a \emph{Thom form} $\mathrm{th}_2$ supported on a neighborhood of $0 \in S^2$.
}
\adjustbox{
  rndfbox=4pt
}{
\hspace{-.2cm}
\begin{tikzpicture}
\begin{scope}[
  scale=.8,
  shift={(.7,-4.9)}
]

\draw[
  line width=.8,
  ->,
  darkgreen
]
  (1.5,1) 
  .. controls (2,1.6) and (3,2.6) .. 
  (6,1);

\node[
  scale=.7,
  rotate=-18
] at (4.7,1.8) {
  \color{darkblue}
  \bf
  classifying map
};

\node[
  scale=.7,
  rotate=-18
] at (4.5,1.4) {
  \color{darkblue}
  \bf
  $n$
};

  \shade[
    right color=gray, left color=lightgray,
    fill opacity=.9
  ]
    (3,-3)
      --
    (-1,-1)
      --
        (-1.21,1)
      --
    (2.3,3);

  \draw[dashed]
    (3,-3)
      --
    (-1,-1)
      --
    (-1.21,1)
      --
    (2.3,3)
      --
    (3,-3);

  \node[
    scale=1
  ] at (3.2,-2.1)
  {$\infty$};

  \begin{scope}[rotate=(+8)]
  \shadedraw[
    dashed,
    inner color=olive,
    outer color=lightolive,
  ]
    (1.5,-1)
    ellipse
    (.2 and .37);
  \draw
   (1.5,-1)
   to 
    node[above, yshift=-1pt]{
     \;\;\;\;\;\;\;\;\;\;\;
     \rotatebox[origin=c]{7}{
     \scalebox{.7}{
     \color{darkorange}
     \bf
       anyon
     }
     }
   }
    node[below, yshift=+6.3pt]{
     \;\;\;\;\;\;\;\;\;\;\;\;
     \rotatebox[origin=c]{7}{
     \scalebox{.7}{
     \color{darkorange}
     \bf
       worldline
     }
     }
   }
   (-2.2,-1);
  \draw
   (1.5+1.2,-1)
   to
   (4,-1);
  \end{scope}

  \begin{scope}[shift={(-.2,1.4)}, scale=(.96)]
  \begin{scope}[rotate=(+8)]
  \shadedraw[
    dashed,
    inner color=olive,
    outer color=lightolive,
  ]
    (1.5,-1)
    ellipse
    (.2 and .37);
  \draw
   (1.5,-1)
   to
   (-2.3,-1);
  \draw
   (1.5+1.35,-1)
   to
   (4.1,-1);
  \end{scope}
  \end{scope}
  \begin{scope}[shift={(-1,.5)}, scale=(.7)]
  \begin{scope}[rotate=(+8)]
  \shadedraw[
    dashed,
    inner color=olive,
    outer color=lightolive,
  ]
    (1.5,-1)
    ellipse
    (.2 and .32);
  \draw
   (1.5,-1)
   to
   (-1.8,-1);
  \end{scope}
  \end{scope}
  
\end{scope}

\node[
  scale=.73,
  rotate=-27
] at (2.21,-5.21) {
  \color{darkblue}
  \bf
  \def\arraystretch{.9}
  \begin{tabular}{l}
    flux
    \\
    $B_2 =$ 
    \\
    $\;\;n^\ast(\mathrm{th}_2)$
  \end{tabular}
};

\node[
  scale=.73,
] at (2.05,-5) {
  \color{darkblue}
  \bf
};

\node[
  rotate=-140
] at (6,-4) {
  \includegraphics[width=2cm]{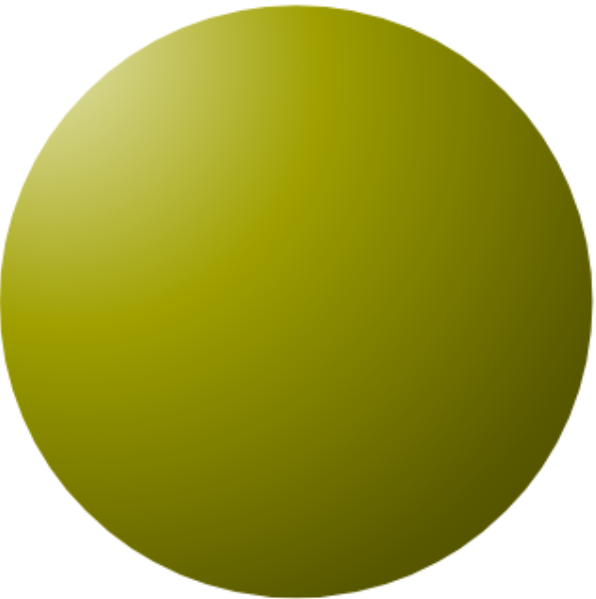}
};

\node[
  scale=.7
] at (6,-5.2) {
  \color{darkblue}
  \bf
  2-sphere $S^2$
};

\node[
  scale=.8,
  rotate=-22
] at (.05,-4.55)
{$\Sigma^2$};

\end{tikzpicture}
\hspace{-.7cm}
}
\end{SCfigure}

By the \emph{Pontrjagin theorem} (cf. \cref{ThePontrjaginConstruction}), these observables $\mathcal{O}_L$ are spanned by characteristic functions on topological equivalence classes $[L]$ of \emph{vacuum diagrams of flux quanta}, given by \emph{cobordism classes of framed links} $L$!

By \cref{2CohomotopicalObservablesOnPlane}, these topological classes of flux vacuum diagram are labeled by the integers, and hence the remaining question is which integer this is:
\begin{proposition}
  \label[proposition]{DerivingCSWilsonLoopObs}
  The integer invariant of framed link diagrams detected by the topological observable $\mathcal{O}_L$ \cref{2CohomotopicalObservablesOnPlane} is the \emph{total crossing number} or \emph{writhe} $\#L$ of the link, which is exactly the \emph{renormalized Wilson loop observables} of \emph{abelian spin Chern-Simons theory}.
\end{proposition}
Namely, evaluated in a pure topological quantum state $\vert K \rangle$ we have
\begin{equation}
  \label{KnotTheoreticWilsonLoopObservables}
  \begin{aligned}
  \langle K \vert
  \mathcal{O}_L
  \vert K \rangle
  & =
  \exp\big(
    \tfrac{
      \pi \mathrm{i}
    }{
      K
    }
    \,
    \# L
  \big)
  \\
  & =
  \exp\Bigg(
    \tfrac{
      \pi \mathrm{i}
    }{
      K
    }
    \bigg(\;
      \displaystyle{
        \sum_{i \neq j \in \pi_0(L)}
      }
      \underset{
        \mathclap{
          \adjustbox{
            raise=-1pt,
            scale=.7
          }{
            \color{gray}
            \bf
            \def\arraystretch{.85}
            \begin{tabular}{c}
              linking
              \\
              numbers
            \end{tabular}
          }
        }
      }{
        \mathrm{lnk}(L_i, L_j)
      }
      \,+\,
      \displaystyle{
        \sum_{i \in \pi_0(L)}
      }
      \underset{
        \mathclap{
          \adjustbox{
            raise=-1pt,
            scale=.7
          }{
            \color{gray}
            \bf
            \def\arraystretch{.85}
            \begin{tabular}{c}
              framing
              \\
              numbers
            \end{tabular}
          }
        }
      }{
        \mathrm{frm}(L_i)
      }
  \;  \bigg)
\!  \Bigg)
  \mathrlap{\,,}
  \end{aligned}
\end{equation}
which is the traditional formula for renormalized Wilson loop observables in abelian Chern-Simons theory. \footnote{
  The traditional renormalization process of these observables is somewhat \emph{ad hoc}, and the end result \cref{KnotTheoreticWilsonLoopObservables} is traditionally more of a definition than a derivation. 
} 
For instance, for the vacuum flux diagram which is the trefoil knot, we obtain:
\begin{equation}
\Bigg\langle K\, \Bigg\vert
\adjustbox{
  raise=-2.6cm,
  scale=.4
}{
\begin{tikzpicture}
\foreach \n in {0,1,2} {
\begin{scope}[
  rotate=\n*120-4
]
\draw[
  line width=2,
  -Latex
]
 (0,-1)
   .. controls
   (-1,.2) and (-2,2) ..
 (0,2)
   .. controls
   (1,2) and (1,1) ..
  (.9,.7);
\end{scope}

\node[darkgreen]
  at (\n*120+31:.7) {
    \scalebox{1}{$+$}
  };
};
\end{tikzpicture}
}
\hspace{-7pt}
\Bigg\vert
\, K
\Bigg\rangle
\;\;
=
\;\;
\Bigg\langle K\, \Bigg\vert
\adjustbox{
  raise=-2.6cm,
  scale=.4
}{
\begin{tikzpicture}
\foreach \n in {0,1,2} {
\begin{scope}[
  rotate=\n*120-4
]
\draw[
  line width=2,
  -Latex
]
 (0,-1)
   .. controls
   (-1,.2) and (-2,2) ..
 (0,2)
   .. controls
   (1,2) and (1,1) ..
  (.9,.7);
\end{scope}
};

\draw[white,fill=white]
  (-.9,-.7)
  rectangle
  (.9,.2);

\draw[line width=2]
  (-.79,.21) 
    .. controls
    (-.7,-.25) and (+.7,-.25) ..
  (.79,.21);

\draw[line width=2]
  (-.27,-.7) 
  .. controls
    (-.4,-.2) and 
    (+.4,-.2) ..
  (+.27,-.7);

\draw[white,fill=white]
 (-.5,.8) 
 rectangle (+.5,-.2);

\draw[line width=2]
  (-.5,.58)
  .. controls
    (-.0,.7) and (-.0,-.26)
    ..
  (-.5,-.06);

\draw[line width=2]
  (+.5,.58)
  .. controls
    (+.0,.7) and (+.0,-.26)
    ..
  (+.5,-.06);

\end{tikzpicture}
}
\hspace{-7pt}
\Bigg\vert
\, K
\Bigg\rangle
\;\;
=
\;\;
\exp\big(
  \pi \mathrm{i}
  \,
  \tfrac{3}{K}
\big)
\mathrlap{\,.}
\end{equation}
But this means equivalently that a \emph{braiding phase}
\begin{equation}
  \zeta 
    := 
  e^{ 
    \tfrac{\pi \mathrm{i} }{K} 
  }
\end{equation}
(or its inverse, depending on orientation) is picked up for each braiding of worldlines of flux quanta --- exactly as seen in FQH materials (cf. \cref{FQHAnyons}).

\subsubsection{Superconducting Island Anyons}
\label{OnSuperconductingIslandAnyons}
\footnote{
  Discussions of theoretical and experimental aspects of doping FQH electron liquids by superconducting islands are listed in \cite{SS26-Islands,nLab:SupercondunctingIslandsInFQH}.

  The prediction, from 2-cohomotopical flux quantization, that superconducting islands in FQH systems support parastatistical anyons is due to \cite[Fig. D \& \S3.8]{SS25-FQH}, based on Thm. 3.50 there, with exposition in \cite{SS26-Islands}.
}

Since 2-cohomotopical flux quantization thereby matches the expected and experimentally observed nature of solitonic FQH anyons on the torus (\cref{AnyonsOnTheTorus}) and on the plane (\cref{OnWilsonLoopObservables}), we gain confidence that this is the correct globally completed model for FQH topological order. Any further consequences of 2-cohomotopical flux quantization therefore become \emph{predictions} about novel anyonic phenomena possible in FQG materials. 

Consider therefore now as a domain a \emph{punctured} surface 
\begin{equation}
  \Sigma^2_{g,b,n>0}
  :=
  \Sigma^2_{g,b} \setminus \mathbf{n}
  \,,
\end{equation}
punctured at a finite set $\mathbf{n} \subset \Sigma^2_{g,b}$ of points.
Under one-point compactification, each puncture is identified with the point-at-infinity (\cref{SolitonicFields}), so that the mathematical situation of puncturing models a physical situation where flux solitons are prevented from approaching any of the $\mathbf{n}$ punctures, cf. \cref{SuperconductingIslandsIllustrated}. 

\begin{figure}[htb]
\caption{
  \label{SuperconductingIslandsIllustrated}
  The understanding of FQH anyons as surplus flux effectively quantized in 2-Cohomotopy first retrodicts (\cref{AnyonsOnTheTorus,OnWilsonLoopObservables}) the expected behaviour of solitonic anyons (quasi-hole vortices) but then furthermore predicts (\cref{OnSuperconductingIslandAnyons}) non-abelian defect anyons associated with points where magnetic flux is expelled (such as superconducting islands in the FQH liquid).  
}
\adjustbox{
  scale=.9,
  rndfbox=4pt
}{
\begin{tikzpicture}[scale=0.8]

\draw[
  gray!30,
  fill=gray!15
]
  (0-.3,0) --
  (-.3+3,1) --
  (9.7+3, 1) --
  (9.7, 0) -- cycle;

\foreach \n in {1,...,13} {
  \draw[
    line width=2pt
  ]
    (\n*.7, -2) --
    (\n*.7, -1);
  \draw[
    line width=2pt
  ]
    (\n*.7, +1) --
    (\n*.7, +2);
}

\draw[
  line width=2pt
]
  (4*.7, -1) .. controls
  (4*.7, -.3) and
  (4*.7+.3,-.6) ..
  (4*.7+.3, 0);
\draw[
  line width=2pt
]
  (4*.7, +1) .. controls
  (4*.7, +.3) and
  (4*.7+.3,+.6) ..
  (4*.7+.3, 0);

\draw[
  line width=2pt
]
  (3*.7, -1) .. controls
  (3*.7, -.3) and
  (3*.7-.3,-.6) ..
  (3*.7-.3, 0);
\draw[
  line width=2pt
]
  (3*.7, +1) .. controls
  (3*.7, +.3) and
  (3*.7-.3,+.6) ..
  (3*.7-.3, 0);

\draw[
  line width=2pt
]
  (11*.7, -1) .. controls
  (11*.7, -.3) and
  (11*.7+.2,-.6) ..
  (11*.7+.2, 0);
\draw[
  line width=2pt
]
  (11*.7, +1) .. controls
  (11*.7, +.3) and
  (11*.7+.2,+.6) ..
  (11*.7+.2, 0);

\draw[
  line width=2pt
]
  (10*.7, -1) .. controls
  (10*.7, -.3) and
  (10*.7-.2,-.6) ..
  (10*.7-.2, 0);
\draw[
  line width=2pt
]
  (10*.7, +1) .. controls
  (10*.7, +.3) and
  (10*.7-.2,+.6) ..
  (10*.7-.2, 0);

\draw[
  line width=2pt
]
  (6*.7, -1) .. controls
  (6*.7, -.3) and
  (6*.7+.75,-.6) ..
  (6*.7+.75, 0);

\draw[
  line width=2pt
]
  (9*.7, -1) .. controls
  (9*.7, -.3) and
  (9*.7-.75,-.6) ..
  (9*.7-.75, 0);

\draw[
  line width=2pt
]
  (7*.7, -1) .. controls
  (7*.7, -.3) and
  (7*.7+.25,-.6) ..
  (7*.7+.25, 0);
\draw[
  line width=2pt
]
  (8*.7, -1) .. controls
  (8*.7, -.3) and
  (8*.7-.25,-.6) ..
  (8*.7-.25, 0);

\draw[
  gray,
  line width=2
]
  (-.3,0) --
  (2.05,0);
\draw[
  gray,
  line width=2
]
  (2.85,0) --
  (6.95,0);
\draw[
  gray,
  line width=2
]
  (7.75,0) --
  (9.7,0);
  
\shadedraw[
  draw opacity=0,
  inner color=olive,
  outer color=lightolive
]
  (7.5*.7,0) ellipse 
  (.58 and .12);

\draw[
  line width=2pt
]
  (7*.7, +1) .. controls
  (7*.7, +.3) and
  (7*.7+.25,+.6) ..
  (7*.7+.25, 0);

\draw[
  line width=2pt
]
  (8*.7, +1) .. controls
  (8*.7, +.3) and
  (8*.7-.25,+.6) ..
  (8*.7-.25, 0);

\draw[
  line width=2pt
]
  (6*.7, +1) .. controls
  (6*.7, +.3) and
  (6*.7+.75,+.6) ..
  (6*.7+.75, 0);

\draw[
  line width=2pt
]
  (9*.7, +1) .. controls
  (9*.7, +.3) and
  (9*.7-.75,+.6) ..
  (9*.7-.75, 0);

\draw[
  line width=2pt
]
  (1*.7, -1) --
  (1*.7, +1); 

\draw[
  line width=2pt
]
  (2*.7, -1) .. controls
  (2*.7, -.3) and
  (2*.7-.1,-.6) ..
  (2*.7-.1, 0);
\draw[
  line width=2pt
]
  (2*.7, +1) .. controls
  (2*.7, +.3) and
  (2*.7-.1,+.6) ..
  (2*.7-.1, 0);

\draw[
  line width=2pt
]
  (5*.7, -1) .. controls
  (5*.7, -.3) and
  (5*.7+.1,-.6) ..
  (5*.7+.1, 0);
\draw[
  line width=2pt
]
  (5*.7, +1) .. controls
  (5*.7, +.3) and
  (5*.7+.1,+.6) ..
  (5*.7+.1, 0);

\draw[
  line width=2pt
]
  (12*.7, -1) --
  (12*.7, +1); 

\draw[
  line width=2pt
]
  (13*.7, -1) --
  (13*.7, +1); 

\draw[
  gray!30,
  draw opacity=.5,
  fill=gray!25,
  fill opacity=.5
]
  (0-.3,0) --
  (-.3-3,-1) --
  (9.7-3, -1) --
  (9.7,0) -- cycle;

\draw[fill=white] 
  (3.5*.7,0) ellipse 
  (.57 and .07);
\draw[fill=white] 
  (10.5*.7,0) ellipse 
  (.4 and .07);

\begin{scope}

\clip
  (4,0) rectangle
  (+6,-1);

\shadedraw[
  draw opacity=0,
  inner color=olive,
  outer color=lightolive
]
  (7.5*.7,0) ellipse 
  (.58 and .12);

\end{scope}

\node at (2.4,-2.5)
 {
   \scalebox{.65}{
     \color{darkblue}
     \bf
     \def\arraystretch{.9}
     \begin{tabular}{c}
       flux-expelling
       \\
       defect (puncture):
       \\
       {\bf non-abelian} anyon
     \end{tabular}
   }
 };

\node at (5.3,-2.5)
 {
   \scalebox{.65}{
     \color{darkblue}
     \bf
     \def\arraystretch{.9}
     \begin{tabular}{c}
       flux quantum
       \\
       soliton (vortex):
       \\
     abelian anyon
     \end{tabular}
     }
 };

\end{tikzpicture}
}
\end{figure}

This may be realized in experiment (cf. \cite{Gul:2022and}) by doping the semiconducting host material of the FQH systems by \emph{superconducting islands} at the position of the punctures.
\footnote{This is because a superconductor tends to expel magnetic flux, by the \emph{Meissner effect}. The engineering challenge is to have the superconducting phase survive the high magnetic field strength that is needed for the ambient FQH phase.}

In the solid state physics literature it is argued in somewhat roundabout ways that such superconducting islands in FQH materials are associated with non-abelian \emph{parafermionic} anyons. From our formalism we obtain an unambiguous prediction:

\begin{proposition}[{\cite[Ex. 3.59]{SS25-FQH}}]
  The covariantized \textup{(\cref{TheModularFunctor})} 2-cohomotopical quantum observables on the 3-punctured sphere form the group algebra of the \emph{framed symmetric group} on 3 strands with total framing divisible by 3:
  \begin{equation}
    \pi_1
    \Big(
    \mathrm{Map}^\ast\big(
      \Sigma^2_{0,0,3}
      ,\,
      S^2
    \big)
    \sslash
    \mathrm{Diff}\big(
      \Sigma^2_{0,0,3}
    \big)
    \Big)
    \subset
    \mathbb{Z}^3 
      \rtimes 
    \mathrm{Sym}_3
    \mathrlap{\,.}
  \end{equation}
\end{proposition}

Accordingly, the corresponding topological quantum states are $\mathrm{Sym}_3$-representations, exhibiting the superconducting islands as parastatistical defect anyons (cf. \parencites[Rem. 3.61]{SS25-FQH}{SS26-Islands}).

\subsection{Anyonic C-Field Flux Quanta}
\label{AnyonicCFieldFluxQuanta}
\footnote{
  This section \cref{AnyonicCFieldFluxQuanta} presents a new result, not published before. The end result has some structural similarity with the proposal in \cite[\S 2]{Hartnoll2006}, but the details are nominally different.

  Further different but related realizations of anyons, as higher flux quanta on M5-brane probes in globally completed 11D SuGra, have been found in \cite{SS23-Defect} (realizing more enticing non-abelian $\mathfrak{su}(2)$-anyons, but at the cost of a more complex notion of \emph{intersecting solitons} due to \cite[\S 2.4-5]{SS22-Conf}), and in \parencites{SS25-Seifert,SS25-Srni,BaSS26-MString} (via a variant flux quantization in \emph{equivariant twistorial cohomotopy} introduced in \parencites{FSS22-Twistorial}{SS25-TEC}, cf. \cite[\S 12]{FSS23-Char}).
  The latter construction is not unlike what is suggested by informal physics arguments in \cite{ChoGangKim2020}, followed up on in \parencites{CuiQiuWang2023}{BonettiSchaeferNamekiWu2024}.

  All of these are examples of what is called \emph{geometric engineering} (cf. \cite{Duplij2017,nLab:GeometricEngineering}) of (FQH) anyons (on M5-branes) in 11D SuGra/M-theory, whereby a model for a strongly coupled quantum system is identified in a subsector of M-brane theory and thereby made amenable to tools of analysis that would not otherwise be evident.
  
}

With the above effective description of FQH anyons (in \cref{OnFQAHMaterials}) naturally embedded (by \cite{SS25-WilsonLoops}) into the gauge sector of 5D supergravity globally completed according to \emph{Hypothesis h} (\cref{Hypothesish}), we are led to ask whether analogous phenomena appear in the gauge sector of 11D SuGra when globally completed according to \emph{Hypothesis H}. Indeed this is the case, essentially by the same arguments shifted in degrees.

\subsubsection{Solitonic 5-branes}
\label{OnSolitonic5Branes}
\footnote{
  The observation in this \cref{OnSolitonic5Branes} is 
  due to \cite{KSS26-HigherAnyons}.
}

According to the general formula for solitonic $p$-brane domains \cref{OnSolitonicBranes} and in higher-dimensional analogy to \cref{SpatialDomainForSolitonicStrings}, the spatial domain on which to measure 3-soliton charge in 11d SuGra is of the form
\begin{equation}
  X^{10}_{\mathrm{dom}}
  :=
  \mathbb{R}^{3}_{\plus}
  \wedge
  \big(
    \mathbb{R}^1
    \times
    \Sigma^6
  \big)_{\cpt}\,.
\end{equation}
However, the Pontrjagin theorem (as in \cref{ThePontrjaginConstruction}) makes the 4-Cohomotopy charge on such a 
domain correspond to $6-4 = 2$-manifolds inside $\Sigma^6$, whose further extension along the $\mathbb{R}^3$-factor
therefore makes up the worldvolume of M5-branes.

In any case, the 3-solitonic ordinary topological observables \cref{OrdinaryTopObsAsGroupAlgebras} are
\begin{equation}
  \begin{aligned}
  \mathrm{TopObs}_0
  & \simeq
  \mathbb{C}\bigg[
  \pi_0\,
  \mathrm{Map}^\ast\Big(
    \mathbb{R}^{3}_{\plus}
    \wedge
    \big(
      \mathbb{R}^1
      \times
      \Sigma^6
    \big)_{\cpt}
    ,\,
    S^4
  \Big)
  \bigg]
  \\
  & \simeq
  \mathbb{C}\bigg[
  \pi_0\,
  \mathrm{Map}^\ast\Big(
    \mathbb{R}^1_{\cpt}
    \times
    \Sigma^6_{\cpt}
    ,\,
    S^4
  \Big)
  \bigg]
  \\
  & \simeq
  \mathbb{C}\Big[
  \pi_1\,
  \mathrm{Map}^\ast\Big(
    \Sigma^6_{\cpt}
    ,\,
    S^4
  \Big)
  \Big]
  \mathrlap{\,,}
  \end{aligned}
\end{equation}
and if $\Sigma^6$ is already compact then this reduces to:
\begin{equation}
  \mbox{
    $\Sigma^6$ compact
  }
  \;\;\;
  \Rightarrow
  \;\;\;
  \mathrm{TopObs}_0
  \simeq
  \mathbb{C}\Big[
    \pi_1
    \,
    \mathrm{Map}\big(
      \Sigma^6
      ,\,
      S^4
    \big)
  \Big]
  \mathrlap{\,.}
\end{equation}

\subsubsection{The higher analog of anyons on the torus}

Consider the 11D background spacetime of the form:
\begin{equation}
  \label{TheS3S3Spacetime}
  X^{1,10}
    \defneq
  \mathrm{AdS}_3 
    \times
    \mathbb{R}^2
    \times 
  S^3 \times S^3 
  \,,
\end{equation}
which is the near-horizon geometry of a certain intersection of black M2/M5-branes (cf. \cite[\S 2.2, p. 9,16]{BoonstraPeetersSkendris1998}).
Taking here physical spacetime to occupy an \emph{entanglement wedge} subregion of $\mathrm{AdS}_3$ (which is globally hyperbolic, in contrast to the full $\mathrm{AdS}$ spacetime, cf. \cite{HeadrickEtAl2014}), the actual Cauchy surface is diffeomorphic to
\begin{equation}
  \label{TheS3S3CauchySurface}
  X^{10}
  \simeq
  \mathbb{R}^2
    \times
  \mathbb{R}^2
    \times 
  S^3 \times S^3 
  \,,
\end{equation}
whence the spatial domain for 3-soliton charge topological observables
\begin{equation}
  \label{TheS3S3SpatialDomain}
  X^{10}_{\mathrm{dom}}
  \simeq
  \mathbb{R}^3_{\plus}
  \wedge
  \big(
    \mathbb{R}^1
      \times 
    S^3 \times S^3 
  \big)_{\cpt}
  \,,
\end{equation}

The charge as fixed light-front time is classified by the 4-Cohomotopy of the remaining sphere factors:
\begin{lemma}
  On the above spacetime \cref{TheS3S3Spacetime} there are two 4-cohomotopical charge sectors:
  \begin{equation}
    \label{4CohomotopyOfS3TimesS3}
    \pi^4\big(S^3 \times S^3\big)
    \defneq
    \pi_0\, 
    \mathrm{Map}\big(
      S^3 \times S^3
      ,\,
      S^4
    \big)
    \simeq
    \pi_6(S^4)
    \simeq
    \mathbb{Z}_{/2}
    \mathrlap{\,.}
  \end{equation}
\end{lemma}
\begin{proof}
  The standard CW-complex structure of $S^3 \times S^3$ has two 3-cells $S^3 \vee S^3$ and a single 6-cell attached via the Whitehead product of their identity maps (this by the very definition of the Whitehead product):
  \begin{equation}
    \begin{tikzcd}
      S^5 
      \ar[
        rr,
        "{ [\iota_3, \iota_3]_{\mathrm{Wh}} }"
      ]
      \ar[
        d,
        hook
      ]
      &{}
      \ar[
        dr,
        phantom,
        "{ \ulcorner }"{pos=.9}
      ]
      &
      S^3 \vee S^3
      \ar[
        d,
        hook
      ]
      \\
      D^6
      \ar[
        rr,
        "{
          i_6
        }"
      ]
      &&
      S^3 \times S^3
      \mathrlap{\,.}
    \end{tikzcd}
  \end{equation}
  But since $\pi^4(S^3) \simeq \pi_3(S^4) = 0$, it follows that all maps $\begin{tikzcd}[sep=small]S^3 \times S^3 \ar[r] & S^4\end{tikzcd}$ factor through the cell attachment along $\begin{tikzcd}[sep=small]S^5 \ar[r] & \ast\end{tikzcd}$, which is $S^6$.
  This shows the first claimed isomorphism. The second is by the standard tables of low homotopy groups of spheres.
\end{proof}

\begin{proposition}
\label[proposition]{FundamentalGroupOfChargeOnS3TimesS3}
The fundamental group of both of these charge sectors, $[k] \in \mathbb{Z}_{/2}$ is:
\begin{equation}
  \label{PiOneOfMapsFromS3TimesS3ToS4}
  \pi_1\,
  \mathrm{Map}_{[k]}\big(
    S^3 \times S^3
    ,\,
    S^4
  \big)
  \simeq
  \widehat{\mathbb{Z}^2}
  \times 
  \mathbb{Z}_{/12}
  \mathrlap{\,,}
\end{equation}
where $\widehat{\mathbb{Z}^2}$ is the integer Heisenberg group at level=2 \textup{(\cref{IntegerHeisenbergGroupAtLevel2})}.
\end{proposition}
\begin{proof}
First to see that the claim holds for the sector of charge $0 \,\mathrm{mod}\,2$:

The proof of \cref{FundamentalGroupOfMapsFromTorusToSphere} in \cref{AnyonsOnTheTorus} immediately generalizes from flux on the torus $T^2 = S^1 \times S^1$ with classifying space $S^2$ to higher flux on $S^{2k+1} \times S^{2k+1}$ with classifying space $S^{2k+2}$, for all $k \in \mathbb{N}$: The stable splitting in \cref{StableSplittingOfSuspendedTorus} generalizes to 
\begin{equation}
  \Sigma\Big(
    \big(S^{2k+1} \times S^{2k+1}\big)_{\cpt}
  \Big)
  \simeq
  S^1 
  \vee 
  S^{2k+2}_a
  \vee
  S^{2k+2}_b
  \vee
  S^{4k+3}
\end{equation}
and from there the diagrams \cref{FirstDiagramForDerivingHeisenberg,SecondDiagramForDerivingHeisenberg} apply verbatim up to the evident shifts in degrees. For $k = 1$ this gives the claim to be proven. (The torsion factor in \cref{PiOneOfMapsFromS3TimesS3ToS4} comes from $\pi_7(S^4) \simeq \mathbb{Z} \times \mathbb{Z}_{/12}$.)

Hence, to conclude, it is sufficient to see that the other charge sector has the fundamental group. This is the content of the following \cref{CoincidingFundamentalGroupsOfMapS3TimesS3ToS4}.
\end{proof}
\begin{lemma}
  \label[lemma]{CoincidingFundamentalGroupsOfMapS3TimesS3ToS4}
  The two connected components \cref{4CohomotopyOfS3TimesS3} of $\mathrm{Map}\big(S^3 \times S^3,\, S^4\big)$ have isomorphic fundamental groups.
\end{lemma}
\begin{proof}
  We adapt the proof idea of \cite[Prop. 1]{Hansen1974}.

  First observe that the \emph{pointed} mapping space has the same fundamental groups, as shown by the homotopy long exact sequence of the fibration of evaluation at a basepoint in the component $[k] \in \mathbb{Z}_{/2}$ and using the the 2-connectedness of the 4-sphere:
  \begin{equation}
    \begin{tikzcd}[
      row sep=0pt,
      column sep=16pt
    ]
      \mathrm{Map}_{[k]}^{\ast}\big(%
        S^3 \!\times\! S^3%
        ,\,%
        S^4
      \big)
      \ar[
        r,
        "{
          \mathrm{fib}_k
        }"
      ]
      &
      \mathrm{Map}_{[k]}\big(%
        S^3 \!\times\! S^3%
        ,\,%
        S^4
      \big)
      \ar[
        r,
        "{ \mathrm{ev}_{k} }"
      ]
      &
      S^4
      \\[-6pt]
      &&
      \grayoverbrace{
        \pi_2 S^4
      }{ 0 }
      \ar[
        dll,
        snake left
      ]
      \\[+6pt]
      \pi_1
      \mathrm{Map}_{[k]}^{\ast}\big(%
        S^3 \!\times\! S^3%
        ,\,%
        S^4
      \big)
      \ar[
        r,
        "{ \sim }"
      ]
      &
      \pi_1
      \mathrm{Map}_{[k]}\big(%
        S^3 \!\times\! S^3%
        ,\,%
        S^4
      \big)
      \ar[
        r,
        "{ \mathrm{ev}_{k} }"
      ]
      &
      \grayunderbrace{
      \pi_1
      S^4}{1}
      \mathrlap{\,.}
    \end{tikzcd}
  \end{equation}
  Therefore, it is now sufficient to show that the connected components of the pointed mapping space are homotopy equivalent: 
  \begin{equation}
    \label{HomotopyEquivalenceBetweenComponentsOfMapS3TimesS3ToS4}
    \mathrm{Map}_{[0]}^\ast\big(
      S^3 \times S^3,
      S^4
    \big)
    \underset{\mathrm{hmtp}}{\simeq}
    \mathrm{Map}_{[1]}^\ast\big(
      S^3 \times S^3,
      S^4
    \big)
    \mathrlap{\,.}
  \end{equation}

  To that end, consider the \emph{pinching map} $\phi$
  \begin{equation}
    \begin{tikzcd}[
      column sep=15pt
    ]
      S^{5}
      \ar[d]
      \ar[
        r,
        hook
      ]
      &
      \mathring{D}^6
      \ar[
        r,
        hook
      ]
      &
      D^6
      \ar[
        r,
        hook,
        "{ i_6 }"
      ]
      \ar[
        dr,
        phantom,
        "{ \ulcorner }"{pos=.9}
      ]
      &
      S^3 \times S^3
      \ar[
        d,
        "{ \phi }"
      ]
      \\
      \ast
      \ar[rrr]
      &&&
      (S^3 \times S^3) \vee S^6
    \end{tikzcd}
  \end{equation}
  and use it to construct an endo-map
  \begin{equation}
    \begin{tikzcd}
      \mathrm{Map}^\ast\big(
        S^3 \times S^3,
        S^4
      \big)
      \ar[
        r,
        "{ \Theta }"
      ]
      &
      \mathrm{Map}^\ast\big(
        S^3 \times S^3,
        S^4
      \big)
    \end{tikzcd}
  \end{equation}
  by setting
  \begin{equation}
    \Theta(f)
    :
    \begin{tikzcd}[sep=20pt]
      (S^3 \!\times\! S^3)
      \ar[
        r,
        "{ \phi }"
      ]
      &
      (S^3 \!\times\! S^3)
      \vee
      S^6
      \ar[
        rr,
        "{ 
          f 
          \,\vee\, 
          [1]
        }"
      ]
      &&
      S^4 \vee S^4
      \ar[
        r,
        "{ \nabla }"
      ]
      &
      S^4
      \mathrlap{.}
    \end{tikzcd}
  \end{equation}
  Now observe that $\Theta$ swaps connected components, and that $\Theta \circ \Theta$ is homotopic to the identity (because $\phi$ is homotopy co-associative, $\nabla$ is homotopy associative and $[2]$ is homotopic to $[0]$). Therefore, $\Theta$ restricts on the $[0]$-components to a homotopy equivalence as claimed in \cref{HomotopyEquivalenceBetweenComponentsOfMapS3TimesS3ToS4}.  
\end{proof}

But \cref{FundamentalGroupOfChargeOnS3TimesS3} has the remarkable consequence that, with 11D SuGra globally completed by flux quantization in 4-Cohomotopy:
\begin{standout}
  C-Field flux quanta on entanglement wedges of $\mathrm{AdS}_3 \times S^3 \times S^3 \times \mathbb{R}^2$ \cref{TheS3S3Spacetime} have the same anyonic quantum observables as surplus magnetic flux quanta, hence as quasi-hole anyons, of toroidal FQH materials. 
\end{standout}
This suggests further possibilities for modelling anyonic topological quantum materials (which remain ill-understood) by embedding their theory into global completions of higher supergravity theories. See also the outlook in \cref{TopoloticalQuantumGates}.

\section{Outlook: Beyond the Topological Sector}

We close with some comments on what to do next.

\subsection{Quantum 11D Supergravity}
\label{QuantumSupergravity}
\footnote{
 Basic review of 11D SuGra includes 
 \parencites[\S 1]{Duff1999World}{MiemiecSchnakenburg2006}[\S 10]{FreedmanVanProeyen2012}.}

We have focused in \cref{GlobalCompletionOfHigherGaugeFields} on the (higher) gauge field sectors of (higher-dimensional) supergravity theories, disregarding the \emph{backreaction} of these fields on a background field of (super-)gravity.
But for the special case of 11D supergravity (at least), we can actually characterize the topological quantum observables of the full 11D SuGra theory globally completed by flux quantization --- including the dynamical super-gravitational sector, keeping just the bare choice of underlying manifold as background:

\subsubsection{The phase space of completed 11D SuGra}
\label{PhaseSpaceOfCompleted11DSugra}
\footnote{
  Something close to  \cref{11DSugraViaSuperFluxBianchi} has been claimed in \parencites{CremmerFerrara1980}{BrinkHowe1980} with little indication of its proof, and then again in \cite[\S III.8.5]{CDF1991} with indication of a couple of the easy steps of the proof. The full proof is laid out in \cite[Thm. 3.1]{GSS24-SuGra}, using computer algebra to double-check the easier steps and to complete the computationally intensive steps (the verification takes over 11 minutes on a contemporary notebook computer). 
  The observation that this allows to construct the higher phase space of global completions of 11D SuGra by flux quantization follows \cite[Claim 1.1]{GSS24-SuGra}.

  The analogous statement for the dynamics of M5-brane probes in 11D SuGra originates in \parencites{HoweSezgin1997}[\S 5.2]{Sorokin2000}, with a streamlined rederivation in \cite{GSS25-M5} and the first-ever nontrivial example rigorously constructed in \cite{GSS25-Embedding}.  
}

This is due to the following `miracle':
\begin{proposition}[{\cite[Thm. 3.1]{GSS24-SuGra}}]
  \label[proposition]{11DSugraViaSuperFluxBianchi}
  An $(11\vert \mathbf{32})$-dimensional super-spacetime $X$ with super-frame field $(e,\psi)$ \textup{(the gravity \& gravitino field)} is a solution of 11D supergravity with C-field flux density $G_4$ iff the super-flux densities
  \begin{equation}
    \label{The11DSuperFluxDensities}
    \begin{aligned}
    G_4^s 
    &
    :=
    G_4  
      + 
    \tfrac{1}{2}\big(\,
      \overline{\psi}
      \Gamma_{a_1 a_2}
      \psi
    \big)
    e^{a_1}
    e^{a_2}
    \\
    G_4^s 
    &
    :=
    G_7
      + 
    \tfrac{1}{5!}\big(\,
      \overline{\psi}
      \Gamma_{a_1 \cdots a_5}
      \psi
    \big)
    e^{a_1}
    \cdots
    e^{a_5}
    \end{aligned}
  \end{equation}
  constitute a closed $\mathfrak{l}S^4$-valued differential super-form \cref{ClosedS4ValuedForms}:
  \footnote{
    This means that 
    \cref{SuperFluxDensitiesAsClosedlS4ValuedForms} is equivalent to all of:
    \begin{enumerate}
      \item the \emph{Einstein equation} for $e$ sourced by $G_4$,

      \item the \emph{Rarita-Schwinger equation} of motion for $\psi$

      \item the \emph{Maxwell-Chern-Simons equation} for $G_4$ (\cref{HigherGaugeFields})

      \emph{including} the constraint that $G_7 = \star_{{}_{11\vert 0}} G_4$.
    \end{enumerate}
  }
  \begin{equation}
    \label{SuperFluxDensitiesAsClosedlS4ValuedForms}
    \big(
      G_4^s, 
      G_7^s
    \big)
    \in
    \Omega^1_{\mathrm{dR}}\big(
      X;
      \mathfrak{l}S^4
    \big)_{\mathrm{cl}}
    \mathrlap{\,.}
  \end{equation}
\end{proposition}

But this reduces all of dynamical 11D SuGra to the (rational) homotopy theoretic form of its gauge sector established in \cref{GlobalCompletionOfHigherGaugeFields}, subject to the constraint \cref{The11DSuperFluxDensities}. In consequence, the flux-quantized phase space stack of the full gravitational theory, $\mathrm{PhsSpc}^{\mathrm{grv}}$, is obtained much like that of its gauge sector in \cref{OnTheCompletedPhaseSpace}, just replacing the sheaf of all closed $\mathfrak{l}S^4$-valued differential forms with the sheaf of super-frame fields and closed $\mathfrak{l}S^4$-valued super-forms subject to the constraint \cref{The11DSuperFluxDensities}. Otherwise the homotopy pullback construction \cref{GaugeFieldsAsLifts} of the phase space stack is the same. 
From there we may again form the shape of this phase space, $\shape \mathrm{PhsSpc}^{\mathrm{grv}}$, and thereby obtain the complete topological observables of the theory, just as in \cref{OnTopologicalObservables}.

Of course, the complication now is that,  due to the constraint \cref{The11DSuperFluxDensities}, 
Prop. \ref{ShapeOfAdiff} no longer applies (at least not directly), so that we lose the more explicit description of \cref{TheShapeOfPhaseSpace} and with it much of the computational control over these topological observables of 11D SuGra. This is both expected and desired: The observables in the full SuGra theory must be much richer and hence harder to analyze than in the pure gauge sector. But at least we thereby have their abstract description and may eventually find mathematical means to further unravel their structure.

(Analogous statements apply when including M5-brane probes into 11D SuGra: Their phase space is again given by promoting their worldvolume Bianchi identity \cref{M5WorldvolumeBianchi} from the plain flux density $H_3$ to its constrained super-space form, cf. \cite[Prop. 3.18]{GSS25-M5}).

\subsubsection{Effective Quantum Supergravity of FQH Excitations}
\footnote{
  More references on effective super-symmetry/gravity of the density wave modes in FQH systems are listed in \cite{SS26-SDiff,nLab:SuSyInFQHReferences}.
  The relevant supergravity is necessarily a non-relativistic (\emph{Newton-Cartan}) limit of ordinary SuGra (cf. \cite{BergshoeffRosseel2024}).
}

We have seen in \cref{OnFQAHMaterials,AnyonicCFieldFluxQuanta} that the topological gauge sector of (higher-dimensional) supergravity allows to ``geometrically engineer'' the effective topological physics of FQH anyons in a globally consistent and predictive way. This raises the question whether the remaining gravitational, fermionic and non-topological aspects of SuGra (as in \cref{QuantumSupergravity}) might naturally serve to engineer the finer non-topological physics of FQH systems, too. 

There is growing indication that just this is the case, namely the  \emph{density wave} excitations in the FQH electron gas, above the topological Laughlin ground state, have the following properties:
\begin{enumerate}
\item The long-wavelength limit of the bosonic mode looks like a (chiral) \emph{graviton}
\\
(theoretically: \cite{Liou2019}, and experimentally: \cite{Liang2024}).

\item The long-wavelength limit of the fermionic mode looks like a \emph{gravitino}
\\
(cf. \parencites{Haldane2013}[p. 8]{Wang2023}).

\item The bosonic/fermionic modes exhibit an emergent (broken) \emph{supersymmetry},
\\
(theoretically: \cite{GromovMartinecRyu2020,Liu2024Resolving}, and experimentally: \cite{Pu2023Signatures}),
\item 
whose key properties are usefully modeled by \emph{supergravity}
\\
(cf. \cite{NPBG2023,Du2025Chiral}).
\end{enumerate}

All this suggests that --- where we found (\cref{AnyonicCFieldFluxQuanta,OnFQAHMaterials}) the anyonic topological ground state of FQH systems faithfully embedded into the topological gauge sectors of higher-dimensional supergravity, after cohomotopical flux quantization --- we may expect the full (geometric, non-topological) completed quantum supergravity indicated in \cref{PhaseSpaceOfCompleted11DSugra} to know about the non-topological physics of FQH excitations. 

This remains to be seen. But something like this will be needed in order to proceed from the above understanding of the fundamental nature of FQH anyons to an understanding of their physical control:

\subsection{Topological Quantum Gates}
\label{TopoloticalQuantumGates}
\footnote{
  Exposition of topological quantum gates in our context is in \cite[\S 3]{MySS2024}. 
  The idea goes back to \parencites{Kitaev2003}{FreedmanKitaevLarsenWang2003}, traditional review includes \parencites{Nayak2008}{SatiValera2025}. The plausible \emph{necessity} of topological quantum protection for non-toy quantum computers has been highlighted (besides \cite[abst.]{FreedmanKitaevLarsenWang2003}) by \cite{DasSarma2022Hype}. And yet, the theoretical understanding of the required controllably movable defect anyons (cf. \cref{TopologicalQuantumGates}) has remained immature and may require a new theoretical approach that properly handles the traditionally elusive global topology of higher gauge fields.   
}
\newline \nopagebreak
The mere \emph{presence} of anyonic topological order is fascinating, but to live up to the promise of topological quantum hardware, the anyons (their positions) need to be \emph{controllable} (by adiabatic tuning of the material's external parameters), so as to operate \emph{topological quantum gates} on the material's ground state by controlled \emph{braiding} of anyon worldlines (cf. \cref{TopologicalQuantumGates}). 

\begin{figure}[htb]
\caption{%
  \label{TopologicalQuantumGates}%
  The original and still most substantial idea of topologically protected quantum operations assumes the \emph{controlled movement} of \emph{defect anyons}. This will require understanding the fine-grained response of these defect anyons to ``adiabatic'' external tuning of parameters of the host material and hence a refinement of their quantum theory beyond the topological sector. For geometric engineering of anyons in SuGra (as in \cref{AnyonicCFieldFluxQuanta}) this may mean entering the quantum gravity regime proper (as in \cref{QuantumSupergravity}). 
}
\centering
\adjustbox{
  scale=.8,
  rndfbox=4pt
}{
\begin{tikzpicture}[
  baseline=(current bounding box.center)
]

\clip
  (-8.4, 1.5) rectangle
  (4.9,-3.3);

  \shade[right color=lightgray, left color=white]
    (3,-3)
      --
    (-1,-1)
      --
    (-1.21,1)
      --
    (2.3,3);

  \draw[]
    (3,-3)
      --
    (-1,-1)
      --
    (-1.21,1)
      --
    (2.3,3)
      --
    (3,-3);

\draw[-Latex]
  ({-1 + (3+1)*.3},{-1+(-3+1)*.3})
    to
  ({-1 + (3+1)*.29},{-1+(-3+1)*.29});

\draw[-Latex]
    ({-1.21 + (2.3+1.21)*.3},{1+(3-1)*.3})
      --
    ({-1.21 + (2.3+1.21)*.29},{1+(3-1)*.29});

\draw[-Latex]
    ({2.3 + (3-2.3)*.5},{3+(-3-3)*.5})
      --
    ({2.3 + (3-2.3)*.49},{3+(-3-3)*.49});

\draw[-latex]
    ({-1 + (-1.21+1)*.53},{-1 + (1+1)*.53})
      --
    ({-1 + (-1.21+1)*.54},{-1 + (1+1)*.54});

  \begin{scope}[rotate=(+8)]
   \draw[dashed]
     (1.5,-1)
     ellipse
     ({.2*1.85} and {.37*1.85});
   \begin{scope}[
     shift={(1.5-.2,{-1+.37*1.85-.1})}
   ]
     \draw[->, -Latex]
       (0,0)
       to
       (180+37:0.01);
   \end{scope}
   \begin{scope}[
     shift={(1.5+.2,{-1-.37*1.85+.1})}
   ]
     \draw[->, -Latex]
       (0,0)
       to
       (+37:0.01);
   \end{scope}
   \begin{scope}[shift={(1.5,-1)}]
     \draw (.43,.65) node
     { \scalebox{.8}{$
     $} };
  \end{scope}
  \draw[fill=white, draw=gray]
    (1.5,-1)
    ellipse
    ({.2*.3} and {.37*.3});
  \draw[line width=3.5, white]
   (1.5,-1)
   to
   (-2.2,-1);
  \draw[line width=1.1]
   (1.5,-1)
   to node[
     above, 
     yshift=-4pt, 
     pos=.85]{
     \;\;\;\;\;\;\;\;\;\;\;\;\;
     \rotatebox[origin=c]{7}
     {
     \scalebox{.7}{
     \color{darkorange}
     \bf
     \colorbox{white}{anyonic defect}
     }
     }
   }
   (-2.2,-1);
  \draw[
    line width=1.1
  ]
   (1.5+1.2,-1)
   to
   (3.5,-1);
  \draw[
    line width=1.1,
    densely dashed
  ]
   (3.5,-1)
   to
   (4,-1);

  \draw[line width=3, white]
   (-2,-1.3)
   to
   (0,-1.3);
  \draw[-latex]
   (-2,-1.3)
   to
   node[
     below, 
     yshift=+3pt,
     xshift=-7pt
    ]{
     \scalebox{.7}{
       \rotatebox{+7}{
       \color{darkblue}
       \bf
       parameter
       }
     }
   }
   (0,-1.3);
  \draw[dashed]
   (-2.7,-1.3)
   to
   (-2,-1.3);

 \draw
   (-3.15,-.8)
   node{
     \scalebox{.7}{
       \rotatebox{+7}{
       \color{darkgreen}
       \bf
       braiding
       }
     }
   };

  \end{scope}

  \begin{scope}[shift={(-.2,1.4)}, scale=(.96)]
  \begin{scope}[rotate=(+8)]
  \draw[dashed]
    (1.5,-1)
    ellipse
    (.2 and .37);
  \draw[fill=white, draw=gray]
    (1.5,-1)
    ellipse
    ({.2*.3} and {.37*.3});
  \draw[line width=3.1, white]
   (1.5,-1)
   to
   (-2.3,-1);
  \draw[line width=1.1]
   (1.5,-1)
   to
   (-2.3,-1);
  \draw[line width=1.1]
   (1.5+1.35,-1)
   to
   (3.6,-1);
  \draw[
    line width=1.1,
    densely dashed
  ]
   (3.6,-1)
   to
   (4.1,-1);
  \end{scope}
  \end{scope}

  \begin{scope}[shift={(-1,.5)}, scale=(.7)]
  \begin{scope}[rotate=(+8)]
  \draw[dashed]
    (1.5,-1)
    ellipse
    (.2 and .32);
  \draw[fill=white, draw=gray]
    (1.5,-1)
    ellipse
    ({.2*.3} and {.32*.3});
  \draw[line width=3.1, white]
   (1.5,-1)
   to
   (-1.8,-1);
\draw
   (1.5,-1)
   to
   (-1.8,-1);
  \draw
    (5.23,-1)
    to
    (6.4-.6,-1);
  \draw[densely dashed]
    (6.4-.6,-1)
    to
    (6.4,-1);
  \end{scope}
  \end{scope}

\draw (1.73,-1.06) node
 {
  \scalebox{.8}{
    $k_{{}_{I}}$
  }
 };

\begin{scope}
[ shift={(-2,-.55)}, rotate=-82.2  ]

 \begin{scope}[shift={(0,-.15)}]

  \draw[]
    (-.2,.4)
    to
    (-.2,-2);

  \draw[
    white,
    line width=1.1+1.9
  ]
    (-.73,0)
    .. controls (-.73,-.5) and (+.73-.4,-.5) ..
    (+.73-.4,-1);
  \draw[
    line width=1.1
  ]
    (-.73+.01,0)
    .. controls (-.73+.01,-.5) and (+.73-.4,-.5) ..
    (+.73-.4,-1);

  \draw[
    white,
    line width=1.1+1.9
  ]
    (+.73-.1,0)
    .. controls (+.73,-.5) and (-.73+.4,-.5) ..
    (-.73+.4,-1);
  \draw[
    line width=1.1
  ]
    (+.73,0+.03)
    .. controls (+.73,-.5) and (-.73+.4,-.5) ..
    (-.73+.4,-1);

  \draw[
    line width=1.1+1.9,
    white
  ]
    (-.73+.4,-1)
    .. controls (-.73+.4,-1.5) and (+.73,-1.5) ..
    (+.73,-2);
  \draw[
    line width=1.1
  ]
    (-.73+.4,-1)
    .. controls (-.73+.4,-1.5) and (+.73,-1.5) ..
    (+.73,-2);

  \draw[
    white,
    line width=1.1+1.9
  ]
    (+.73-.4,-1)
    .. controls (+.73-.4,-1.5) and (-.73,-1.5) ..
    (-.73,-2);
  \draw[
    line width=1.1
  ]
    (+.73-.4,-1)
    .. controls (+.73-.4,-1.5) and (-.73,-1.5) ..
    (-.73,-2);

 \draw
   (-.2,-3.3)
   to
   (-.2,-2);
 \draw[
   line width=1.1,
   densely dashed
 ]
   (-.73,-2)
   to
   (-.73,-2.5);
 \draw[
   line width=1.1,
   densely dashed
 ]
   (+.73,-2)
   to
   (+.73,-2.5);

  \end{scope}
\end{scope}

\begin{scope}[shift={(-5.6,-.75)}]

  \draw[line width=3pt, white]
    (3,-3)
      --
    (-1,-1)
      --
    (-1.21,1)
      --
    (2.3,3)
      --
    (3, -3);

  \shade[right color=lightgray, left color=white, fill opacity=.7]
    (3,-3)
      --
    (-1,-1)
      --
    (-1.21,1)
      --
    (2.3,3);

  \draw[]
    (3,-3)
      --
    (-1,-1)
      --
    (-1.21,1)
      --
    (2.3,3)
      --
    (3, -3);

\draw (1.73,-1.06) node
 {
  \scalebox{.8}{
    $k_{{}_{I}}$
  }
 };

\draw[-Latex]
  ({-1 + (3+1)*.3},{-1+(-3+1)*.3})
    to
  ({-1 + (3+1)*.29},{-1+(-3+1)*.29});

\draw[-Latex]
    ({-1.21 + (2.3+1.21)*.3},{1+(3-1)*.3})
      --
    ({-1.21 + (2.3+1.21)*.29},{1+(3-1)*.29});

\draw[-Latex]
    ({2.3 + (3-2.3)*.5},{3+(-3-3)*.5})
      --
    ({2.3 + (3-2.3)*.49},{3+(-3-3)*.49});

\draw[-latex]
    ({-1 + (-1.21+1)*.53},{-1 + (1+1)*.53})
      --
    ({-1 + (-1.21+1)*.54},{-1 + (1+1)*.54});

  \begin{scope}[rotate=(+8)]
   \draw[dashed]
     (1.5,-1)
     ellipse
     ({.2*1.85} and {.37*1.85});
   \begin{scope}[
     shift={(1.5-.2,{-1+.37*1.85-.1})}
   ]
     \draw[->, -Latex]
       (0,0)
       to
       (180+37:0.01);
   \end{scope}
   \begin{scope}[
     shift={(1.5+.2,{-1-.37*1.85+.1})}
   ]
     \draw[->, -Latex]
       (0,0)
       to
       (+37:0.01);
   \end{scope}
  \draw[fill=white, draw=gray]
    (1.5,-1)
    ellipse
    ({.2*.3} and {.37*.3});
 \end{scope}

   \begin{scope}[shift={(-.2,1.4)}, scale=(.96)]
  \begin{scope}[rotate=(+8)]
  \draw[dashed]
    (1.5,-1)
    ellipse
    (.2 and .37);
  \draw[fill=white, draw=gray]
    (1.5,-1)
    ellipse
    ({.2*.3} and {.37*.3});
\end{scope}
\end{scope}

  \begin{scope}[shift={(-1,.5)}, scale=(.7)]
  \begin{scope}[rotate=(+8)]
  \draw[dashed]
    (1.5,-1)
    ellipse
    (.2 and .32);
  \draw[fill=white, draw=gray]
    (1.5,-1)
    ellipse
    ({.2*.3} and {.37*.3});
\end{scope}
\end{scope}

\begin{scope}
[ shift={(-2,-.55)}, rotate=-82.2  ]

 \begin{scope}[shift={(0,-.15)}]

 \draw[line width=3, white]
   (-.2,-.2)
   to
   (-.2,2.35);
 \draw
   (-.2,.5)
   to
   (-.2,2.35);
 \draw[dashed]
   (-.2,-.2)
   to
   (-.2,.5);

\end{scope}
\end{scope}

\begin{scope}
[ shift={(-2,-.55)}, rotate=-82.2  ]

 \begin{scope}[shift={(0,-.15)}]

 \draw[
   line width=3, white
 ]
   (-.73,-.5)
   to
   (-.73,3.65);
 \draw[
   line width=1.1
 ]
   (-.73,.2)
   to
   (-.73,3.65);
 \draw[
   line width=1.1,
   densely dashed
 ]
   (-.73,.2)
   to
   (-.73,-.5);
 \end{scope}
 \end{scope}

\begin{scope}
[ shift={(-2,-.55)}, rotate=-82.2  ]

 \begin{scope}[shift={(0,-.15)}]

 \draw[
   line width=3.2,
   white]
   (+.73,-.6)
   to
   (+.73,+3.7);
 \draw[
   line width=1.1,
   densely dashed]
   (+.73,-0)
   to
   (+.73,+-.6);
 \draw[
   line width=1.1 ]
   (+.73,-0)
   to
   (+.73,+3.71);
\end{scope}
\end{scope}

\end{scope}

\draw[
  draw=white,
  fill=white
]
  (-8,-2.4) rectangle
  (3,-3.8);

\draw[
  draw=white,
  fill=white
]
  (-1,-1.7) rectangle
  (3.4,-2.5);

\begin{scope}[
  shift={(0,1.3)}
]
\draw
  (-2.2,-4.2) node
  {
    \adjustbox{
      bgcolor=white,
      scale=1.2
    }{
      $
       \mathllap{
          \raisebox{1pt}{
            \scalebox{.58}{
              \color{darkblue}
              \bf
              \def\arraystretch{.9}
              \begin{tabular}{c}
                some quantum state for
                \\
                fixed defect positions
                \\
                $k_1, k_2, \cdots$
                at time
                {\color{purple}$t_1$}
              \end{tabular}
            }
          }
          \hspace{-5pt}
       }
        \big\vert
          \psi({\color{purple}t_1})
        \big\rangle
      $
    }
  };

\draw
  (+3.2,-3.85) node
  {
    \adjustbox{
      bgcolor=white,
      scale=1.2
    }{
      $
        \underset{
          \raisebox{-7pt}{
            \scalebox{.55}{
              \color{darkblue}
              \bf
              \def\arraystretch{.9}
               \begin{tabular}{c}
              another quantum state for
                \\
                fixed defect positions
                \\
                $k_1, k_2, \cdots$
                at time
                {\color{purple}$t_2$}
              \end{tabular}
            }
          }
        }{
        \big\vert
          \psi({\color{purple}t_2})
        \big\rangle
        }
      $
    }
  };

\draw[|->]
  (-1.3,-4.1)
  to
  node[
    sloped,
    yshift=5pt
  ]{
    \scalebox{.7}{
      \color{darkgreen}
      \bf
      unitary adiabatic transport
    }
  }
  node[
    sloped,
    yshift=-5pt,
    pos=.4
  ]{
    \scalebox{.7}{
    }
  }
  (+2.4,-3.4);
\end{scope}

\end{tikzpicture}
}
\end{figure}
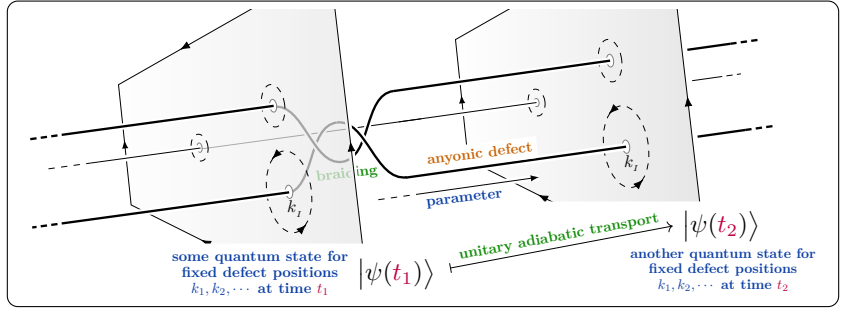

At the hardware level, this requires first understanding and then implementing non-topological interactions of anyons with perturbations in their host material. But if our derivation of anyonic order from the topological gauge sector of globally completed supergravity is more than a coincidence, then this understanding of anyon/substrate interaction will involve understanding completed quantum supergravity beyond the topological sector, along the lines indicated above in \cref{QuantumSupergravity}. It remains to see how this works out.


\clearpage

\noindent
\textbf{Acknowledgements.} 
We thank the organizers of \cite{ICMS:2025:GHS} for the opportunity to present this work, 
and we thank two anonymous referees for useful comments on an earlier version of this text.

\printbibliography

\end{document}